\documentclass[12pt,letterpaper]{article}
\usepackage[T1]{fontenc}
\usepackage[utf8]{inputenc}
\usepackage{newtxtext}
\usepackage{amsmath,amsthm,mathtools}
\usepackage{newtxmath,bm}
\usepackage[margin=1in]{geometry}
\usepackage{microtype}
\usepackage{xcolor,graphicx,booktabs,array}
\usepackage{adjustbox}
\usepackage{tikz}
\usepackage{pgfplots}
\pgfplotsset{compat=1.18}
\usetikzlibrary{arrows.meta,calc,positioning,fit,backgrounds,decorations.pathreplacing}
\usepackage[numbers,sort&compress]{natbib}
\usepackage[font=small,labelfont=bf]{caption}
\usepackage{placeins}
\usepackage{needspace}
\usepackage{enumitem}
\usepackage{hyperref}
\newcommand{\doi}[1]{\href{https://doi.org/#1}{\nolinkurl{doi:#1}}}
\hypersetup{colorlinks=true,linkcolor=blue!45!black,citecolor=blue!45!black,
 urlcolor=blue!45!black,pdftitle={Quantum simultaneous decoding via pretty-good measurements of tilted states},pdfauthor={Mark M. Wilde}}
\numberwithin{equation}{section}
\newtheorem{theorem}{Theorem}[section]
\newtheorem{lemma}[theorem]{Lemma}
\newtheorem{proposition}[theorem]{Proposition}
\newtheorem{corollary}[theorem]{Corollary}
\theoremstyle{definition}

\theoremstyle{remark}
\newtheorem{remark}[theorem]{Remark}
\DeclareMathOperator{\Tr}{Tr}
\DeclareMathOperator{\ran}{ran}
\DeclareMathOperator{\supp}{supp}
\DeclareMathOperator{\spanop}{span}

\DeclareMathOperator{\sgn}{sgn}
\newcommand{\cH}{\mathcal H}
\newcommand{\cX}{\mathcal X}
\newcommand{\cM}{\mathcal M}
\newcommand{\cC}{\mathcal C}
\newcommand{\cW}{\mathcal W}
\newcommand{\cN}{\mathcal N}

\newcommand{\ee}{\mathbb E}

\newcommand{\cc}{\mathbb C}
\newcommand{\ket}[1]{|#1\rangle}
\newcommand{\bra}[1]{\langle#1|}
\newcommand{\proj}[1]{\ket{#1}\!\bra{#1}}
\newcommand{\norm}[1]{\left\|#1\right\|}
\newcommand{\rhot}[1]{\widetilde{\rho}^{#1}_{BL}}
\newcommand{\rhoh}[1]{\widehat{\rho}^{#1}_{BL}}
\newcommand{\Pe}{P_{\mathrm e}}
\newcommand{\eps}{\varepsilon}

\newcommand{\sref}[1]{Section~\ref{#1}}
\newcommand{\fref}[1]{Figure~\ref{#1}}
\newcommand{\thref}[1]{Theorem~\ref{#1}}
\newcommand{\lemref}[1]{Lemma~\ref{#1}}
\newcommand{\corref}[1]{Corollary~\ref{#1}}
\newcommand{\prref}[1]{Proposition~\ref{#1}}
\definecolor{inkblue}{RGB}{36,69,101}
\definecolor{softblue}{RGB}{232,240,247}
\definecolor{inkgreen}{RGB}{28,103,87}
\definecolor{softgreen}{RGB}{233,244,239}
\definecolor{inkorange}{RGB}{160,83,32}
\definecolor{softorange}{RGB}{251,241,230}
\tikzset{>=Latex,operator/.style={draw=inkblue,rounded corners=2pt,fill=white,
 minimum height=9mm,align=center,inner sep=5pt},
 note/.style={font=\small,align=center},
 quantum/.style={line width=.8pt},classical/.style={line width=.8pt,dashed},
 flow/.style={->,line width=.85pt,draw=inkblue},
 panel/.style={draw=inkblue!45,rounded corners=3pt,fill=softblue,inner sep=7pt,align=center}}
\allowdisplaybreaks[1]\widowpenalty=10000\clubpenalty=10000
\title{\textbf{Quantum simultaneous decoding via\\[3pt] pretty-good measurements of tilted states}}
\author{Mark M. Wilde\thanks{Email: \href{mailto:wilde@cornell.edu}{wilde@cornell.edu}.}\\[3pt]
\small\textit{School of Electrical and Computer Engineering, Cornell University}\\
\small\textit{Ithaca, New York 14850, USA}}
\date{October 10, 2026}
\begin{document}
\maketitle
\begin{abstract}
A multiple-access channel transmits messages from several senders to one receiver. For quantum outputs, the receiver must distinguish states depending on all messages, while tests for different decoding errors may not commute. In this paper, we give a self-contained account of simultaneous decoding, in which one joint measurement recovers the messages at the established achievable rates using independent codebooks. The proof combines Sen's tilted-subspace method with Cheng's analysis of the pretty-good measurement, a standard measurement for distinguishing quantum states. Auxiliary labels close to a common reference state separate the geometric directions associated with different errors. A Gram-matrix bound then allows separately chosen binary tests to certify the performance of one measurement. The labels are removed from the completed decoder, leading to a deterministic packing theorem that applies to finite state ensembles, including unequal priors. We develop the two-sender case first, with worked examples and figures, before treating any fixed number of senders. Sharper finite-error guarantees for the same decoder are developed separately, with sharp bounds for two specified steps of the analysis rather than a claim of globally optimal decoding. Applications recover known achievable rates for unassisted and entanglement-assisted classical communication, distributed classical compression with quantum side information, and coherent-state communication through a beamsplitter or passive multiport device. The common proof makes the geometric mechanism and its quantitative consequences available without repeating the construction for each application.
\end{abstract}
\newpage
{\fontsize{9.5}{10}\selectfont\setlength{\parskip}{0pt}\let\qsdaddvspace\addvspace\renewcommand{\addvspace}[1]{\qsdaddvspace{2.5pt}}\tableofcontents}
\newpage
\section{Introduction}\label{sec:intro}

\subsection{Background and motivation}
A multiple-access channel describes communication from several senders to one receiver. Even when all inputs are classical, a quantum output changes the decoding problem: the receiver must distinguish quantum states indexed by a tuple of independently chosen messages. A useful decoder should recover the whole tuple while accounting for every possible subset of incorrect messages. The two-sender protocol is shown in \fref{fig:mac}.

The point-to-point coding theorems of Holevo and Schumacher--Westmoreland established the basic method for transmitting classical information through quantum channels~\cite{Holevo1998,SchumacherWestmoreland1997}. Winter determined the classical--quantum multiple-access capacity region~\cite{Winter2001}. Successive decoding, together with time sharing, proves achievability of that region. Simultaneous decoding asks for a more direct random-coding statement: independently generate one unsplit codebook for each sender and construct a joint decoder that achieves all the rate constraints at once. Here unsplit means that a sender's message is not divided into separately encoded submessages. This distinction matters in network problems, in which one cannot freely choose a separate rate splitting or time-sharing construction for each receiver.

Fawzi \textit{et al}.~proved a two-sender simultaneous-decoding theorem and formulated the three-sender conjecture needed in their analysis of quantum interference channels~\cite[Theorem~2 and Conjecture~4]{FawziEtAl2012}. Sen also developed a sequential-decoding approach to the two-sender problem~\cite{Sen2011}. Qi \textit{et al}.~subsequently isolated a connection with multiple quantum hypothesis testing and constructed a decoder for an arbitrary number of senders, although their particular choice of tests gave generally suboptimal asymptotic rates~\cite[Sections~4 and~5]{QiWangWilde2018}.

Sen ultimately resolved the general simultaneous-decoding problem through a quantum joint typicality lemma with intersection and union properties~\cite{Sen2021Joint,Sen2021Network}. His proof introduced tilted subspaces and smoothing and augmentation. Subsequent developments include multiparty packing and relay communication~\cite{DingEtAl2020}, computation over classical--quantum multiple-access channels~\cite{SohailEtAl2022}, structured coding for three-user broadcast and interference channels~\cite{GouiaaPadakandlaBroadcast2025,GouiaaPadakandla2025}, and multipartite covering and decoupling~\cite{Sen2025Covering}. The simultaneous-decoding conjecture is thus already resolved. These results motivate a detailed account of the mechanism in the multiple-access setting and of the quantitative estimates that it supports.

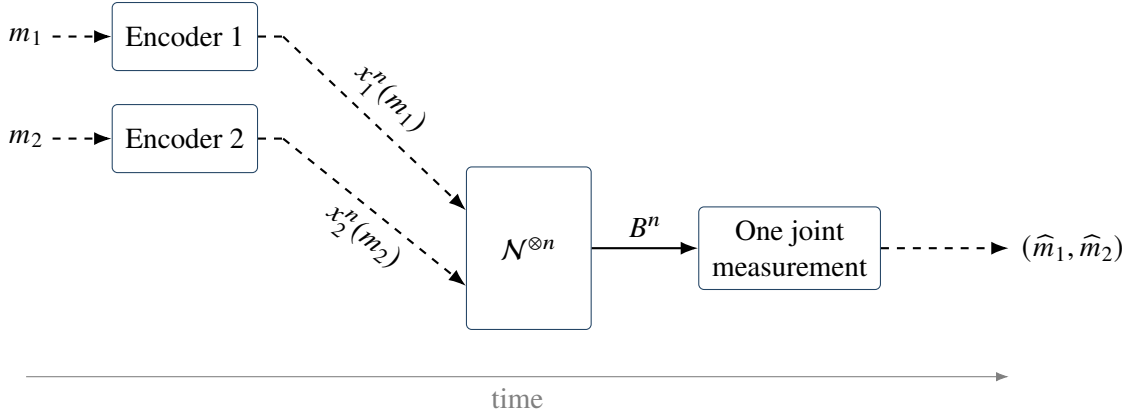
\begin{figure}[!t]
\centering
\begin{adjustbox}{max width=\linewidth}
\begin{tikzpicture}[font=\small,x=1cm,y=1cm]
\node at (0,2.8) {$m_1$};
\node[operator,minimum width=1.55cm] (E1) at (2.1,2.8) {Encoder 1};
\node at (0,1.45) {$m_2$};
\node[operator,minimum width=1.55cm] (E2) at (2.1,1.45) {Encoder 2};
\node[operator,minimum width=1.65cm,minimum height=2.15cm] (N) at (6.65,0) {$\cN^{\otimes n}$};
\node[operator,minimum width=1.85cm] (D) at (10.1,0) {One joint\\measurement};
\draw[classical,->] (.35,2.8)--(E1.west);
\draw[classical,->] (.35,1.45)--(E2.west);
\draw[classical] (E1.east)--(3.4,2.8);
\draw[classical,->] (3.4,2.8)--node[above,sloped,pos=.48] {$x_1^n(m_1)$} (5.82,.5);
\draw[classical] (E2.east)--(3.4,1.45);
\draw[classical,->] (3.4,1.45)--node[below,sloped,pos=.53] {$x_2^n(m_2)$} (5.82,-.5);
\draw[quantum,->] (N.east)--node[above] {$B^n$} (D.west);
\draw[classical,->] (D.east)--(13,0) node[right] {$(\widehat m_1,\widehat m_2)$};
\draw[->,gray] (0,-1.7)--(13,-1.7) node[midway,below] {time};
\end{tikzpicture}
\end{adjustbox}
\caption{Classical communication over a two-sender classical--quantum multiple-access channel. Each encoder uses an ordinary classical codebook. The receiver performs one measurement, indexed by message pairs, on the joint quantum output $B^n$. Dashed lines carry classical information and the solid line carries a quantum system. No auxiliary message label is transmitted in the protocol.}
\label{fig:mac}
\end{figure}

\subsection{Summary of contributions}
One of our purposes here is to explain the mechanism of Sen's simultaneous-decoding result~\cite{Sen2021Joint,Sen2021Network} through a self-contained treatment of multiple-access communication. We combine its tilted-subspace geometry with Cheng's state-based pretty-good-measurement analysis~\cite{Cheng2023}. 

We prove a deterministic packing bound for an arbitrary collection of states indexed by a Cartesian product of message sets. Each candidate message tuple is supplied with one binary test for every possible nonempty error pattern. These tests can be chosen separately and need not commute. The bound has a correct-message term, a sum of false-acceptance terms, and a small perturbation term. Its constants are independent of the output dimension and the message sizes.

The construction assigns each message index a weak auxiliary label. A candidate-dependent projection retains a sender's complete weak label when that sender's index agrees with the candidate; otherwise it retains only a reference component. Thus, every competing tuple is associated with an exact error-pattern vector. The linear independence of these vectors permits an elementary Gram-matrix estimate for the sum of the rejected subspaces. The auxiliary message-label viewpoint is already present in Sen's construction~\cite[Section~4]{Sen2021Joint}; our emphasis is on its explicit use in the fixed-ensemble packing calculation. The geometric projectors are used only as binary analysis tests. We define the actual extended-space decoder as the pretty-good measurement of the tagged output states and bound its error with Cheng's trace inequality~\cite[Appendix~A]{Cheng2023}. We then compress the completed measurement against a fixed auxiliary state. This compression gives a valid decoder on the original quantum output. The labels are mathematical devices, not a communication resource.

We first apply the fixed-ensemble bound to a classical-input quantum-output multiple-access channel with independent codebooks. For a fixed product input distribution, the packing bound and ordinary binary quantum Stein tests recover the established simultaneous-decoding constraints~\cite{Sen2021Joint,Sen2021Network}
\begin{equation}\label{eq:intro-region}
 \sum_{i\in S}R_i\le I(X_S;B|X_{S^c})_\omega,
 \qquad \varnothing\ne S\subseteq[k],
\end{equation}
with boundary points understood as limits of achievable rates. Here $k$ is an arbitrary fixed number of senders, $X_S$ denotes the classical inputs of the senders in $S$, and $\omega$ is the input--output state induced by the product distribution. The precise random-coding statements appear in \thref{thm:two-capacity} and \thref{thm:k-capacity}.

Our main contribution is this elementary packing-level presentation. Its explicit constants also avoid the doubly exponential dependence on the number of classical registers in the displayed general joint typicality estimate of Sen~\cite[Section~3, Fact~4]{Sen2021Network}. At fixed label strength our correct-rejection coefficient is exponential in $k$; after holding the label-removal contribution below a fixed target, the resulting coefficient is $2^{O(k\log k)}$. Section~\ref{subsec:sen-comparison} gives the matched-parameter comparison and explains why it does not imply a uniform improvement at every error tolerance or an automatic improvement to every use of the general lemma.

We also obtain sharper finite-error bounds for the same decoder. To preserve the exposition, the basic proof is unchanged in structure: a compact refinement is stated in \sref{subsubsec:sharp-packing}, with its full derivation in Appendix~\ref{app:quantitative}. The refined Gram coefficient is optimal for the specified correct-rejection inequality and label family, and the label-transfer function is optimal given only the label overlap and a bound on tagged error. Neither statement asserts global optimality of the decoding error bound. The geometric mechanism is a special instance of Sen's tilting. We do not claim an independent origin for that mechanism or a new capacity formula. The useful specialization is to employ deterministic, injective message labels and to bound a fixed pair of competing messages before averaging over the codebook. This order of analysis makes separate smoothing estimates unnecessary for the problem considered here.

The same deterministic bound gives two further reductions. Position-based coding, introduced by Anshu, Jain, and Warsi~\cite{AnshuJainWarsi2019} and applied to multiple-access communication by Qi \textit{et al}.~\cite{QiWangWilde2018}, converts entanglement-assisted communication into discrimination of reference positions. This recovers the established achievable rate regions for entanglement-assisted classical communication over a multiple access channel~\cite{HsiehDevetakWinter2008,XuWilde2013,ShiEtAl2021}.

For classical compression with quantum side information, Devetak and Winter~\cite{DevetakWinter2003} treated one classical sender and one receiver holding quantum side information, and Renes and Renner~\cite{RenesRenner2012} developed the corresponding one-shot setting. Here we consider several senders holding correlated classical data. A nonuniform-prior version of the packing bound and independent two-universal hashing give a simultaneous decoder, recovering the distributed rate constraints studied in~\cite{Winter1999Thesis,LiuXiFan2023}. Salek and Winter already gave a simultaneous-decoding proof using independent random binning and Sen's construction~\cite[Appendix~A]{SalekWinter2022}; here the task is a direct application of the same fixed-ensemble packing theorem. Both applications follow from explicit one-shot formulas, without a second geometric proof.

As a physical application, we recover the coherent-state achievable region of a two-sender multiple access channel based on a beamsplitter~\cite{YenShapiro2005} and extend the same derivation to a $k$-port splitter with $k$ senders. The proof first uses finite optical constellations, whose received states span a finite-dimensional space, and then establishes the Gaussian limit under a uniform-message average photon-number constraint.

\Needspace{9\baselineskip}
\subsection{Organization of the paper}
We begin our development by presenting the case of two senders first---this case already contains three potentially noncommuting tests and every ingredient of the proof, and it is helpful for understanding later developments. After that, we extend the whole development to $k$ senders. Figure~\ref{fig:roadmap} provides a preview of the construction and analysis of the measurement, summarizing how the construction of the measurement is separated from its analysis.

\begin{figure}[!t]
\centering
\begin{adjustbox}{max width=\linewidth}
\begin{tikzpicture}[font=\footnotesize]
\node[panel,text width=2.7cm,minimum height=1.2cm] (states) at (0,1.8) {Tagged states\\$\rho_m\otimes\proj{t_m}$};
\node[panel,text width=3.0cm,minimum height=1.2cm] (pgm) at (5,1.8) {Pretty-good measurement\\of the tagged states};
\node[panel,text width=3.0cm,minimum height=1.2cm] (actual) at (10,1.8) {Actual decoder\\$\Lambda_m=J^\dagger\widetilde\Lambda_mJ$};
\node[panel,fill=softgreen,text width=2.7cm,minimum height=1.2cm] (tests) at (0,-.8) {Separate tests\\$P_{1,m},P_{2,m},P_{12,m}$};
\node[panel,fill=softgreen,text width=3.0cm,minimum height=1.2cm] (gamma) at (5,-.8) {Rejection span\\and tests $\Gamma_m$};
\node[align=left,text width=3cm] at (10,-.8) {Gram matrix:\\correct rejection\\Exact labels:\\false acceptance};
\draw[flow] (states)--(pgm);
\draw[flow] (pgm)--node[above]{compress}(actual);
\draw[flow] (tests)--(gamma);
\draw[->,densely dashed,inkorange,thick] (gamma.north)--node[right]{Cheng's error bound}(pgm.south);
\end{tikzpicture}
\end{adjustbox}
\caption{Construction and analysis have separate roles. The upper row defines the measurement from the tagged states and then removes the labels. The lower row constructs binary tests that certify its performance; the measurement itself is constructed from the states, not from these tests. The dashed arrow denotes an inequality, not a physical operation. Here $J\ket h=\ket h\otimes\ket{00}$.}
\label{fig:roadmap}
\end{figure}
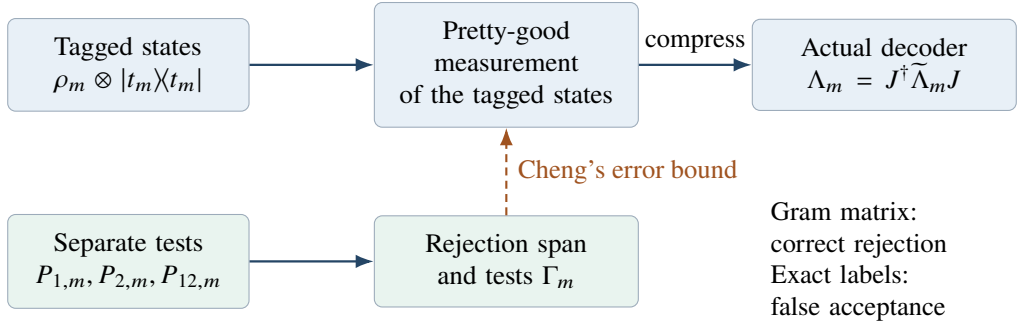

\sref{sec:notation} recalls the binary-testing and PGM tools, while including posterior-sampling intuition and an alternative proof of Cheng's trace inequality. Sections~\ref{sec:two-setup}--\ref{sec:two-coding} state and prove the two-sender theorem. \sref{sec:k} gives the complete extension to a fixed number of senders and collects the general-effect, hypothesis-testing, and nonuniform-prior forms of the packing bound. Sections~\ref{sec:ea} and~\ref{sec:compression} apply it to entanglement-assisted communication and distributed compression, respectively. \sref{sec:bosonic} treats coherent-state optical channels. \sref{sec:relation} compares the construction with prior work and states its scope, and \sref{sec:conclusion} concludes with a summary and future directions. Appendix~\ref{app:quantitative} proves quantitative refinements of the main error bound, including their extension to the application-specific one-shot bounds.

\section{Notation and preliminary facts}\label{sec:notation}
The construction outlined in the introduction has two parts: geometric binary tests certify distinguishability, and one measurement decodes the states. Before developing the geometry, we explain how binary-testing estimates bound that measurement's error. We first fix the operator and information-theoretic conventions, then give the posterior-sampling intuition, prove Cheng's trace inequality, and derive the pretty-good-measurement (PGM) bound used throughout.
\subsection{States, operators, and measurements}
Except in the bosonic application of \sref{sec:bosonic}, we take all Hilbert spaces to be finite dimensional and classical alphabets to be finite. We write $[M]\coloneqq\{1,\ldots,M\}$ for a positive integer $M$. The identity operator is denoted by $I$, with a system subscript when needed. For Hermitian operators $A$ and $B$ on the same space, $A\ge B$ means that $A-B$ is positive semidefinite; equivalently, $\langle h|(A-B)|h\rangle\ge0$ for every vector $\ket h$. We use $A\le B$ to mean $B\ge A$. 

A state is a positive semidefinite operator of trace one. We use system labels as subscripts, and omit them for single-system expressions when the space has been specified. The range and support of an operator are denoted by $\ran$ and $\supp$. The projector onto a subspace $\cW$ is denoted by $\Pi_{\cW}$. For a set of vectors $\mathcal V$, $\spanop\mathcal V$ denotes the set of its finite complex linear combinations. A sum of subspaces denotes their linear span, not their set-theoretic union. In the finite-dimensional case, $\supp H=(\ker H)^\perp=\ran H$ for a positive semidefinite operator~$H$.

An effect is an operator $T$ satisfying $0\le T\le I$. A positive operator-valued measure (POVM) is a tuple $(\Lambda_z)_z$ of positive semidefinite operators satisfying $\sum_z\Lambda_z=I$. The probability of outcome $z$ on a state $\rho$ is given by $\Tr[\Lambda_z\rho]$. Throughout, a \emph{projector} means a self-adjoint idempotent operator $P$ satisfying $P=P^\dagger=P^2$.

For an operator $Z$, the trace norm and operator norm are defined as
\begin{equation}
 \norm{Z}_1\coloneqq\Tr\sqrt{Z^\dagger Z},
 \qquad
 \norm{Z}\coloneqq\sup_{\ket h:\,\langle h|h\rangle=1}\norm{Z\ket h}.
\end{equation}
For a Hilbert-space vector, we write $\norm{\ket h}\coloneqq\sqrt{\langle h|h\rangle}$ to denote the standard Hilbert-space norm. Thus the argument distinguishes the vector norm from the operator norm; the trace norm retains its subscript~$1$. The normalized trace distance is defined as
\begin{equation}\label{eq:trace-distance}
 \Delta(\rho,\sigma)\coloneqq\tfrac12\norm{\rho-\sigma}_1.
\end{equation}
Its measurement interpretation gives, for every effect $T$,
\begin{equation}\label{eq:effect-distance}
 \left|\Tr[T\rho]-\Tr[T\sigma]\right|\le \Delta(\rho,\sigma).
\end{equation}
Indeed, write $Z\coloneqq \rho-\sigma=Z_+-Z_-$ for its positive and negative parts. Since $\Tr Z=0$, both $\Tr Z_+$ and $\Tr Z_-$ equal $\norm Z_1/2$. The inequalities $0\le T\le I$ bound $\Tr[TZ]$ between $-\Tr Z_-$ and $\Tr Z_+$. We will also use
\begin{equation}\label{eq:pure-distance}
 \Delta(\proj u,\proj v)=\sqrt{1-|\langle u|v\rangle|^2},
 \qquad
 \norm{\rho\otimes Z}_1=\norm Z_1
\end{equation}
for normalized vectors $\ket u,\ket v$ and a state $\rho$. The first identity follows by diagonalizing $\proj u-\proj v$ on the span of $\ket u,\ket v$; its nonzero eigenvalues are the two numbers $\pm\sqrt{1-|\langle u|v\rangle|^2}$. The second follows by diagonalizing $\rho$ and summing its eigenvalues.

For a Hermitian matrix $D$ with spectral decomposition $D=\sum_j d_j\proj{e_j}$, its sign is the Hermitian matrix
\begin{equation}\label{eq:matrix-sign}
 \sgn(D)\coloneqq\sum_j\sgn(d_j)\proj{e_j},\qquad
 \sgn(x)\coloneqq\begin{cases}1,&x>0,\\0,&x=0,\\-1,&x<0.\end{cases}
\end{equation}
In particular, $\sgn(D)D=|D|$, where $|D|\coloneqq\sqrt{D^2}$, and $\sgn(D)^2=I-\Pi_{\ker D}$. The value at zero is specified because our alternative proof of Cheng's trace inequality will add the kernel projector to obtain a self-adjoint unitary.

\subsection{Entropy and binary hypothesis testing}
Throughout, we take all logarithms to base two. Given a state $\rho$ and a positive semidefinite operator~$\sigma$, the von Neumann entropy and Umegaki quantum relative entropy~\cite{Umegaki1962} are defined as
\begin{align}
H(\rho)&\coloneqq-\Tr[\rho\log_2\rho],\label{eq:line-2-6-1}\\
D(\rho\Vert\sigma)&\coloneqq\Tr[\rho(\log_2\rho-\log_2\sigma)],
 \quad\text{if }\supp\rho\subseteq\supp\sigma,\label{eq:line-2-6-2}
\end{align}
and $D(\rho\Vert\sigma)\coloneqq+\infty$ otherwise. The logarithms are evaluated on their supports. For a multipartite state $\rho_{ABC}$, we write $H(A)_\rho\coloneqq H(\rho_A)$ and define
\begin{align}
H(A|C)_\rho&\coloneqq H(AC)_\rho-H(C)_\rho,\label{eq:line-2-7-1}\\
I(A;B)_\rho&\coloneqq H(A)_\rho+H(B)_\rho-H(AB)_\rho,\label{eq:line-2-7-2}\\
I(A;B|C)_\rho&\coloneqq H(AC)_\rho+H(BC)_\rho-H(C)_\rho-H(ABC)_\rho.\label{eq:line-2-7-3}
\end{align}
These definitions apply equally to quantum systems and classical registers represented in an orthonormal basis. We refer to~\cite{Wilde2017} for background on these information quantities.

In binary hypothesis testing, an effect $T$ accepts the null hypothesis $\rho$ and rejects the alternative~$\sigma$. Its two error probabilities are defined as
\begin{equation}\label{eq:binary-errors}
 \alpha(T;\rho)\coloneqq\Tr[(I-T)\rho],
 \qquad
 \beta(T;\sigma)\coloneqq\Tr[T\sigma].
\end{equation}
We use the direct part of the quantum Stein lemma~\cite{HiaiPetz1991,OgawaNagaoka2000}: whenever $r<D(\rho\Vert\sigma)$, there is a sequence of effects $U_n$ such that
\begin{equation}\label{eq:stein-standard}
 \alpha(U_n;\rho^{\otimes n})\xrightarrow[n\to\infty]{} 0,
 \qquad
 \beta(U_n;\sigma^{\otimes n})\le2^{-nr}.
\end{equation}
In \lemref{lem:projective-stein} we recall the derivation of the projective threshold tests needed in the main proof from this standard statement.

For states $\rho,\sigma$ and $0\le\eps<1$, define the hypothesis testing relative entropy as
\begin{equation}\label{eq:DH}
 D_H^\eps(\rho\Vert\sigma)
 \coloneqq-\log_2\min\left\{
 \Tr[E\sigma]:0\le E\le I,\ \Tr[E\rho]\ge1-\eps
 \right\}.
\end{equation}
This quantity appears in various one-shot information-processing tasks~\cite{BuscemiDatta2010,BrandaoDatta2011,WangRenner2012}, with closely related smoothing conventions. The minimum exists in the finite-dimensional case by compactness. We use the convention $2^{-\infty}=0$.

\subsection{Posterior sampling and binary testing costs}\label{subsec:binary-intuition}
Before stating Cheng's trace inequality, it is useful to understand its classical meaning. Suppose two hypotheses give prior-weighted probabilities $a_y$ and $b_y$ for an observation $y$. When $a_y+b_y>0$, the posterior probabilities are $a_y/(a_y+b_y)$ and $b_y/(a_y+b_y)$. A posterior-sampling rule draws its answer from these probabilities, rather than always choosing the larger one. This is the classical form of the pretty-good measurement~\cite{Renes2017}.

The weighted probability that the first hypothesis is true but the second is selected is
\begin{equation}\label{eq:posterior-directional}
 q_{\mathrm{cl}}\coloneqq\sum_{y:a_y+b_y>0}\frac{a_yb_y}{a_y+b_y}.
\end{equation}
The opposite direction gives the same expression. In contrast, the optimal binary rule chooses the larger posterior and has total error $\sum_y\min\{a_y,b_y\}$. For every pair of nonnegative numbers,
\begin{equation}\label{eq:scalar-parallel-bound}
 \frac{ab}{a+b}\le\min\{a,b\},
\end{equation}
with the fraction defined to be zero when $a=b=0$. Indeed, the fraction is at most $a$ because $b\le a+b$, and at most $b$ because $a\le a+b$. Thus \emph{one directional error} of posterior sampling is no larger than the optimal \emph{total} binary error. The total posterior-sampling error is twice that directional contribution: the two directional errors are equal, and their sum is $2q_{\mathrm{cl}}$, not $q_{\mathrm{cl}}$. Thus the comparison does not assert that posterior sampling improves on the optimal decision rule. Figure~\ref{fig:posterior-sampling} illustrates this distinction.

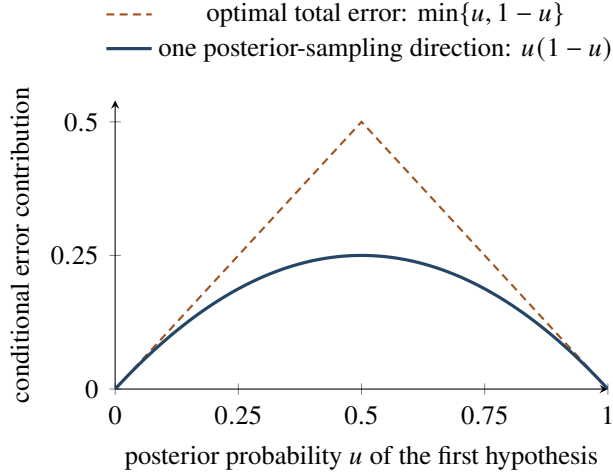
\begin{figure}[!t]
\centering
\begin{adjustbox}{max width=\linewidth}
\begin{tikzpicture}[font=\small]
\begin{axis}[
 width=8.1cm,height=5.4cm,xmin=0,xmax=1,ymin=0,ymax=.54,
 axis lines=left,xlabel={posterior probability $u$ of the first hypothesis},
 ylabel={conditional error contribution},
 xtick={0,.25,.5,.75,1},ytick={0,.25,.5},
 tick label style={font=\footnotesize},label style={font=\footnotesize},
 legend style={draw=none,font=\footnotesize,at={(.5,1.10)},anchor=south},
 domain=0:1,samples=121]
 \addplot[inkorange,thick,densely dashed] {min(x,1-x)};
 \addlegendentry{optimal total error: $\min\{u,1-u\}$}
 \addplot[inkblue,very thick] {x*(1-x)};
 \addlegendentry{one posterior-sampling direction: $u(1-u)$}
\end{axis}
\node[align=left,anchor=west,text width=5.0cm] at (8.45,3.45)
 {The inequality $u(1-u)\le\min\{u,1-u\}$ is the classical comparison behind Cheng's trace bound.};
\node[align=left,anchor=west,text width=5.0cm] at (8.45,1.30)
 {Cheng's trace inequality retains the one-directional comparison when the signal operators do not commute.};
\end{tikzpicture}
\end{adjustbox}
\caption{The scalar comparison behind the trace inequality. For an observation of total weight $a+b$, set $u\coloneqq a/(a+b)$. Dividing its contribution $ab/(a+b)$ by $a+b$ gives $u(1-u)$. The plot compares this single directional contribution with the optimal total binary error.}
\label{fig:posterior-sampling}
\end{figure}

We now pass from classical hypotheses to quantum hypotheses by replacing the prior-weighted likelihoods with positive semidefinite operators. If $A\coloneqq\sum_y a_y\proj y$ and $B\coloneqq\sum_y b_y\proj y$ are diagonal in a common basis, the quantum testing problem reduces to the classical decision problem above. Without the commutation assumption, no common observation basis is available. Cheng's trace inequality will nevertheless retain the comparison between one directional PGM error and the optimal total binary error. To formulate that comparison, for positive semidefinite operators $A$ and $B$ on the same finite-dimensional space, define the optimal binary testing cost as
\begin{equation}\label{eq:binary-cost}
 \mathsf e(A,B)\coloneqq\min\left\{
 \Tr[(I-T)A]+\Tr[TB]:0\le T\le I\right\}.
\end{equation}
For the diagonal operators above, this cost equals $\sum_y\min\{a_y,b_y\}$, recovering the optimal classical binary error. When $\Tr(A+B)=1$, the traces of $A$ and $B$ are the prior probabilities, and this cost is an error probability. We also allow arbitrary positive semidefinite operators, because in a decoding problem $B$ will be a sum of competing states. The Holevo--Helstrom formula~\cite{Helstrom1967,Holevo1973} gives
\begin{equation}\label{eq:helstrom-cost}
 \mathsf e(A,B)=\frac{\Tr(A+B)-\norm{A-B}_1}{2}.
\end{equation}
For completeness, set $D\coloneqq A-B$ and write the objective as $\Tr A-\Tr[TD]$. In an eigenbasis of $D$, the diagonal entries of $T$ lie in $[0,1]$. Their maximizing choice accepts the positive eigenvalues and rejects the negative ones, so the maximum is $\Tr D_+$. The identity $2\Tr D_+=\norm D_1+\Tr D$ gives \eqref{eq:helstrom-cost}. This also proves that an optimizing projector exists. The proof below relates the same trace-norm distinguishability to the overlap appearing in the quantum pretty-good measurement.

\subsection{Cheng's trace inequality and its binary meaning}\label{subsec:cheng}
Cheng's analysis replaces the scalar fraction in \eqref{eq:scalar-parallel-bound} by a trace involving positive semidefinite operators that need not commute~\cite[Appendix~A]{Cheng2023}. We give an alternative proof below.

\begin{lemma}[Cheng's trace inequality]\label{lem:cheng}
Let $A,B\ge0$ act on the same finite-dimensional Hilbert space, set $Z\coloneqq A+B$, and define $e\coloneqq \mathsf e(A,B)$ by \eqref{eq:binary-cost}. Inverses vanish on $\ker Z$. If $s\coloneqq \Tr Z>0$, then
\begin{equation}\label{eq:cheng-trace}
 \Tr[AZ^{-1/2}BZ^{-1/2}]
 \le e\left(1-\frac es\right)\le e.
\end{equation}
If $Z=0$, then the desired inequality trivially holds because the trace and the testing cost are both equal to zero. In either case, every effect $T$ satisfies
\begin{equation}\label{eq:cheng-test}
 \Tr[AZ^{-1/2}BZ^{-1/2}]
 \le\Tr[(I-T)A]+\Tr[TB].
\end{equation}
\end{lemma}
\begin{proof}
Define $q\coloneqq\Tr[AZ^{-1/2}BZ^{-1/2}]$.
Assume $s>0$ and work on $\supp Z$, where $Z$ is invertible. This suffices, given that positivity and $A,B\le Z$ imply that both operators vanish on $\ker Z$. Set
\begin{equation}\label{eq:cheng-difference}
 D\coloneqq A-B,\qquad X\coloneqq Z^{-1/4}DZ^{-1/4}.
\end{equation}
The matrix $X$ is Hermitian.  Expanding $D=A-B$ and then expanding $Z=A+B$ therefore gives
\begin{align}
\Tr[X^2]
 &=\Tr[DZ^{-1/2}DZ^{-1/2}]\label{eq:cheng-expansion}\\
&=\Tr[AZ^{-1/2}AZ^{-1/2}]+\Tr[BZ^{-1/2}BZ^{-1/2}]-2q\label{eq:line-2-18-2}\\
&=\Tr[ZZ^{-1/2}ZZ^{-1/2}]-4q\label{eq:line-2-18-3}\\
&=s-4q.\label{eq:line-2-18-4}
\end{align}
The second and third equalities follow because the two cross terms in the expansion have equal traces by cyclicity.
Hence an upper bound on $q$ follows from a lower bound on $\Tr[X^2]$. The desired lower bound, as an intermediate step, is
\begin{equation}\label{eq:cheng-goal}
 \Tr[X^2]\ge\frac{\norm D_1^2}{s}.
\end{equation}
We now prove it. With the sign convention in \eqref{eq:matrix-sign}, define
\begin{equation}\label{eq:cheng-sign-unitary}
 W\coloneqq\sgn(D)+\Pi_{\ker D},\qquad Y\coloneqq Z^{1/4}WZ^{1/4}.
\end{equation}
Here the kernel is taken within $\supp Z$. On a positive, negative, or zero eigenspace of $D$, the operator $W$ has eigenvalue $+1$, $-1$, or $+1$, respectively. Thus
\begin{equation}\label{eq:cheng-sign-identities}
 W=W^\dagger,\qquad W^2=I,\qquad WD=|D|.
\end{equation}
The extra kernel projector makes $W$ unitary without changing $WD$. The placement of the powers of $Z$ in $X$ and $Y$ is chosen so that they cancel inside the following trace:
\begin{equation}\label{eq:cheng-pairing}
 \Tr[XY]=\Tr[Z^{-1/4}DWZ^{1/4}]=\Tr[DW]=\norm D_1.
\end{equation}
For matrices $R$ and $S$, the Cauchy--Schwarz inequality is as follows:
\begin{equation}\label{eq:trace-CS}
 |\Tr[R^\dagger S]|^2\le\Tr[R^\dagger R]\Tr[S^\dagger S].
\end{equation}
Applying it to the Hermitian matrices $X$ and $Y$ yields
\begin{equation}\label{eq:cheng-CS}
 \norm D_1^2\le\Tr[X^2]\Tr[Y^2].
\end{equation}
The remaining factor $\Tr[Y^2]$ can be bounded using the Cauchy--Schwarz inequality again, now with $R=Z^{1/2}$ and $S=WZ^{1/2}W$:
\begin{align}
\Tr[Y^2]
 &=\Tr[Z^{1/2}WZ^{1/2}W]\label{eq:cheng-Y-bound}\\
&\le\sqrt{\Tr Z\;\Tr[(WZ^{1/2}W)^2]}\label{eq:line-2-25-2}\\
&=\sqrt{\Tr Z\;\Tr[WZW]}\label{eq:line-2-25-3}\\
&=\Tr Z=s.\label{eq:line-2-25-4}
\end{align}
We used $W^2=I$ in the third line and cyclicity in the fourth. Combining \eqref{eq:cheng-CS} and \eqref{eq:cheng-Y-bound}--\eqref{eq:line-2-25-4} proves \eqref{eq:cheng-goal}. Substitution in \eqref{eq:cheng-expansion}--\eqref{eq:line-2-18-4}, followed by \eqref{eq:helstrom-cost}, now completes the argument:
\begin{equation}\label{eq:cheng-stronger}
 q\le\frac{s^2-\norm D_1^2}{4s}
 =\frac{(s-\norm D_1)(s+\norm D_1)}{4s}
 =e\left(1-\frac es\right)\le e.
\end{equation}
Indeed, $0\le e\le s/2$: the testing choices $T=0$ and $T=I$ give $e\le\min\{\Tr A,\Tr B\}\le s/2$. Evaluating \eqref{eq:binary-cost} at the specified effect $T$ proves \eqref{eq:cheng-test}. The case $s=0$ is immediate from positivity.
\end{proof}

\subsubsection{Why this is a one-directional error bound}
The trace inequality is now established, but its operational meaning depends on which error is being counted. We return to the two-hypothesis measurement to identify that error and to relate the bound to Renes's earlier result in~\cite[Eq.~(12)]{Renes2017}.

For this operational interpretation, assume $s=\Tr(A+B)=1$. Thus $A$ and $B$ already include prior probabilities whose sum is one, and the unnormalized cost $e=\mathsf e(A,B)$ is now an error probability. The binary PGM has effects $Z^{-1/2}AZ^{-1/2}$ and $Z^{-1/2}BZ^{-1/2}$ on $\supp Z$. The first hypothesis being true and the second being selected contributes $q$ to its error. The other contribution is
\begin{equation}\label{eq:equal-directional-errors}
 \Tr[BZ^{-1/2}AZ^{-1/2}]
 =\Tr[AZ^{-1/2}BZ^{-1/2}]=q.
\end{equation}
Its total error is $2q$, not $q$. Substituting $s=1$ into \eqref{eq:cheng-stronger} gives $q\le e(1-e)$, and therefore
\begin{equation}\label{eq:renes-binary}
 P_{\mathrm{succ}}^{\mathrm{PGM}}=1-2q
 \ge1-2e(1-e)=(1-e)^2+e^2.
\end{equation}
The parameter $s$ is absent from \eqref{eq:renes-binary} only because of this normalization. For arbitrary $s>0$, replacing $A$ and $B$ by $A/s$ and $B/s$ leaves the PGM effects unchanged and replaces $q$ and $e$ by $q/s$ and $e/s$, respectively. The corresponding statement is $1-2q/s\ge(1-e/s)^2+(e/s)^2$.

Equation~\eqref{eq:renes-binary} is the binary specialization of Renes's near-optimality bound~\cite[Eq.~(12)]{Renes2017}. Conversely, the normalization just described recovers the stronger inequality in \lemref{lem:cheng}. Cheng explicitly builds on Renes's quantum bound and the classical bounds of Sason and Verd{\'u}~\cite{SasonVerdu2018,Cheng2023}; related near-optimality results originate in Barnum and Knill~\cite{BarnumKnill2002}. Thus the result does not assert that the PGM outperforms optimal binary discrimination. It bounds one directional contribution by the optimal total cost, just as in the scalar example mentioned in Section~\ref{subsec:binary-intuition}.

\subsection{The state-based pretty-good measurement and a union-type bound}\label{subsec:state-pgm}
We next pass from two hypotheses to a whole ensemble. This step will explain why the trace inequality is useful for decoding. Let $\tau_1,\ldots,\tau_M$ be states, and let $p_m\ge0$ be probabilities satisfying $\sum_m p_m=1$. Define
\begin{equation}\label{eq:weighted-pgm-definition}
 A_m\coloneqq p_m\tau_m,\quad Z\coloneqq\sum_m A_m,\quad
 \Lambda_m\coloneqq Z^{-1/2}A_mZ^{-1/2},\quad
 \Lambda_{\mathrm f}\coloneqq I-\Pi_Z.
\end{equation}
All effects are positive semidefinite, and their sum is
\begin{equation}\label{eq:state-pgm-complete}
 \sum_m\Lambda_m+\Lambda_{\mathrm f}
 =Z^{-1/2}ZZ^{-1/2}+I-\Pi_Z=I.
\end{equation}
The failure effect has zero probability on every $\tau_m$ with $p_m>0$. For a uniform prior, the common factor $1/M$ cancels in the PGM formula. The following proposition turns individually chosen binary-testing estimates into an error bound for this one state-based measurement; the tests need not themselves form a POVM.

\begin{proposition}[Binary tests bound the error of the state-based PGM]\label{prop:pgm-tests}
Let $(p_m,\tau_m)_{m=1}^M$ be a finite ensemble on a finite-dimensional space and let $0\le T_m\le I$ be arbitrary effects. The measurement in \eqref{eq:weighted-pgm-definition}, including its failure effect, has its error probability bounded as
\begin{equation}\label{eq:weighted-pgm-error}
 P_{\mathrm e}^{\mathrm{PGM}}\coloneqq1-\sum_m p_m\Tr[\Lambda_m\tau_m]
 \le\sum_m p_m\Tr[(I-T_m)\tau_m]
 +\sum_m p_m\sum_{\ell\ne m}\Tr[T_\ell\tau_m].
\end{equation}
In particular, for the uniform prior $p_m=1/M$,
\begin{equation}\label{eq:pgm-error}
 1-\frac1M\sum_m\Tr[\Lambda_m\tau_m]
 \le\frac1M\sum_m\Tr[(I-T_m)\tau_m]
 +\frac1M\sum_m\sum_{\ell\ne m}\Tr[T_\ell\tau_m].
\end{equation}
The tests $T_m$ need not commute or sum to the identity.
\end{proposition}
\begin{proof}
Use $A_m$ and $Z$ from \eqref{eq:weighted-pgm-definition}. Fix $m$ with $p_m>0$ and define $B_m\coloneqq\sum_{\ell\ne m}A_\ell$ as the sum of the other prior-weighted states. Since $Z=A_m+B_m\ge A_m$, the state $\tau_m$ is supported on $\supp Z$. Applying \lemref{lem:cheng} to $A_m$ and $B_m$, we obtain
\begin{align}
p_m\bigl(1-\Tr[\Lambda_m\tau_m]\bigr)
 &=\Tr[A_m(\Pi_Z-Z^{-1/2}A_mZ^{-1/2})]\label{eq:pgm-one-versus-rest}\\
&=\Tr[A_mZ^{-1/2}B_mZ^{-1/2}]\label{eq:line-2-33-2}\\
&\le\Tr[(I-T_m)A_m]+\Tr[T_mB_m]\label{eq:line-2-33-3}\\
&=p_m\Tr[(I-T_m)\tau_m]+
        \sum_{\ell\ne m}p_\ell\Tr[T_m\tau_\ell].\label{eq:line-2-33-4}
\end{align}
When $p_m=0$, the left side and the trace expression are zero and the last inequality still holds by positivity. Sum over $m$ and relabel the two finite summation indices:
\begin{equation}\label{eq:pgm-index-swap}
 \sum_m\sum_{\ell\ne m}p_\ell\Tr[T_m\tau_\ell]
 =\sum_m p_m\sum_{\ell\ne m}\Tr[T_\ell\tau_m].
\end{equation}
This proves \eqref{eq:weighted-pgm-error}; choosing $p_m=1/M$ proves \eqref{eq:pgm-error}.
\end{proof}

The bound has the familiar classical form: reject the true hypothesis, or accept one of the alternatives. However, it is not a probability of a union of jointly measurable events. Cheng's trace inequality bounds the error of the state-based PGM in terms of the rejection and false-acceptance probabilities of arbitrary binary tests. Summing these bounds yields the displayed decoding guarantee even when the tests are mutually incompatible. Nor does the inequality construct one test that preserves all separately optimized false-acceptance bounds. For the multiple-access application, Sen's geometry supplies precisely those suitable analysis tests on the enlarged space; the state-based PGM then converts their estimates into a decoding guarantee. In the main channel-coding proof the prior is uniform; the weighted form is reserved for the compression application in Section~\ref{sec:compression}.

\section{Two senders: the problem and the obstruction}\label{sec:two-setup}
We now identify the three tests that a two-sender decoder must accommodate for decoding information transmitted over multiple-access channels. After stating the coding problem, we explain why the classical strategy of intersecting acceptance events does not directly extend to noncommuting quantum projectors. The final example isolates the obstacle that the auxiliary labels will resolve.
\subsection{Channel, codebooks, and rate region}
A two-sender classical--quantum multiple-access channel is specified by states
\begin{equation}\label{eq:two-channel}
 \cN\colon(x_1,x_2)\longmapsto\rho_{x_1,x_2},
 \qquad (x_1,x_2)\in\cX_1\times\cX_2,
\end{equation}
on an output system $B$, where $\cX_1$ is an alphabet for Sender~1 and $\cX_2$ is an alphabet for Sender~2. Fix distributions $p_1,p_2$ and their product $p(x_1,x_2)\coloneqq p_1(x_1)p_2(x_2)$. The corresponding classical--quantum state is
\begin{equation}\label{eq:two-omega}
 \omega_{X_1X_2B}
 \coloneqq\sum_{x_1,x_2}p_1(x_1)p_2(x_2)
   \proj{x_1}_{X_1}\otimes\proj{x_2}_{X_2}\otimes\rho_{x_1,x_2}.
\end{equation}
The input registers in this state are mathematical records of the classical inputs. They are not delivered to the receiver.

An $n$-use code consists of two codebooks $x_i^n(m_i)$, $m_i\in[M_i]$, and a POVM $(\Lambda_{m_1,m_2})_{m_1,m_2}$ on $B^n$. For a message pair $m=(m_1,m_2)$, the output is
\begin{equation}\label{eq:memoryless-two}
 \rho_m^{(n)}
 \coloneqq\bigotimes_{t=1}^n\rho_{x_{1,t}(m_1),x_{2,t}(m_2)}.
\end{equation}
The messages are chosen independently and uniformly at random. Let $\cC$ denote the collection of codebooks and $\Lambda$ the decoding POVM. With a failure outcome counted as an error, the average error is
\begin{equation}\label{eq:code-error}
 \Pe(\cC,\Lambda)
 \coloneqq1-\frac{1}{M_1M_2}\sum_{m_1,m_2}
      \Tr[\Lambda_{m_1,m_2}\rho_m^{(n)}].
\end{equation}
A rate pair $(R_1,R_2)$ is achievable if a sequence of such codes has $\Pe\to0$ and, for each $i\in\{1,2\}$, satisfies $\liminf_{n\to\infty}n^{-1}\log_2 M_i\ge R_i$.

To prove that such codes exist, we use random coding and then select a good realization. Specifically, we generate the codebooks by sampling every $x_i^n(m_i)$ independently according to $p_i^{\otimes n}$. The decoder may depend on the realized codebooks. Codebook randomness is used to prove existence; it need not remain as a shared resource when the code is used.

\begin{theorem}[Two-sender simultaneous decoding; Sen~\cite{Sen2021Joint}]\label{thm:two-capacity}
Let $(x_1,x_2)\mapsto\rho_{x_1,x_2}$ be a channel with finite classical input alphabets and a finite-dimensional quantum output $B$, as in \eqref{eq:two-channel}. Fix a product input distribution $p_1p_2$ and its input--output state $\omega_{X_1X_2B}$ in \eqref{eq:two-omega}. With average error defined by \eqref{eq:code-error}, every nonnegative rate pair satisfying
\begin{align}
R_1&\le I(X_1;B|X_2)_\omega,\label{eq:two-region}\\
R_2&\le I(X_2;B|X_1)_\omega,\label{eq:line-3-5-2}\\
R_1+R_2&\le I(X_1X_2;B)_\omega\label{eq:line-3-5-3}
\end{align}
is achievable using unsplit codebooks with every Sender~$i$ codeword sampled independently from $p_i^{\otimes n}$ and one joint decoding POVM. For positive rates satisfying all three inequalities strictly, one may take $M_i=\lfloor2^{nR_i}\rfloor$, with expected average error tending to zero. Boundary rates are obtained by a limiting choice of code sizes, without time sharing or rate splitting. Zero-rate senders are included.
\end{theorem}
The region is illustrated in \fref{fig:region}. The two-sender result was first established by Fawzi \textit{et al.}~\cite{FawziEtAl2012}; here we follow Sen's geometric framework. The theorem is an achievability statement with a particular random-codebook construction; the capacity characterization and its converse are already established in~\cite{Winter2001}.

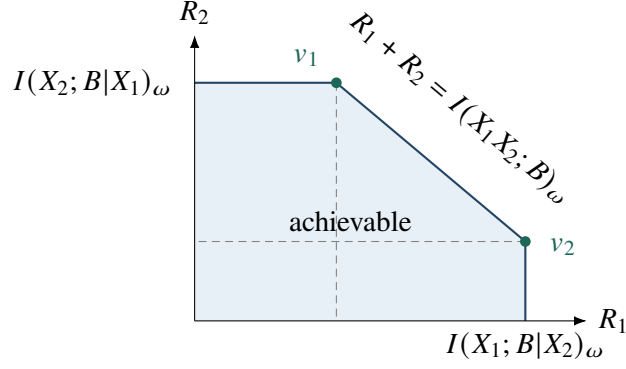
\begin{figure}[!t]
\centering
\begin{adjustbox}{max width=\linewidth}
\begin{tikzpicture}[font=\small,x=1.25cm,y=1.05cm]
\fill[softblue] (0,0)--(3.5,0)--(3.5,1)--(1.5,3)--(0,3)--cycle;
\draw[inkblue,thick] (0,3)--(1.5,3)--(3.5,1)--(3.5,0);
\draw[->] (0,0)--(4.15,0) node[right] {$R_1$};
\draw[->] (0,0)--(0,3.6) node[above] {$R_2$};
\draw[densely dashed,gray] (1.5,3)--(1.5,0);
\draw[densely dashed,gray] (3.5,1)--(0,1);
\fill[inkgreen] (1.5,3) circle (2pt) node[above left,xshift=-2pt,yshift=2pt] {$v_1$};
\fill[inkgreen] (3.5,1) circle (2pt) node[right,xshift=5pt,yshift=-2pt] {$v_2$};
\node at (1.65,1.25) {achievable};
\node[below] at (3.5,0) {$I(X_1;B|X_2)_\omega$};
\node[anchor=east,align=right] at (-.12,3) {$I(X_2;B|X_1)_\omega$};
\node[rotate=-40,anchor=south] at (2.65,2.45) {$R_1+R_2=I(X_1X_2;B)_\omega$};
\end{tikzpicture}
\end{adjustbox}
\caption{A schematic nondegenerate two-sender rate region for one fixed product input distribution. The successive-decoding corners are $v_1=(I(X_1;B)_\omega,I(X_2;B|X_1)_\omega)$ and $v_2=(I(X_1;B|X_2)_\omega,I(X_2;B)_\omega)$. The simultaneous decoder treats every strictly interior rate pair using unsplit, independently generated codebooks, rather than obtaining intermediate points by time sharing. Degenerate regions are included by the boundary argument.}
\label{fig:region}
\end{figure}

\subsection{Three tests for three error events}\label{subsec:two-tests}
For one use of the channel, or for a block treated as a single use, fix a candidate pair $\ell=(\ell_1,\ell_2)$ of messages. There are three ways that the candidate pair $\ell$ can differ from the true pair $m$. We label them by the symbols $1$, $2$, and $12$:
\begin{equation}\label{eq:error-types}
 s(m,\ell)\coloneqq
 \begin{cases}
 1,&m_1\ne\ell_1,\ m_2=\ell_2,\\
 2,&m_1=\ell_1,\ m_2\ne\ell_2,\\
 12,&m_1\ne\ell_1,\ m_2\ne\ell_2.
 \end{cases}
\end{equation}
For each candidate message pair $\ell=(\ell_1,\ell_2)$, we will
choose three acceptance projectors $P_{1,\ell}$, $P_{2,\ell}$,
and $P_{12,\ell}$. Each projector tests the candidate against
a different class of alternative message pairs. The projector
$P_{1,\ell}$ is intended to reject output states associated
with pairs $(m_1,\ell_2)$ for which $m_1\ne\ell_1$:
Sender~2's message agrees with the candidate, but Sender~1's
does not. Similarly, $P_{2,\ell}$ is intended to reject output
states associated with pairs $(\ell_1,m_2)$ for which
$m_2\ne\ell_2$. Finally, $P_{12,\ell}$ is intended to reject
output states associated with pairs $(m_1,m_2)$ for which
both $m_1\ne\ell_1$ and $m_2\ne\ell_2$. All three projectors
should accept the output state of the proposed pair $\rho_\ell$
with high probability.
Thus, for an alternative message pair $m$ of type
$s\in\{1,2,12\}$, the false-acceptance probability of interest
is $\Tr[P_{s,\ell}\rho_m]$, whereas the correct-candidate
acceptance probability is $\Tr[P_{s,\ell}\rho_\ell]$.

Classically, intersecting three acceptance sets preserves the desired false-acceptance bounds, while a union bound controls rejection of the correct message. An exact intersection of quantum acceptance subspaces does not have the corresponding stability property.

\subsection{Why an exact intersection can fail}\label{subsec:intersection-obstruction}
The difficulty in porting the notion of intersection to the quantum case is genuinely noncommutative. For classical acceptance events $E_+$ and $E_-$, the union bound gives
\begin{equation}\label{eq:classical-intersection-probability}
 \Pr(E_+\cap E_-)
 =1-\Pr(E_+^c\cup E_-^c)
 \ge 1-\Pr(E_+^c)-\Pr(E_-^c).
\end{equation}
Thus, if each test accepts with probability at least $1-\eps$, their intersection accepts with probability at least $1-2\eps$. Commuting quantum projectors obey the same inequality, because in a common eigenbasis they are zero--one indicator functions. Noncommuting acceptance subspaces need not have this property. This is the obstruction identified in Sen's analysis~\cite[Sections~1.4 and~3]{Sen2021Joint}.

For clarity, if $\mathcal U$ and $\mathcal V$ are subspaces, their intersection and set-theoretic union are defined as
\begin{equation}\label{eq:subspace-intersection-union}
 \mathcal U\cap\mathcal V\coloneqq\{\ket h:\ket h\in\mathcal U\text{ and }\ket h\in\mathcal V\},\qquad
 \mathcal U\cup\mathcal V\coloneqq\{\ket h:\ket h\in\mathcal U\text{ or }\ket h\in\mathcal V\}.
\end{equation}
The intersection is a subspace, whereas the union need not be: adding vectors taken from its two parts may leave the union. Its linear span is the smallest subspace containing both parts. This distinction is essential below.

The following example exhibits the failure already for a pure qubit state and rank-one projectors; neither mixing nor a large Hilbert space is responsible. Let $\ket{e_0},\ket{e_1}$ be an orthonormal basis of $\cc^2$, and take the correct state to be $\proj{e_0}$. For $0<\theta<\pi/4$, define
\begin{equation}\label{eq:bad-lines}
 \ket{q_\pm}\coloneqq\pm\sin\theta\ket{e_0}+\cos\theta\ket{e_1},
 \qquad P_\pm\coloneqq I-\proj{q_\pm}.
\end{equation}
The acceptance projectors are themselves pure-state projectors:
\begin{equation}\label{eq:acceptance-lines}
 P_\pm=\proj{p_\pm},\qquad
 \ket{p_\pm}\coloneqq\cos\theta\ket{e_0}\mp\sin\theta\ket{e_1}.
\end{equation}
Indeed, $\ket{p_\pm}$ and $\ket{q_\pm}$ are normalized and orthogonal. Each test accepts the correct state with probability
\begin{equation}\label{eq:separate-pure-acceptance}
 \langle e_0|P_\pm|e_0\rangle
 =|\langle p_\pm|e_0\rangle|^2
 =\cos^2\theta=1-\sin^2\theta,
\end{equation}
which tends to one as $\theta\to0$. Nevertheless,
\begin{equation}\label{eq:cancellation-obstruction}
 \ket{q_+}-\ket{q_-}=2\sin\theta\ket{e_0},
 \qquad
 \ket{q_+}+\ket{q_-}=2\cos\theta\ket{e_1}.
\end{equation}
Both coefficients are nonzero, so the two rejected lines span the whole qubit space. Consequently,
\begin{equation}\label{eq:intersection-zero}
 \ran P_+\cap\ran P_-
 =\bigl(\spanop\{\ket{q_+},\ket{q_-}\}\bigr)^\perp
 =\{0\}.
\end{equation}
The intersection projector accepts the correct state with probability zero for every positive $\theta$ in this interval, however small. In contrast, the classical lower bound in \eqref{eq:classical-intersection-probability} would tend to one. Here $0<|\langle p_+|p_-\rangle|=\cos(2\theta)<1$, so the two rank-one acceptance projectors are neither identical nor orthogonal and do not commute. \fref{fig:intersection} depicts the cancellation that causes the failure.

Following Sen's tilting idea~\cite[Section~3]{Sen2021Joint}, the remedy will be to change the geometry of the rejected vectors, not to force the original tests to commute. We attach a small auxiliary component depending on the message index and associate different error types with linearly independent auxiliary directions. A difference of labeled rejection vectors can then no longer cancel their common output component without also leaving auxiliary components. A finite Gram matrix will quantify the resulting control: the projector onto the labeled rejected span is bounded by a constant times the sum of the labeled rejection projectors. For each fixed label strength, that constant is independent of the output space and its subspaces. The labels remain close to a fixed reference state, so they can be removed from the completed measurement at the end. These are mathematical labels used to construct a decoder, not information physically supplied to the receiver.

\begin{figure}[!t]
\centering
\begin{adjustbox}{max width=\linewidth}
\begin{tikzpicture}[font=\small]
\begin{scope}
\fill[softblue] (-1.6,-1.65) rectangle (2.4,2.2);
\draw[->,gray] (-1.3,0)--(2.2,0) node[right] {$\ket{e_0}$};
\draw[->,gray] (0,-1.4)--(0,2) node[above left] {$\ket{e_1}$};
\draw[inkorange,thick] (-.57,-1.55)--(.76,2.05);
\draw[inkgreen,thick] (.57,-1.55)--(-.76,2.05);
\draw[->,inkorange,very thick] (0,0)--(.64,1.73) node[right] {$\ket{q_+}$};
\draw[->,inkgreen,very thick] (0,0)--(-.64,1.73) node[left] {$\ket{q_-}$};
\draw[->,inkblue,very thick] (0,0)--(1.65,0);
\node[below,inkblue] at (1.4,-.12) {correct vector};
\end{scope}
\node[align=left,text width=7.0cm,anchor=west] at (3.1,.4) {Each rejected line has only a small overlap with $\ket{e_0}$.\\[7pt]
Their linear span is nevertheless the whole plane.\\[7pt]
An exact intersection of the two acceptance spaces therefore rejects $\ket{e_0}$ completely.};
\end{tikzpicture}
\end{adjustbox}
\caption{Why an ordinary intersection is unstable. For $\ket{q_\pm}=\pm\sin\theta\ket{e_0}+\cos\theta\ket{e_1}$, each rejection probability on $\ket{e_0}$ is $\sin^2\theta$, but the span of $\ket{q_+}$ and $\ket{q_-}$ is the whole plane when $0<\theta<\pi/4$. The angle is exaggerated for visibility. Weak message labels avoid this problem by adding distinguishable auxiliary directions.}
\label{fig:intersection}
\end{figure}

\FloatBarrier
\section{An elementary two-sender packing bound}\label{sec:two-packing}
The obstruction in the preceding section concerns the geometry of the tests, independently of how the signals were encoded. We therefore temporarily set aside the channel and the random codebook and work with an arbitrary state for each message pair. This section constructs geometric analysis tests, proves their two error estimates, and uses them to certify a decoder on the original output. The states may depend arbitrarily on their message indices.

\subsection{Deterministic statement}
Let $\cM\coloneqq[M_1]\times[M_2]$, let $M\coloneqq M_1M_2$, and let $(\rho_m)_{m\in\cM}$ be arbitrary states on $B$. A decoding POVM has effects $(\Lambda_m)_{m\in\cM}$ and an optional failure effect. For uniformly distributed message pairs, define
\begin{equation}\label{eq:ensemble-error}
 \Pe\coloneqq1-\frac1M\sum_{m\in\cM}\Tr[\Lambda_m\rho_m].
\end{equation}
This is \eqref{eq:code-error} with a general state ensemble in place of the channel outputs; a failure outcome counts as an error.

For every message pair $m$ considered as the decoding candidate and every error type $s\in\{1,2,12\}$ from \eqref{eq:error-types}, let $P_{s,m}$ be a projector on $B$. We collect two kinds of binary-testing errors:
\begin{equation}\label{eq:two-a}
 a\coloneqq\frac1M\sum_{m\in\cM}\sum_{s\in\{1,2,12\}}
       \Tr[(I-P_{s,m})\rho_m]
\end{equation}
and
\begin{equation}\label{eq:two-b}
 b\coloneqq\frac1M\sum_m\left[\sum_{\ell_1\ne m_1}\Tr[P_{1,(\ell_1,m_2)}\rho_m]+\sum_{\ell_2\ne m_2}\Tr[P_{2,(m_1,\ell_2)}\rho_m]+\sum_{\substack{\ell_1\ne m_1\\\ell_2\ne m_2}}\Tr[P_{12,(\ell_1,\ell_2)}\rho_m]\right].
\end{equation}
An empty sum is equal to zero. Each summand has a direct error-probability interpretation. When $m$ was transmitted, $\Tr[(I-P_{s,m})\rho_m]$ is the probability that test $s$ rejects the correct message pair. For each fixed $s$, averaging this probability over $m$ gives an average binary-test error. The quantity $a$ is the sum of these three average errors. Small $a$ therefore means that all three tests usually accept the correct message pair, in the averaged sense used here.

The terms in $b$ are false-acceptance probabilities: the state is still $\rho_m$, but the acceptance projector belongs to a different message pair $\ell$. The notation $P_{s,\ell}$ has two indices because it specifies both an error type $s$ and the entire proposed message pair $\ell$. In the first sum, $\ell=(\ell_1,m_2)$: were this pair output by the decoder, Sender~2's message would be correct and Sender~1's would be incorrect. In the second sum, $\ell=(m_1,\ell_2)$, so the roles are interchanged. In the last sum, both proposed messages differ from the transmitted ones. Thus, for example, the $m_2$ appearing in $P_{1,(\ell_1,m_2)}$ is essential: it specifies the correctly matched second coordinate of that alternative pair.

Strictly speaking, $a$ and $b$ are \emph{sums of error probabilities}, not necessarily probabilities themselves. They satisfy $0\le a\le3$ and $0\le b\le M-1$. The number $a/3$ is the rejection probability when the true pair is uniform and one of its three tests is selected uniformly. When $M>1$, $b/(M-1)$ is the false-acceptance probability when the true pair is uniform, a distinct alternative pair is selected uniformly, and its appropriate test is applied. These are separate binary experiments, not simultaneous outcomes of one measurement. The unnormalized sums are useful because the decoding bound must account for every possible erroneous decision. The actual error probability $\Pe$ in \eqref{eq:ensemble-error} always lies in $[0,1]$.

\Needspace{12\baselineskip}
\begin{theorem}[Deterministic two-sender packing]\label{thm:two-packing}
Let $M_1$ and $M_2$ be positive integers. For each $m\in[M_1]\times[M_2]$, let $\rho_m$ be a state on a finite-dimensional system $B$ and let $P_{s,m}$ be projectors for $s\in\{1,2,12\}$. Let $a$ and $b$ be the sums of binary-test error probabilities in \eqref{eq:two-a} and \eqref{eq:two-b}. For every $0<\delta<1$, there is a decoding POVM on $B$, with a possible failure outcome, whose uniform-message average error $\Pe$ in \eqref{eq:ensemble-error} satisfies
\begin{equation}\label{eq:two-packing}
 \Pe\le d_2(\delta)+C_2(\delta)a+b,
\end{equation}
where
\begin{equation}\label{eq:two-constants}
 d_2(\delta)\coloneqq\sqrt{1-(1-\delta)^2},
 \qquad
 C_2(\delta)\coloneqq\bigl(1-\sqrt{1-\delta}\bigr)^{-2}.
\end{equation}
The constants depend only on $\delta$, not on $M_1,M_2$, the states, the projectors, or the output dimension.
\end{theorem}
The proof occupies the remainder of this section. We first explain the weak labels and the candidate projection, then construct a labeled span of rejected subspaces. Compression along each error direction gives the false-acceptance bounds. A Gram-matrix estimate controls correct-message rejection, and a worked qubit example illustrates why that control is possible. We finally analyze the pretty-good measurement of the tagged states using Cheng's trace inequality and remove the labels.

\subsection{Weak labels and the auxiliary space of a message pair}\label{subsec:two-labels}
The first step is to give the three error events different auxiliary directions without substantially changing the states that we wish to decode. Fully orthogonal message labels would make discrimination trivial, but they would also supply information that the receiver does not have. Identical labels would supply no information, but would leave the subspace-intersection problem unchanged. We use labels between these two extremes: each is close to the same reference vector, with a small component that identifies its message index. We call these auxiliary pure states \emph{weak labels}; their precise definition is given in \eqref{eq:two-label}. The small components will provide just enough linear independence to control the labeled rejection span, while closeness to the reference state will let us remove the labels after decoding. Thus the two requirements that appear to compete---separating the error directions and changing the states only slightly---will be met in two successive parts of the proof.

For Sender~$i$, introduce an auxiliary system $L_i$ with orthonormal basis
\begin{equation}
 \ket0,\ket1,\ldots,\ket{M_i}.
\end{equation}
The vector $\ket0$ is a reference vector and is not a message. For $0<\delta<1$, define
\begin{equation}\label{eq:two-label}
 \ket{u_{m_i}}\coloneqq\sqrt{1-\delta}\ket0+\sqrt\delta\ket{m_i},
 \qquad
 \ket{t_m}_L\coloneqq\ket{u_{m_1}}_{L_1}\otimes\ket{u_{m_2}}_{L_2},
 \qquad L=L_1L_2.
\end{equation}
The two components of $\ket{u_{m_i}}$ are orthogonal, so its squared norm is $(1-\delta)+\delta=1$. The label is weak in the following precise sense: a measurement in the displayed basis would return the reference symbol with probability $1-\delta$ and the message index with probability $\delta$. We do not perform this measurement. In particular, $\proj{u_{m_i}}$ is a coherent pure state, not a mixture of the reference and message basis states. Distinct weak labels are close rather than orthogonal:
\begin{equation}\label{eq:weak-label-overlap}
 \langle u_{\ell_i}|u_{m_i}\rangle=1-\delta\quad(m_i\ne\ell_i),
 \qquad
 \Delta(\proj{u_{m_i}},\proj0)=\sqrt\delta.
\end{equation}

We compare two mathematical ensembles on $BL$, illustrated in \fref{fig:weak-tags}. Their states, for a fixed message pair $m$, are
\begin{equation}\label{eq:labeled-states}
 \rhot m\coloneqq\rho_m\otimes\proj{t_m}_L,
 \qquad
 \rhoh m\coloneqq\rho_m\otimes\proj{00}_L.
\end{equation}
The first state has message-dependent labels and will be used only in the analysis. The second state uses a fixed auxiliary state that a receiver could prepare locally. There is no assertion that an unknown output state $\rho_m$ can be transformed into $\rhot m$ by a physical channel independent of $m$. For each fixed message, both states in \eqref{eq:labeled-states} are product states across $B:L$; attaching a label here does not mean entangling the output with it.

The labels are useful only if they can eventually be removed. The overlap with the reference label is
\begin{equation}\label{eq:two-reference-overlap}
 \langle00|t_m\rangle
 =\langle0|u_{m_1}\rangle\langle0|u_{m_2}\rangle
 =\sqrt{1-\delta}\sqrt{1-\delta}=1-\delta.
\end{equation}
Consequently, \eqref{eq:pure-distance} gives
\begin{align}
\Delta(\rhot m,\rhoh m)
 &=\tfrac12\norm{\rho_m\otimes(\proj{t_m}-\proj{00})}_1\label{eq:two-label-distance}\\
&=\tfrac12\norm{\proj{t_m}-\proj{00}}_1\label{eq:line-4-11-2}\\
&=\sqrt{1-|\langle00|t_m\rangle|^2}
 =\sqrt{1-(1-\delta)^2}=d_2(\delta).\label{eq:line-4-11-3}
\end{align}
The bound does not depend on the number of messages. Thus a fixed measurement that works well on the fictitious ensemble also works nearly as well on the fixed-label ensemble. We will make that comparison only after constructing the entire measurement.

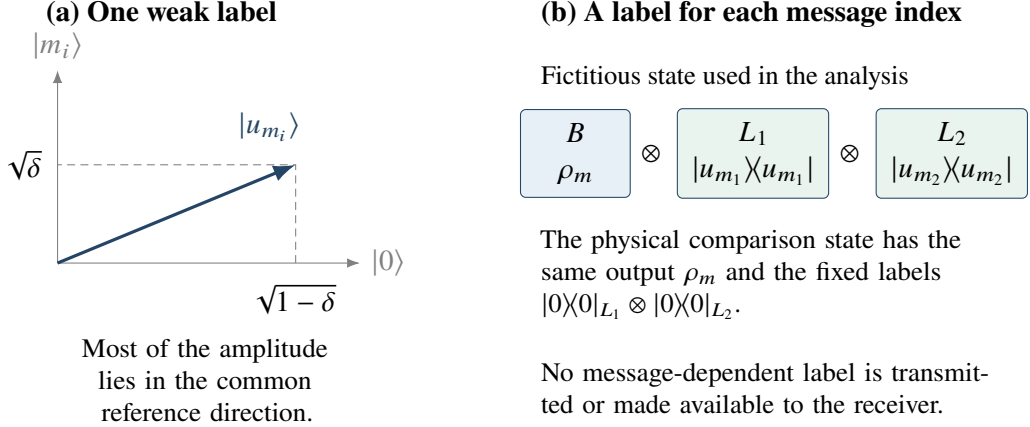
\begin{figure}[!t]
\centering
\begin{adjustbox}{max width=\linewidth}
\begin{tikzpicture}[font=\small]
\begin{scope}[shift={(0,0)}]
 \node[anchor=west,font=\small\bfseries] at (-.3,3.30) {(a) One weak label};
 \draw[->,gray] (0,0)--(4.0,0) node[right] {$\ket0$};
 \draw[->,gray] (0,0)--(0,2.55) node[above] {$\ket{m_i}$};
 \draw[->,inkblue,very thick] (0,0)--(3.15,1.30);
 \node[above,inkblue] at (2.80,1.47) {$\ket{u_{m_i}}$};
 \draw[densely dashed,gray] (3.15,0)--(3.15,1.30)--(0,1.30);
 \node[below] at (3.15,-.12) {$\sqrt{1-\delta}$};
 \node[left] at (-.06,1.30) {$\sqrt\delta$};
 \node[align=center,text width=4.5cm,font=\footnotesize] at (1.95,-1.55)
 {Most of the amplitude lies in the common reference direction.};
\end{scope}
\begin{scope}[shift={(6.35,0)}]
 \node[anchor=west,font=\small\bfseries] at (-.1,3.30) {(b) A label for each message index};
 \node[anchor=west,font=\footnotesize] at (-.1,2.45) {Fictitious state used in the analysis};
 \node[operator,fill=softblue,text width=1.1cm,minimum height=1.05cm] (b) at (.5,1.45) {$B$\\$\rho_m$};
 \node at (1.50,1.45) {$\otimes$};
 \node[operator,fill=softgreen,text width=1.65cm,minimum height=1.05cm] (l1) at (2.85,1.45) {$L_1$\\$\proj{u_{m_1}}$};
 \node at (4.13,1.45) {$\otimes$};
 \node[operator,fill=softgreen,text width=1.65cm,minimum height=1.05cm] (l2) at (5.48,1.45) {$L_2$\\$\proj{u_{m_2}}$};
 \node[align=left,anchor=west,text width=6.1cm,font=\footnotesize] at (-.1,-.15)
 {The physical comparison state has the same output $\rho_m$ and the fixed labels $\proj0_{L_1}\otimes\proj0_{L_2}$.};
 \node[align=left,anchor=west,text width=6.1cm,font=\footnotesize] at (-.1,-1.65)
 {No message-dependent label is transmitted or made available to the receiver.};
\end{scope}
\end{tikzpicture}
\end{adjustbox}
\caption{Weak message labels. Panel (a) shows the two orthogonal components of one normalized label; the components are rescaled for readability. Panel (b) displays the tensor factors of $\rhot m$ for a fixed $m=(m_1,m_2)$. The output state is unchanged. The message-dependent tensor factors are a device for constructing and analyzing a decoder, not additional communication. The normalized trace distance to the fixed-label state is $d_2(\delta)$.}
\label{fig:weak-tags}
\end{figure}

\paragraph{Why keep a two-dimensional space for each candidate?}
Fix a candidate message pair $\ell=(\ell_1,\ell_2)$. In the auxiliary register of Sender~$i$, we want to retain the component indexed by the proposed message $\ell_i$ and the common reference component, but discard the basis vectors belonging to other messages. Define
\begin{equation}\label{eq:two-K}
 K_{i,\ell_i}\coloneqq\proj0+\proj{\ell_i},
 \qquad K_\ell\coloneqq K_{1,\ell_1}\otimes K_{2,\ell_2}.
\end{equation}
The projector $K_{i,\ell_i}$ acts as the identity on the weak label for the candidate. For a different message, it removes just the message-specific component. Substituting the definition in \eqref{eq:two-K} gives
\begin{align}
K_{i,\ell_i}\ket{u_{m_i}}
 &=(\proj0+\proj{\ell_i})
       (\sqrt{1-\delta}\ket0+\sqrt\delta\ket{m_i})\label{eq:one-label-calculation}\\
&=\sqrt{1-\delta}\ket0
    +\sqrt\delta\,\langle\ell_i|m_i\rangle\ket{\ell_i}.\label{eq:line-4-13-2}
\end{align}
Thus
\begin{equation}\label{eq:one-label-projection}
 K_{i,\ell_i}\ket{u_{m_i}}
 =\begin{cases}
   \ket{u_{\ell_i}},&m_i=\ell_i,\\
   \sqrt{1-\delta}\ket0,&m_i\ne\ell_i.
  \end{cases}
\end{equation}
The receiver does not need to know which case holds to define $K_{i,\ell_i}$; its definition uses only the proposed message $\ell_i$. It is important that this projector retains a two-dimensional space rather than only the line through $\ket{u_{\ell_i}}$. The latter rank-one choice would act as
\begin{equation}\label{eq:rank-one-label-projection}
 \proj{u_{\ell_i}}\ket{u_{m_i}}
 =\langle u_{\ell_i}|u_{m_i}\rangle\ket{u_{\ell_i}}
 =\begin{cases}
    \ket{u_{\ell_i}},&m_i=\ell_i,\\
    (1-\delta)\ket{u_{\ell_i}},&m_i\ne\ell_i.
   \end{cases}
\end{equation}
In both cases, the surviving vector lies on the same line. The amplitudes differ, but after normalization there is no directional distinction. Tensoring such rank-one choices would send all error patterns to scalar multiples of the same vector $\ket{t_\ell}$. Attaching the rejection spaces to that one direction would recreate the unstable unlabelled span discussed previously in Section~\ref{subsec:intersection-obstruction}.

By contrast, \eqref{eq:one-label-projection} retains the entire vector $\ket{u_{\ell_i}}$ for a matching message and only the reference vector $\ket0$ for a nonmatching message. These two vectors are linearly independent for $\delta>0$, because only the former has a nonzero $\ket{\ell_i}$ component. The projector does not perfectly distinguish them: they are not orthogonal. It preserves exactly the geometric distinction needed later, without attempting to reveal which case occurred.

The range of $K_\ell$ is the four-dimensional space
\begin{equation}\label{eq:two-local-space}
 \ran K_\ell
 =\spanop\{\ket{00},\ket{\ell_1,0},\ket{0,\ell_2},\ket{\ell_1,\ell_2}\}.
\end{equation}
Although the entire label space can be large, each candidate message pair $\ell$ uses only this four-dimensional space. At this stage, $K_\ell$ is part of the mathematical definition of a detector; we have not prescribed a preliminary measurement that would reveal the error pattern.

\paragraph{Vectors for wrong-message patterns.}
For a fixed candidate message pair $\ell$, all true message pairs $m$ with the same error pattern $s(m,\ell)$ leave the same label direction after projection. We use the letter $w$ to indicate \emph{wrong-message pattern}, and define
\begin{equation}\label{eq:two-w}
 \ket{w_{\ell,1}}\coloneqq\ket0\ket{u_{\ell_2}},\qquad
 \ket{w_{\ell,2}}\coloneqq\ket{u_{\ell_1}}\ket0,\qquad
 \ket{w_{\ell,12}}\coloneqq\ket{00}.
\end{equation}
For example, suppose $m_1\ne\ell_1$ and $m_2=\ell_2$. If the receiver were to declare the proposed pair $\ell$, only Sender~1's message would be decoded incorrectly. Under $K_\ell$, the first label becomes proportional to~$\ket0$, while the second label survives unchanged. This explains both the error-type subscript $1$ and the appearance of Sender~2's weak label in $\ket{w_{\ell,1}}$. If both proposed messages differ from the true messages, only the reference vector survives in each register. At this point, no decoding outcome has occurred; the error type records a comparison between $m$ and $\ell$. These normalized vectors describe error-associated directions, not orthogonal flags from which the receiver can read the error pattern perfectly.

Tensoring the two cases in \eqref{eq:one-label-projection} gives
\begin{equation}\label{eq:two-exact-table}
 \begin{array}{c|c}
 \text{relation between }m\text{ and }\ell &K_\ell\ket{t_m}\\ \hline
 m=\ell &\ket{t_\ell}\\
 s(m,\ell)=1&\sqrt{1-\delta}\ket{w_{\ell,1}}\\
 s(m,\ell)=2&\sqrt{1-\delta}\ket{w_{\ell,2}}\\
 s(m,\ell)=12&(1-\delta)\ket{w_{\ell,12}}
 \end{array}
\end{equation}
Each nonmatching message coordinate contributes one factor $\sqrt{1-\delta}$. This table is exact: no approximation, channel property, or codebook average has been used. The labels must identify message indices, rather than codeword values, so that the table remains true even when distinct messages have identical codewords.

The three directions are particularly simple in an orthonormal basis of their span. Set
\begin{equation}\label{eq:two-e-basis}
 \ket{e_0}\coloneqq\ket{00},\qquad
 \ket{e_1}\coloneqq\ket{\ell_1,0},\qquad
 \ket{e_2}\coloneqq\ket{0,\ell_2}.
\end{equation}
Then
\begin{align}
\ket{w_{\ell,1}}&=\sqrt{1-\delta}\ket{e_0}+\sqrt\delta\ket{e_2},\label{eq:two-three-directions}\\
\ket{w_{\ell,2}}&=\sqrt{1-\delta}\ket{e_0}+\sqrt\delta\ket{e_1},\label{eq:line-4-20-2}\\
\ket{w_{\ell,12}}&=\ket{e_0}.\label{eq:line-4-20-3}
\end{align}
They are linearly independent for $\delta>0$: in a linear combination equal to zero, the $\ket{e_2}$ coefficient first forces the coefficient of $\ket{w_{\ell,1}}$ to vanish; the $\ket{e_1}$ coefficient does the same for $\ket{w_{\ell,2}}$; the remaining coefficient must then vanish. \fref{fig:labels} illustrates these directions. Their independence, not orthogonality, will control the sum of the rejection subspaces. To relate this to the obstruction in \sref{subsec:intersection-obstruction}, a linear combination of labeled rejection vectors has a separate transverse component for each of the first two error types. Cancelling the common reference component alone cannot eliminate those transverse components. The Gram-matrix lemma will make this statement quantitative even for sums $\sum_s\ket{h_s}\otimes\ket{w_{\ell,s}}$, where each $\ket{h_s}$ is an output vector drawn from its associated rejection space, rather than a scalar multiplying the label. This is why adding only small label components can regularize the rejected span without requiring mutually orthogonal error flags.

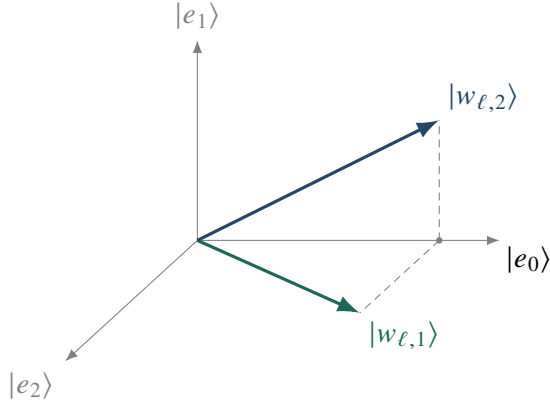
\begin{figure}[!t]
\centering
\begin{adjustbox}{max width=\linewidth}
\begin{tikzpicture}[font=\small]
\begin{scope}[shift={(2.15,1.35)}]
 \coordinate (O) at (0,0);
 \coordinate (E0) at (4.0,0);
 \coordinate (E1) at (0,2.65);
 \coordinate (E2) at (-1.75,-1.60);
 \coordinate (C0) at (3.20,0);
 \coordinate (W1) at (2.15,-.96);
 \coordinate (W2) at (3.20,1.59);
 \draw[->,gray] (O)--(E0);
 \node[below right,inner sep=2pt] at (E0) {$\ket{e_0}$};
 \draw[->,gray] (O)--(E1) node[above] {$\ket{e_1}$};
 \draw[->,gray] (O)--(E2) node[below left] {$\ket{e_2}$};
 \draw[densely dashed,gray] (C0)--(W1);
 \draw[densely dashed,gray] (C0)--(W2);
 \draw[->,inkgreen,very thick] (O)--(W1);
 \node[below right,inkgreen,inner sep=3pt] at (W1) {$\ket{w_{\ell,1}}$};
 \draw[->,inkblue,very thick] (O)--(W2);
 \node[above right,inkblue,inner sep=3pt] at (W2) {$\ket{w_{\ell,2}}$};
 \fill[gray] (C0) circle (1.2pt);
\end{scope}
\node[align=left,anchor=west,text width=5.1cm] at (8.0,3.30)
 {$\ket{w_{\ell,1}}$ retains Sender~2's label.\\[9pt]
  $\ket{w_{\ell,2}}$ retains Sender~1's label.\\[9pt]
  $\ket{w_{\ell,12}}=\ket{e_0}$.};
\node[align=left,anchor=west,text width=5.1cm,font=\footnotesize] at (8.0,.65)
 {All three vectors remain near the reference direction when $\delta$ is small, but their distinct transverse components make them linearly independent.};
\end{tikzpicture}
\end{adjustbox}
\caption{Auxiliary geometry for one candidate message pair. The basis vectors in \eqref{eq:two-e-basis} are orthonormal; the drawing is an oblique projection, not a representation of their Euclidean angles on the page. The dashed segments start at the common reference component $\sqrt{1-\delta}\ket{e_0}$ and show the transverse components in \eqref{eq:two-three-directions}--\eqref{eq:line-4-20-3}. The output space $B$ is not represented. The correct label also has a component $\delta\ket{\ell_1,\ell_2}$ in a fourth auxiliary direction, not shown here.}
\label{fig:labels}
\end{figure}

\FloatBarrier
\subsection{A labeled union of rejected subspaces}\label{subsec:two-rejection}
We now combine the three tests by working with their rejection conditions. In a classical problem, an observed output passes every test for a proposed message pair exactly when that output lies outside the union of the corresponding rejection events. The union adds no new sample points: a point belongs to it only if at least one rejection event contains that point. A union bound therefore controls its probability by the sum of the individual rejection probabilities.

The quantum counterpart requires more care. A linear operator that annihilates two vectors must annihilate their sum, and hence a detector that rejects several subspaces must reject their entire linear span. Unlike a set-theoretic union, that span contains superpositions that need not belong to any single rejected subspace. The pure-state example in \sref{subsec:intersection-obstruction} shows the consequence: subtracting two nearly parallel rejection vectors produces the correct state, although each separate rejection probability is small.

We use a message-label realization of Sen's tilted-span construction~\cite[Section~3]{Sen2021Joint}. Instead of taking the rejected span on $B$, we first associate each rejection space with its error-label vector in~$L$. The exact identity \eqref{eq:two-exact-table} makes this assignment appropriate for every competing message pair. Linear independence of the error-label vectors then prevents arbitrary cancellation between the labeled summands. Later, \lemref{lem:gram} will show that the projector onto their span is bounded by a controlled constant times the sum of their projectors. This is the quantitative substitute for the classical union bound; the original tests need not commute.

For each candidate message pair $\ell$ and error type $s\in\{1,2,12\}$, set
\begin{equation}\label{eq:two-rejections}
 Q_{s,\ell}\coloneqq I_B-P_{s,\ell},
\end{equation}
where $P_{s,\ell}$ is an acceptance projector defined after \eqref{eq:error-types}.
The space $\ran Q_{s,\ell}$ consists of output vectors rejected with certainty by the binary test for type $s$. Define its labeled counterpart and the span of all three counterparts by
\begin{equation}\label{eq:two-bad-space}
\cW_{s,\ell}\coloneqq
 \ran Q_{s,\ell}\otimes\spanop\{\ket{w_{\ell,s}}\},
\end{equation}
\begin{equation}\label{eq:W_ell-linear-span}
\cW_\ell\coloneqq
 \spanop\left(\bigcup_{s\in\{1,2,12\}}\cW_{s,\ell}\right)
 =\sum_{s\in\{1,2,12\}}\cW_{s,\ell}.
\end{equation}
To be clear, the sum on the right-hand side of \eqref{eq:W_ell-linear-span} is explicitly a linear span, not a sum of operators or a probabilistic mixture. Equivalently,
\begin{align}
\cW_\ell
 &=\spanop\!\left\{
   \ket {h_s} \otimes\ket{w_{\ell,s}}:
   \ket {h_s}\in\ran Q_{s,\ell},\ s\in\{1,2,12\}
   \right\}\label{eq:two-bad-explicit}\\
&=\left\{
   \sum_{s\in\{1,2,12\}}\ket{h_s}\otimes\ket{w_{\ell,s}}:
   \ket{h_s}\in\ran Q_{s,\ell}\ \text{for each }s
   \right\}.\label{eq:line-4-24-2}
\end{align}
The output vectors $\ket{h_s}$ need not be normalized and may be zero. The two expressions agree because all terms having the same $s$ can be combined inside the subspace $\ran Q_{s,\ell}$. In particular, $\cW_\ell$ contains the zero vector. If two of its vectors have output vectors $\ket{h_s}$ and $\ket{h'_s}$, their linear combination with scalar coefficients $\alpha,\beta\in\cc$ has output vectors $\alpha\ket{h_s}+\beta\ket{h'_s}\in\ran Q_{s,\ell}$. This verifies directly that $\cW_\ell$ is a subspace.

Rejecting this span enforces all three rejection conditions: whenever an output vector lies in $\ran Q_{s,\ell}$ and its label is $\ket{w_{\ell,s}}$, the combined detector must reject it. The span is the \emph{smallest subspace} with this property, just as a union is the smallest classical set containing all the rejection events. The difference is that a subspace must also contain all linear combinations. Auxiliary labels are what will make this necessary enlargement quantitatively controllable.

Let $\Pi_\ell\coloneqq\Pi_{\cW_\ell}$. An alternative characterization of both this projector and its range is obtained from the following positive semidefinite operator:
\begin{equation}\label{eq:two-positive-sum}
 H_\ell\coloneqq\sum_{s\in\{1,2,12\}}R_{s,\ell},
 \qquad R_{s,\ell}\coloneqq Q_{s,\ell}\otimes\proj{w_{\ell,s}}.
\end{equation}
The next lemma states the characterization precisely. The support is a subspace, whereas its support projector is an operator; these two objects should be distinguished.
\begin{lemma}[Support of the labeled rejection sum]\label{lem:labeled-support}
Fix a message pair $\ell$ and projectors $Q_{s,\ell}$ on $B$, for $s\in\{1,2,12\}$. Let the normalized vectors $\ket{w_{\ell,s}}$ be those in \eqref{eq:two-w}, let $\mathcal W_\ell$ be the span in \eqref{eq:W_ell-linear-span}, and define $H_\ell$ by \eqref{eq:two-positive-sum}. Then
\begin{equation}\label{eq:labeled-support-statement}
 \supp H_\ell=\mathcal W_\ell,\qquad
 \Pi_{\supp H_\ell}=\Pi_\ell.
\end{equation}
\end{lemma}
\begin{proof}
Each $R_{s,\ell}$ is a projector because its two tensor factors are projectors. Its range is $\mathcal W_{s,\ell}$.
For an arbitrary vector $\ket\psi$ on $BL$, self-adjointness and idempotence of $R_{s,\ell}$ give
\begin{align}
\langle\psi|H_\ell|\psi\rangle
 &=\sum_s\langle\psi|R_{s,\ell}|\psi\rangle\label{eq:two-positive-kernel}\\
&=\sum_s\langle\psi|R_{s,\ell}^\dagger R_{s,\ell}|\psi\rangle
 =\sum_s\norm{R_{s,\ell}\ket\psi}^2\label{eq:line-4-27-2}\\
&=\sum_s\norm{(Q_{s,\ell}\otimes\bra{w_{\ell,s}})\ket\psi}^2.\label{eq:line-4-27-3}
\end{align}
For the last equality, use
$R_{s,\ell}=(I_B\otimes\ket{w_{\ell,s}})(Q_{s,\ell}\otimes\bra{w_{\ell,s}})$
and note that tensoring with a normalized vector preserves the norm. Every term in the sum is nonnegative. Hence the expectation in \eqref{eq:two-positive-kernel}--\eqref{eq:line-4-27-3} is zero exactly when $R_{s,\ell}\ket\psi=0$ for every $s$. These equalities imply $H_\ell\ket\psi=0$, and the converse follows by taking the expectation of $H_\ell\ket\psi=0$.

For each $s$, the identity $R_{s,\ell}\ket\psi=0$ is equivalent to $\ket\psi\perp\cW_{s,\ell}$. Indeed, if $\ket y=R_{s,\ell}\ket x$ lies in the range, then
$\langle y|\psi\rangle=\langle x|R_{s,\ell}|\psi\rangle=0$.
Conversely, orthogonality to the range implies orthogonality to $R_{s,\ell}\ket\psi$, and therefore
$0=\langle\psi|R_{s,\ell}|\psi\rangle=\norm{R_{s,\ell}\ket\psi}^2$.
A vector is orthogonal to every summand if and only if it is orthogonal to every linear combination of their vectors. We have proved
\begin{equation}\label{eq:two-support-proof}
 \ker H_\ell
 =\bigcap_s\ker R_{s,\ell}
 =\bigcap_s\cW_{s,\ell}^{\perp}
 =\cW_\ell^\perp,
 \qquad
 \supp H_\ell=(\ker H_\ell)^\perp=\cW_\ell.
\end{equation}
The projectors onto these equal subspaces are equal, proving the second assertion.
\end{proof}

To contrast this with the classical problem, suppose that the original rejection projectors $Q_{s,\ell}$ commute. We can then recover the ordinary union bound directly in a common eigenbasis. The projector onto their unlabelled span is
\begin{equation}\label{eq:classical-join}
 Q_{\cup,\ell}\coloneqq I_B-\prod_s(I_B-Q_{s,\ell})
 =I_B-\prod_sP_{s,\ell}
 \le\sum_sQ_{s,\ell}.
\end{equation}
To prove both assertions, diagonalize all the $Q_{s,\ell}$ in a common orthonormal basis $(\ket z)_z$, and write
$Q_{s,\ell}\ket z=q_s(z)\ket z$, where $q_s(z)\in\{0,1\}$.
The eigenvalue of $Q_{\cup,\ell}$ on $\ket z$ is
\begin{equation}\label{eq:classical-join-eigenvalues}
 1-\prod_s(1-q_s(z))
 =\mathbf1\{q_s(z)=1\ \text{for at least one }s\}
 \le\sum_s q_s(z).
\end{equation}
It is equal to zero precisely outside every rejection space, and one on their span. This proves that $Q_{\cup,\ell}$ is the desired projector and proves the operator inequality. Taking the expectation in an arbitrary state $\rho$ gives
\begin{equation}\label{eq:classical-join-probability}
 \Tr[Q_{\cup,\ell}\rho]\le\sum_s\Tr[Q_{s,\ell}\rho].
\end{equation}
In the common eigenbasis, the numbers $\langle z|\rho|z\rangle$ form a probability distribution, so this is precisely the classical union bound.

For noncommuting rejection projectors, the projector onto their span need not be bounded above by $\sum_sQ_{s,\ell}$, unlike \eqref{eq:classical-join}. The span and its projector are still well defined; it is this operator upper bound that can fail. Our construction does not assert otherwise. Instead, it will establish
$\Pi_\ell\le C_2(\delta)H_\ell$
with the explicit finite coefficient $C_2(\delta)$ from \eqref{eq:two-constants}. Thus, for any state $\tau$ on $BL$, the combined rejection is bounded by
$C_2(\delta)\sum_s\Tr[R_{s,\ell}\tau]$.
The coefficient depends on the label strength, not on how nearly parallel the original rejection spaces are. Holding that strength fixed while the individual rejection probabilities tend to zero is what overcomes the quantum intersection obstruction.

\paragraph{Constructing an acceptance test from the rejected span.}
The ambient space associated with the candidate message pair $\ell$ is $\cH_B\otimes\ran K_\ell$. The span $\cW_\ell$ is contained in this space. To verify containment explicitly, write an arbitrary $\ket\psi\in\cW_\ell$ in the second form of \eqref{eq:two-bad-explicit}--\eqref{eq:line-4-24-2}. Since $K_\ell\ket{w_{\ell,s}}=\ket{w_{\ell,s}}$ for every $s$, linearity gives
\begin{align}
(I_B\otimes K_\ell)\ket\psi
 &=\sum_s\ket{h_s}\otimes K_\ell\ket{w_{\ell,s}}\label{eq:explicit-containment}\\
&=\sum_s\ket{h_s}\otimes\ket{w_{\ell,s}}=\ket\psi.\label{eq:line-4-32-2}
\end{align}
A vector is fixed by a projector exactly when it belongs to its range. Since $\ran(I_B\otimes K_\ell)=\cH_B\otimes\ran K_\ell$, this proves the claimed containment. We want to retain precisely the vectors in this ambient space that are orthogonal to the rejected span. Therefore set
\begin{equation}\label{eq:two-Gamma}
 \Gamma_\ell\coloneqq I_B\otimes K_\ell-\Pi_\ell.
\end{equation}
Above we write the identity operator $I_B$ explicitly: $K_\ell$ acts only on $L$, whereas $\Gamma_\ell$ and $\Pi_\ell$ act on $BL$.

The subtraction in \eqref{eq:two-Gamma} removes a contained subspace, not an arbitrary positive semidefinite operator. For every $\ket\psi$, the vector $\Pi_\ell\ket\psi$ belongs to $\cW_\ell$ and is therefore fixed by $I_B\otimes K_\ell$, by \eqref{eq:explicit-containment}--\eqref{eq:line-4-32-2}. This proves the first identity below; taking adjoints proves the second:
\begin{equation}\label{eq:containment}
\Pi_\ell= (I_B\otimes K_\ell)\Pi_\ell
 =\Pi_\ell(I_B\otimes K_\ell).
\end{equation}
Both terms in \eqref{eq:two-Gamma} are self-adjoint, and
\begin{align}
\Gamma_\ell^2
 &=(I_B\otimes K_\ell)^2
   -(I_B\otimes K_\ell)\Pi_\ell
   -\Pi_\ell(I_B\otimes K_\ell)+\Pi_\ell^2\label{eq:Gamma-projective}\\
&=I_B\otimes K_\ell-\Pi_\ell-\Pi_\ell+\Pi_\ell
 =\Gamma_\ell.\label{eq:line-4-35-2}
\end{align}
Thus $\Gamma_\ell$ is a projector, with range
\begin{equation}\label{eq:Gamma-range}
 \ran\Gamma_\ell
 =(\cH_B\otimes\ran K_\ell)\cap\cW_\ell^\perp.
\end{equation}
It restricts acceptance to the local label space associated with $\ell$ and, inside that space, keeps only the orthogonal complement of the rejected span. Equivalently,
\begin{equation}\label{eq:ambient-decomposition}
 \cH_B\otimes\ran K_\ell=\cW_\ell\oplus\ran\Gamma_\ell.
\end{equation}
Here $\oplus$ denotes an orthogonal direct sum. This decomposition is orthogonal even though the three summands defining $\cW_\ell$ need not be orthogonal to one another. In particular, for every $\ket\psi\in\cW_\ell$, both projectors in the difference act as the identity, so
\begin{equation}\label{eq:Gamma-annihilates-span}
 \Gamma_\ell\ket\psi
 =(I_B\otimes K_\ell)\ket\psi-\Pi_\ell\ket\psi
 =\ket\psi-\ket\psi=0.
\end{equation}
For a vector in the ambient space orthogonal to $\cW_\ell$, the first term acts as the identity and the second gives zero, so $\Gamma_\ell$ accepts that vector with certainty. \fref{fig:rejected-span} summarizes these two actions.

\begin{figure}[!t]
\centering
\begin{adjustbox}{max width=\linewidth}
\begin{tikzpicture}[font=\small]
\node[anchor=west,font=\bfseries] at (0,4.00) {(a) Form the span of all labeled rejections};
\node[panel,font=\footnotesize,inner sep=4pt,minimum width=4.45cm,minimum height=.90cm] (q1) at (1.75,2.95)
 {$\ran Q_{1,\ell}\otimes\spanop\{\ket{w_{\ell,1}}\}$};
\node[panel,font=\footnotesize,inner sep=4pt,minimum width=4.45cm,minimum height=.90cm] (q2) at (6.65,2.95)
 {$\ran Q_{2,\ell}\otimes\spanop\{\ket{w_{\ell,2}}\}$};
\node[panel,font=\footnotesize,inner sep=4pt,minimum width=4.45cm,minimum height=.90cm] (q12) at (11.55,2.95)
 {$\ran Q_{12,\ell}\otimes\spanop\{\ket{w_{\ell,12}}\}$};
\node[operator,fill=softorange,text width=3.0cm,minimum height=.9cm] (w) at (6.65,1.25)
 {$\cW_\ell$\\Their linear span};
\draw[flow] (q1.south)--(w.west);
\draw[flow] (q2.south)--(w.north);
\draw[flow] (q12.south)--(w.east);
\node[anchor=west,font=\bfseries] at (0,-.1) {(b) Keep the orthogonal complement inside the candidate space};
\draw[draw=inkblue!65,rounded corners=3pt] (0,-3.75) rectangle (13.3,-.8);
\node[anchor=north] at (6.65,-.92) {$\cH_B\otimes\ran K_\ell$};
\fill[softorange] (.20,-3.54) rectangle (5.48,-1.60);
\fill[softgreen] (5.70,-3.54) rectangle (13.10,-1.60);
\node[align=center,text width=4.8cm] at (2.84,-2.53)
 {$\cW_\ell$\\[5pt]Rejected by $\Gamma_\ell$\\[3pt]Projector $\Pi_\ell$};
\node[align=center,text width=6.6cm] at (9.4,-2.53)
 {$(\cH_B\otimes\ran K_\ell)\cap\cW_\ell^\perp$\\[5pt]Accepted by $\Gamma_\ell$\\[3pt]Projector $I_B\otimes K_\ell-\Pi_\ell$};
\end{tikzpicture}
\end{adjustbox}
\caption{From separate tests to one candidate detection projector. Panel (a) represents the sum of subspaces in \eqref{eq:two-bad-space}; its boxes are not a claim that the summands are orthogonal. Panel (b) represents the orthogonal decomposition \eqref{eq:ambient-decomposition}. The shaded areas are schematic and do not encode dimensions. The complement is taken inside the local space associated with $\ell$, not inside all of $BL$. A later Gram-matrix bound controls how much of the correct state lies in $\cW_\ell$.}
\label{fig:rejected-span}
\end{figure}

\paragraph{The first rejected line from the qubit example.}
Return to the pure-state example in \sref{subsec:intersection-obstruction}. For the moment, retain just its first rejection condition: set
$Q_{1,\ell}=\proj{q_+}$ and $Q_{2,\ell}=Q_{12,\ell}=0$.
Thus $P_{1,\ell}=\proj{p_+}$, with $\ket{p_+}$ and $\ket{q_+}$ from \eqref{eq:acceptance-lines} and \eqref{eq:bad-lines}. This temporary simplification lets us see explicitly what subtraction does to one of the very same rejection directions. The full two-line example will be evaluated in \sref{subsec:qubit-example}.

Only one labeled rejection space is now nonzero, so
\begin{align}
\cW_\ell&=\spanop\{\ket{q_+}_B\otimes\ket{w_{\ell,1}}_L\},\label{eq:single-rejection-example}\\
\Pi_\ell&=\proj{q_+}_B\otimes\proj{w_{\ell,1}}_L,\label{eq:line-4-39-2}\\
\Gamma_\ell&=I_B\otimes K_\ell-\proj{q_+}_B\otimes\proj{w_{\ell,1}}_L.\label{eq:line-4-39-3}
\end{align}
For every $\ket v\in\ran K_\ell$, substitute $K_\ell\ket v=\ket v$, $\langle q_+|p_+\rangle=0$, and $\langle q_+|q_+\rangle=1$ to obtain
\begin{align}
\Gamma_\ell(\ket{p_+}\otimes\ket v)
 &=\ket{p_+}\otimes K_\ell\ket v
   -\ket{q_+}\langle q_+|p_+\rangle\otimes
      \ket{w_{\ell,1}}\langle w_{\ell,1}|v\rangle\label{eq:single-rejection-action}\\
&=\ket{p_+}\otimes\ket v,\label{eq:line-4-40-2}\\
\Gamma_\ell(\ket{q_+}\otimes\ket v)
 &=\ket{q_+}\otimes K_\ell\ket v
   -\ket{q_+}\langle q_+|q_+\rangle\otimes
      \ket{w_{\ell,1}}\langle w_{\ell,1}|v\rangle\label{eq:line-4-40-3}\\
&=\ket{q_+}\otimes
      \bigl(\ket v-\langle w_{\ell,1}|v\rangle\ket{w_{\ell,1}}\bigr).\label{eq:line-4-40-4}
\end{align}
Thus the accepted output direction $\ket{p_+}$ is untouched. In the rejected output direction $\ket{q_+}$, only the component along its designated error label is removed. In particular, $\ket{q_+}\otimes\ket{w_{\ell,1}}$ is annihilated, whereas $\ket{q_+}\otimes\ket v$ is retained when $\ket v\perp\ket{w_{\ell,1}}$.

This is not the same as rejecting $\ket{q_+}$ for every label. That selective action is useful because an actual type-$1$ competing pair has exactly the direction $\ket{w_{\ell,1}}$ after projection by $K_\ell$. For the correct tagged vector, the rejected probability in this one-line example is
\begin{align}
\Tr[\Pi_\ell(\proj{e_0}\otimes\proj{t_\ell})]
 &=|\langle q_+|e_0\rangle|^2
   |\langle w_{\ell,1}|t_\ell\rangle|^2\label{eq:single-rejection-probability}\\
&=\sin^2\theta\,
   |\langle0|u_{\ell_1}\rangle
      \langle u_{\ell_2}|u_{\ell_2}\rangle|^2
 =(1-\delta)\sin^2\theta.\label{eq:line-4-41-2}
\end{align}
It is small for small $\theta$. Restoring the second rejected line changes the projector onto the span, so one cannot simply add two such expressions. The Gram-matrix argument below supplies the missing bound, and the full qubit calculation will show the effect exactly.

\FloatBarrier
\subsection{Error estimates and the geometric bound}\label{subsec:two-estimates}
The geometric analysis projector $\Gamma_\ell$ is now defined in \eqref{eq:two-Gamma}. It is not the final decoding POVM; it is a binary test used to certify that POVM's performance. We first check the property that motivated it: every alternative message is rejected according to its intended binary test (Section~\ref{subsec:two-wrong}). We then use the Gram matrix to prove that imposing all these rejection conditions does not spoil acceptance of the true message (Section~\ref{subsec:gram}). Section~\ref{subsec:two-correct} then analyzes the probability of correctly accepting the message when testing with $\Gamma_\ell$. The worked qubit example in Section~\ref{subsec:qubit-example} connects both estimates to the obstruction in Section~\ref{subsec:intersection-obstruction}.

\subsubsection{Incorrect-message acceptance by compression}\label{subsec:two-wrong}
The span construction in \eqref{eq:two-bad-space} must satisfy two requirements: reject alternative message pairs with small false-acceptance probability, and retain the correct pair with high probability. We prove the first requirement here. Fix a proposed pair $\ell$ and a true pair $m\ne\ell$, and set $s\coloneqq s(m,\ell)$ as in \eqref{eq:error-types}. The appropriate comparison is between acceptance by $\Gamma_\ell$ from \eqref{eq:two-Gamma} and by the original binary test $P_{s,\ell}$ introduced in \sref{subsec:two-tests}.

The exact label identity makes that comparison possible without intersecting or sandwiching the original tests on $B$. After projection by $K_\ell$, every true pair of this error type has the same auxiliary direction $\ket{w_{\ell,s}}$, up to the scalar already displayed in \eqref{eq:two-exact-table}. We therefore restrict the detector to states with that fixed label. The resulting operator on $B$ will be an acceptance effect supported inside $\ran P_{s,\ell}$. This is the sense in which the combined test is no more likely to accept a type-$s$ alternative than its designated individual test.

Define the isometry
\begin{equation}\label{eq:V-compression}
 V_{\ell,s}\colon\cH_B\longrightarrow\cH_B\otimes\cH_L,
 \qquad
 V_{\ell,s}\ket h\coloneqq\ket h\otimes\ket{w_{\ell,s}}.
\end{equation}
The only action of $V_{\ell,s}$ is to tensor in the fixed normalized error-label state $\ket{w_{\ell,s}}$; it does not rotate, filter, or measure the output vector. Thus
$V_{\ell,s}^\dagger V_{\ell,s}=I_B$,
and for every state $\rho$ on $B$,
$V_{\ell,s}\rho V_{\ell,s}^\dagger=\rho\otimes\proj{w_{\ell,s}}$.
This error label is a product of reference vectors and surviving weak labels, not an additional unknown register. Compressing the detection projector $\Gamma_\ell$ along this isometry gives
\begin{equation}\label{eq:compressed-A}
 A_{\ell,s}\coloneqq V_{\ell,s}^\dagger\Gamma_\ell V_{\ell,s}
 =(I_B\otimes\bra{w_{\ell,s}})\Gamma_\ell
   (I_B\otimes\ket{w_{\ell,s}}).
\end{equation}
This is an effect on $B$. To see explicitly that $0 \leq  A_{\ell,s}\leq I$, use $0\le\Gamma_\ell\le I_{BL}$ to obtain
\begin{equation}\label{eq:A-effect}
 0\le A_{\ell,s} = V_{\ell,s}^\dagger\Gamma_\ell V_{\ell,s}
 \le V_{\ell,s}^\dagger I_{BL}V_{\ell,s}=I_B.
\end{equation}
The acceptance probability of $\Gamma_\ell$ on $\rho\otimes\proj{w_{\ell,s}}$ is exactly the acceptance probability of $A_{\ell,s}$ on~$\rho$. Explicitly, for every state $\rho$ on $B$,
\begin{align}
\Tr[\Gamma_\ell(\rho\otimes\proj{w_{\ell,s}})]
 &=\Tr[\Gamma_\ell V_{\ell,s}\rho V_{\ell,s}^\dagger]\label{eq:compression-probability-identity}\\
&=\Tr[V_{\ell,s}^\dagger\Gamma_\ell V_{\ell,s}\rho]
 =\Tr[A_{\ell,s}\rho].\label{eq:line-4-45-2}
\end{align}
The second equality follows from cyclicity of the trace (between the original and extended spaces).

A crucial property of the effect $A_{\ell,s}$ is that it has no support in the rejection space of $P_{s,\ell}$. To see this, consider that, for every $\ket h$, the vector
\begin{equation}
 V_{\ell,s}Q_{s,\ell}\ket h
 =Q_{s,\ell}\ket h\otimes\ket{w_{\ell,s}}
 \ \text{belongs to }\cW_\ell.
\end{equation}
By \eqref{eq:Gamma-annihilates-span}, $\Gamma_\ell$ annihilates every vector in $\cW_\ell$. Applying that identity to $V_{\ell,s}Q_{s,\ell}\ket h$ for every $\ket h$ gives the first equality below. Left multiplication by $V_{\ell,s}^\dagger$, followed by the definition of $A_{\ell,s}$, gives the second; taking its adjoint gives the third:
\begin{equation}\label{eq:annihilation}
 \Gamma_\ell V_{\ell,s}Q_{s,\ell}=0,
 \qquad
 A_{\ell,s}Q_{s,\ell}=0,
 \qquad
 Q_{s,\ell}A_{\ell,s}=0.
\end{equation}
Inserting $I_B=P_{s,\ell}+Q_{s,\ell}$ on both sides of $A_{\ell,s}$ now gives
\begin{align}
A_{\ell,s}
 &=(P_{s,\ell}+Q_{s,\ell})A_{\ell,s}(P_{s,\ell}+Q_{s,\ell})\label{eq:A-support-expansion}\\
&=P_{s,\ell}A_{\ell,s}P_{s,\ell}
  +P_{s,\ell}A_{\ell,s}Q_{s,\ell}
  +Q_{s,\ell}A_{\ell,s}P_{s,\ell}
  +Q_{s,\ell}A_{\ell,s}Q_{s,\ell}\label{eq:line-4-48-2}\\
&=P_{s,\ell}A_{\ell,s}P_{s,\ell}.\label{eq:line-4-48-3}
\end{align}
Combining this support identity with \eqref{eq:A-effect} yields
\begin{equation}\label{eq:A-le-P}
 0\le A_{\ell,s}=P_{s,\ell}A_{\ell,s}P_{s,\ell}
 \le P_{s,\ell}I_BP_{s,\ell}=P_{s,\ell}.
\end{equation}

Equation~\eqref{eq:A-le-P} is the operator version of the classical implication that a candidate can pass a combined test only if it passes the relevant individual test. Importantly, each error type has its own compression $A_{\ell,s}$. We have not constructed a single effect on $B$ that is simultaneously dominated by all three $P_{s,\ell}$ while retaining their acceptance probabilities.

This also explains the difference from sandwiching the original projectors. A product such as $P_1P_2P_1$ is bounded above by $P_1$, but not generally by $P_2$; the final support depends on the outer projector. Here domination follows instead from an exact annihilation condition in the enlarged space. For each $s$, the appropriate compression is supported on $\ran P_{s,\ell}$ regardless of the other two tests. There is no commutativity assumption, and no need to move one original projector past another.

\paragraph{Substituting the states and labels.}
We spell out the calculation when declaring $\ell$ would decode only Sender~1's message incorrectly. Thus $m_1\ne\ell_1$ and $m_2=\ell_2$. Equation~\eqref{eq:two-exact-table} gives
\begin{equation}\label{eq:wrong-projected-label}
 K_\ell\proj{t_m}K_\ell
 =(\sqrt{1-\delta}\ket{w_{\ell,1}})
   (\sqrt{1-\delta}\bra{w_{\ell,1}})
 =(1-\delta)\proj{w_{\ell,1}}.
\end{equation}
The support relation for $\Gamma_\ell$ is
\begin{equation}\label{eq:Gamma-local-support}
 (I_B\otimes K_\ell)\Gamma_\ell
 =\Gamma_\ell(I_B\otimes K_\ell)=\Gamma_\ell.
\end{equation}
Inserting these two identities, together with \eqref{eq:labeled-states}, into the false-acceptance probability gives
\begin{align}
\Tr[\Gamma_\ell\rhot m]
 &=\Tr[\Gamma_\ell(\rho_m\otimes\proj{t_m})]\label{eq:two-wrong-bound}\\
&=\Tr[(I_B\otimes K_\ell)\Gamma_\ell(I_B\otimes K_\ell)
            (\rho_m\otimes\proj{t_m})]\label{eq:line-4-52-2}\\
&=\Tr[\Gamma_\ell(I_B\otimes K_\ell)
            (\rho_m\otimes\proj{t_m})(I_B\otimes K_\ell)]\label{eq:line-4-52-3}\\
&=\Tr[\Gamma_\ell(\rho_m\otimes K_\ell\proj{t_m}K_\ell)]\label{eq:line-4-52-4}\\
&=(1-\delta)\Tr[\Gamma_\ell(\rho_m\otimes\proj{w_{\ell,1}})]\label{eq:line-4-52-5}\\
&=(1-\delta)\Tr[\Gamma_\ell V_{\ell,1}\rho_m V_{\ell,1}^\dagger]\label{eq:line-4-52-6}\\
&=(1-\delta)\Tr[V_{\ell,1}^\dagger\Gamma_\ell V_{\ell,1}\rho_m]\label{eq:line-4-52-7}\\
&=(1-\delta)\Tr[A_{\ell,1}\rho_m]\label{eq:line-4-52-8}\\
&\le(1-\delta)\Tr[P_{1,\ell}\rho_m]
 \le\Tr[P_{1,\ell}\rho_m].\label{eq:line-4-52-9}
\end{align}
The third and seventh lines use cyclicity of the trace. The fourth uses that $K_\ell$ acts only on the label register. The sixth uses the defining action of $V_{\ell,1}$. Finally, \eqref{eq:A-le-P} and positivity of $\rho_m$ give the first inequality; $0<1-\delta<1$ gives the second.

The same substitutions give, for the other two error types,
\begin{alignat}{2}
\Tr[\Gamma_\ell\rhot m]
 &=(1-\delta)\Tr[A_{\ell,2}\rho_m]
 \le\Tr[P_{2,\ell}\rho_m],&&\qquad s(m,\ell)=2,\label{eq:two-other-wrong}\\
\Tr[\Gamma_\ell\rhot m]
 &=(1-\delta)^2\Tr[A_{\ell,12}\rho_m]
 \le\Tr[P_{12,\ell}\rho_m],&&\qquad s(m,\ell)=12.\label{eq:line-4-53-2}
\end{alignat}
For the second line, each of the two nonmatching message coordinates supplies one factor $\sqrt{1-\delta}$ at the vector level and hence one factor $1-\delta$ at the probability level. In all three cases,
\begin{equation}\label{eq:two-wrong-unified}
 \Tr[\Gamma_\ell\rhot m]\le\Tr[P_{s(m,\ell),\ell}\rho_m].
\end{equation}
There is no additive error in this estimate. This is essential because the coding argument will sum it over many competitors.

\subsubsection{A Gram-matrix bound for rejected subspaces}\label{subsec:gram}
The false-acceptance estimate is now established for every alternative pair. It remains to show that imposing all the labeled rejection conditions does not discard too much of the correct state. The unresolved step is precisely the one that fails without labels: small overlap with each rejected subspace does not automatically imply small overlap with their span.

To see what must be controlled, consider a vector in that span written as $\sum_j\ket{h_j}\otimes\ket{v_j}$, where the vectors $\ket{h_j}$ come from the different output rejection spaces and $\ket{v_j}$ are their auxiliary labels. If the tensor-product terms $\ket{h_j}\otimes\ket{v_j}$ nearly canceled one another, normalizing their small residual could produce a direction with large overlap with the correct tagged state.  This was the mechanism in \eqref{eq:cancellation-obstruction}. Linearly independent labels prevent such cancellation from being arbitrarily severe. Their Gram matrix measures how much cancellation remains possible.

For auxiliary vectors $\ket{v_1},\ldots,\ket{v_r}$, introduce an $r$-dimensional register $J$ with orthonormal basis $(\ket j)_{j=1}^r$. Their \emph{Gram matrix} is the operator
\begin{equation}\label{eq:gram-prelim-definition}
 G\coloneqq\sum_{j,\ell=1}^r
       \langle v_j|v_\ell\rangle\ket j\bra\ell,
 \qquad G_{j\ell}=\langle v_j|v_\ell\rangle.
\end{equation}
It records every pairwise inner product. For scalar coefficients $z_j$, set $\ket z\coloneqq \sum_jz_j\ket j$. Expanding the squared norm and then using the matrix entries gives
\begin{align}
\left\|\sum_jz_j\ket{v_j}\right\|^2
 &=\left(\sum_j\overline z_j\bra{v_j}\right)
   \left(\sum_\ell z_\ell\ket{v_\ell}\right)\label{eq:gram-scalar-intuition}\\
&=\sum_{j,\ell}\overline z_j z_\ell\langle v_j|v_\ell\rangle
 =\langle z|G|z\rangle.\label{eq:line-4-56-2}
\end{align}
In particular, $G\ge0$. If $G\ge\lambda I_J$ for $\lambda>0$, then
\begin{equation}\label{eq:gram-scalar-lower}
 \left\|\sum_jz_j\ket{v_j}\right\|^2
 \ge\lambda\langle z|z\rangle
 =\lambda\sum_j|z_j|^2.
\end{equation}
The inequality follows by taking the expectation of the positive semidefinite operator $G-\lambda I_J$ in $\ket z$. It says that a linear combination cannot have arbitrarily small norm relative to its coefficients. The next lemma replaces the scalar coefficients $z_j$ by output vectors $\ket{h_j}$ in arbitrary rejection spaces, and then converts the norm bound into the needed projector inequality.

\Needspace{15\baselineskip}
\begin{lemma}[Subspaces with independent auxiliary vectors]\label{lem:gram}
Let $Q_1,\ldots,Q_r$ be projectors on a finite-dimensional Hilbert space $\cH$, and let $\ket{v_1},\ldots,\ket{v_r}$ be vectors in a finite-dimensional Hilbert space $\cH_L$. On an $r$-dimensional register $J$ with orthonormal basis $(\ket j)_{j=1}^r$, let
\begin{equation}\label{eq:gram-definition}
 G\coloneqq\sum_{j,\ell=1}^r
       \langle v_j|v_\ell\rangle\ket j\bra\ell
 \ge\lambda I_J
\end{equation}
for some $\lambda>0$. Define the subspace
\begin{align}
\cW&\coloneqq\spanop\!\left\{
   \ket h\otimes\ket{v_j}:\ket h\in\ran Q_j,\ j\in[r]
   \right\}\label{eq:gram-subspace}\\
&=\sum_{j=1}^r\ran Q_j\otimes\spanop\{\ket{v_j}\},\label{eq:line-4-59-2}
\end{align}
and let $\Pi_{\cW}$ be its projector. Then
\begin{equation}\label{eq:gram-bound}
 \Pi_{\cW}\le\lambda^{-1}\sum_{j=1}^r Q_j\otimes\proj{v_j}.
\end{equation}
The projectors $Q_j$ need not commute. The auxiliary vectors need not be normalized; when they are normalized, the summands on the right before multiplication by $\lambda^{-1}$ are projectors.
\end{lemma}
\begin{proof}
We first place the output vectors from the separate rejection spaces in orthogonal bookkeeping slots. Define
\begin{align}
\mathcal D
 &\coloneqq\spanop\!\left\{
   \ket h\otimes\ket j:\ket h\in\ran Q_j,\ j\in[r]
   \right\}\label{eq:gram-domain}\\
&=\left\{\sum_j\ket{h_j}\otimes\ket j:
            \ket{h_j}\in\ran Q_j\ \text{for each }j\right\}
 \subseteq\cH\otimes\cH_J.\label{eq:line-4-61-2}
\end{align}
The vectors in different slots are orthogonal because $\langle j|\ell\rangle=0$ for $j\ne\ell$. Hence the projector onto $\mathcal D$ and the norm of one of its vectors are
\begin{align}
\Pi_{\mathcal D}&=\sum_jQ_j\otimes\proj j,\label{eq:gram-domain-projector}\\
\left\|\sum_j\ket{h_j}\otimes\ket j\right\|^2
 &=\sum_{j,\ell}\langle h_j|h_\ell\rangle\langle j|\ell\rangle
 =\sum_j\norm{\ket{h_j}}^2.\label{eq:line-4-62-2}
\end{align}
For the projector identity, each summand projects onto one slot, their pairwise products vanish, and their ranges together span $\mathcal D$.

Define a linear operator $R\colon\cH_J\to\cH_L$ that replaces the bookkeeping basis by the auxiliary vectors:
\begin{equation}\label{eq:gram-label-map}
 R\coloneqq\sum_j\ket{v_j}\bra j,
 \qquad
 R^\dagger R
 =\sum_{j,\ell}\langle v_j|v_\ell\rangle\ket j\bra\ell=G.
\end{equation}
This operator is not assumed to be unitary or an isometry; its change of norm is governed by $G$. Restrict $I_{\cH}\otimes R$ to $\mathcal D$ and denote the resulting operator by $T$:
\begin{equation}\label{eq:gram-V}
 T\colon\mathcal D\longrightarrow\cH\otimes\cH_L,
 \qquad
 T\left(\sum_j\ket{h_j}\otimes\ket j\right)
 =\sum_j\ket{h_j}\otimes\ket{v_j}.
\end{equation}
For every $\ket h\in\mathcal D$, substitution gives
\begin{align}
\norm{T\ket h}^2
 &=\langle h|(I_{\cH}\otimes R^\dagger)(I_{\cH}\otimes R)|h\rangle\label{eq:gram-lower-step}\\
&=\langle h|I_{\cH}\otimes(R^\dagger R)|h\rangle
 =\langle h|I_{\cH}\otimes G|h\rangle\label{eq:line-4-65-2}\\
&=\lambda\langle h|h\rangle
   +\langle h|I_{\cH}\otimes(G-\lambda I_J)|h\rangle\label{eq:line-4-65-3}\\
&\ge\lambda\norm{\ket h}^2.\label{eq:line-4-65-4}
\end{align}
Tensoring a positive semidefinite operator with an identity preserves positivity, so the last term is nonnegative. This proves $T^\dagger T\ge\lambda I_{\mathcal D}$. It is the version for the output vectors $\ket{h_j}$ of \eqref{eq:gram-scalar-lower}.

By definition, $\ran T=\cW$. We determine $T^\dagger$ explicitly, remembering that its output must belong to $\mathcal D$. For $\ket h\in\mathcal D$ and $\ket\psi\in\cH\otimes\cH_L$,
\begin{align}
\langle Th|\psi\rangle
 &=\langle h|(I_{\cH}\otimes R^\dagger)|\psi\rangle\label{eq:gram-adjoint}\\
&=\langle h|\Pi_{\mathcal D}(I_{\cH}\otimes R^\dagger)|\psi\rangle.\label{eq:line-4-66-2}
\end{align}
We inserted $\Pi_{\mathcal D}$ because $\Pi_{\mathcal D}\ket h=\ket h$. Thus
$T^\dagger\ket\psi=\Pi_{\mathcal D}(I_{\cH}\otimes R^\dagger)\ket\psi$,
regarded as a vector in $\mathcal D$. Applying $T$ to it and expanding the projector yields
\begin{align}
TT^\dagger
 &=(I_{\cH}\otimes R)\Pi_{\mathcal D}(I_{\cH}\otimes R^\dagger)\label{eq:VVdagger}\\
&=(I_{\cH}\otimes R)
   \left(\sum_jQ_j\otimes\proj j\right)
   (I_{\cH}\otimes R^\dagger)\label{eq:line-4-67-2}\\
&=\sum_jQ_j\otimes(R\ket j)(\bra jR^\dagger)
 =\sum_jQ_j\otimes\proj{v_j}.\label{eq:line-4-67-3}
\end{align}
This is the positive sum whose support projector we must bound.

First identify its kernel. For every $\ket\psi$,
$\langle\psi|TT^\dagger|\psi\rangle=\norm{T^\dagger\ket\psi}^2$.
Hence $\ker(TT^\dagger)=\ker T^\dagger$. Moreover,
$T^\dagger\ket\psi=0$ exactly when
$\langle Th|\psi\rangle=0$ for every $\ket h\in\mathcal D$.
Since $\ran T=\cW$, this proves
\begin{equation}\label{eq:gram-kernel}
 \ker(TT^\dagger)=\ker T^\dagger=(\ran T)^\perp=\cW^\perp.
\end{equation}
Next let $\ket f\ne0$ satisfy $TT^\dagger\ket f=\mu\ket f$ with $\mu>0$. Then
\begin{align}
\norm{T^\dagger\ket f}^2
 &=\langle f|TT^\dagger|f\rangle
 =\mu\norm{\ket f}^2>0,\label{eq:gram-nonzero-eigenvalue}\\
T^\dagger T(T^\dagger\ket f)
 &=T^\dagger(TT^\dagger\ket f)
 =\mu(T^\dagger\ket f).\label{eq:line-4-69-2}
\end{align}
Thus $T^\dagger\ket f$ is a nonzero eigenvector of $T^\dagger T$, and $T^\dagger T\ge\lambda I_{\mathcal D}$ implies $\mu\ge\lambda$. Every nonzero eigenvalue of $TT^\dagger$ is therefore at least $\lambda$. Its nonzero eigenvectors span $\cW$, by \eqref{eq:gram-kernel}. If $(\ket{f_a})_a$ is an orthonormal eigenbasis of that space, its spectral decomposition gives
\begin{equation}\label{eq:gram-support-gap}
 TT^\dagger=\sum_a\mu_a\proj{f_a}
 \ge\lambda\sum_a\proj{f_a}
 =\lambda\Pi_{\cW}.
\end{equation}
Combine this with \eqref{eq:VVdagger}--\eqref{eq:line-4-67-3} and divide by $\lambda$ to obtain \eqref{eq:gram-bound}. If all $Q_j$ have zero range, both sides are zero; otherwise, zero-dimensional individual ranges cause no change in the proof.
\end{proof}

\paragraph{A controlled union bound.}
For a state $\tau$ on $\cH\otimes\cH_L$, the lemma implies that
\begin{equation}\label{eq:gram-probability}
 \Tr[\Pi_{\cW}\tau]
 \le\lambda^{-1}\sum_j\Tr[(Q_j\otimes\proj{v_j})\tau].
\end{equation}
This is the precise replacement for a union bound in the labeled space. If the auxiliary vectors are orthonormal, then $G=I_J$, the summands have orthogonal ranges, and the projector onto their sum equals the sum of their projectors. For normalized, linearly independent vectors, the factor $\lambda^{-1}$ quantifies the loss of orthogonality. For general unnormalized vectors it also accounts for their scales, and the terms are understood as expectations of positive semidefinite operators. When the vectors are normalized, each term is specifically the probability of a projective rejection test. If the auxiliary vectors coincide, the positive lower bound on $G$ fails, and one recovers the possibility of the unstable unlabelled span. The lemma is therefore not a universal bound for arbitrary unions of quantum subspaces: its hypothesis is the quantitative separation supplied by the labels.

\paragraph{The three error vectors.}
For our vectors in the order $1,2,12$, direct inner products in \eqref{eq:two-w} give
\begin{equation}\label{eq:two-Gram}
 G_{\mathrm{err}}\coloneqq 
 \begin{pmatrix}
  1&1-\delta&\sqrt{1-\delta}\\
  1-\delta&1&\sqrt{1-\delta}\\
  \sqrt{1-\delta}&\sqrt{1-\delta}&1
 \end{pmatrix}.
\end{equation}
For instance, the overlap of the first two vectors is
$\langle0|u_{\ell_1}\rangle\langle u_{\ell_2}|0\rangle=1-\delta$.
We could diagonalize this $3\times3$ matrix, but a tensor-product argument gives a simpler bound that will also apply to $k$ senders.

Adding $\ket{t_\ell}$ completes all four choices of the two single-register vectors $\ket0$ and $\ket{u_{\ell_i}}$. To see the tensor-product structure explicitly, set
\begin{equation}\label{eq:two-binary-label-choices}
 \ket{v_{i,0}}\coloneqq\ket0,\qquad
 \ket{v_{i,1}}\coloneqq\ket{u_{\ell_i}}.
\end{equation}
In lexicographic order of $(a,b)\in\{0,1\}^2$, the vectors $\ket{v_{1,a}}\otimes\ket{v_{2,b}}$ are
\begin{equation}\label{eq:two-full-vector-order}
 \ket{w_{\ell,12}},\quad
 \ket{w_{\ell,1}},\quad
 \ket{w_{\ell,2}},\quad
 \ket{t_\ell}.
\end{equation}
The single-register Gram matrix is the same in each register, because $\langle0|u_{\ell_i}\rangle=\sqrt{1-\delta}$. Furthermore,
\begin{equation}\label{eq:two-Gram-entry-factorization}
 \langle v_{1,a}\otimes v_{2,b}|v_{1,a'}\otimes v_{2,b'}\rangle
 =\langle v_{1,a}|v_{1,a'}\rangle
  \langle v_{2,b}|v_{2,b'}\rangle.
\end{equation}
These are exactly the entries of a tensor product of the two single-register Gram matrices. Hence
\begin{equation}\label{eq:two-full-Gram}
 G_{\mathrm{full}}
 \coloneqq \begin{pmatrix}1&\sqrt{1-\delta}\\\sqrt{1-\delta}&1\end{pmatrix}^{\otimes2}.
\end{equation}
The $2\times2$ factor has eigenvalues $1+\sqrt{1-\delta}$ and $1-\sqrt{1-\delta}$. Eigenvalues of its tensor square are products of two factor eigenvalues, so
\begin{equation}\label{eq:two-full-Gram-lower}
 G_{\mathrm{full}}\ge
       \bigl(1-\sqrt{1-\delta}\bigr)^2 I_4.
\end{equation}
The ordering in \eqref{eq:two-Gram} is $1,2,12$, whereas \eqref{eq:two-full-vector-order} is $12,1,2,\text{correct}$. For $\ket z$ with coefficients $(z_1,z_2,z_{12})$, form $\ket{\widetilde z}$ with coefficients $(z_{12},z_1,z_2,0)$ in the latter order. This permutation and insertion preserve the norm and identify the required principal submatrix. Consequently,
\begin{align}
\langle z|G_{\mathrm{err}}|z\rangle
 &=\langle\widetilde z|G_{\mathrm{full}}|\widetilde z\rangle\label{eq:two-Gram-lower}\\
&\ge\bigl(1-\sqrt{1-\delta}\bigr)^2\norm{\ket{\widetilde z}}^2\label{eq:line-4-78-2}\\
&=C_2(\delta)^{-1}\norm{\ket z}^2.\label{eq:line-4-78-3}
\end{align}
Applying \lemref{lem:gram} to \eqref{eq:two-bad-space} therefore yields
\begin{equation}\label{eq:two-bad-bound}
 \Pi_\ell\le C_2(\delta)
  \sum_{s\in\{1,2,12\}}Q_{s,\ell}\otimes\proj{w_{\ell,s}}.
\end{equation}
For each fixed $\delta>0$, the constant is finite and independent of the output dimension, the candidate message pair, and the message sizes. Its divergence as $\delta\downarrow0$ expresses the tradeoff already visible in the labels: making them harder to distinguish also weakens the control on their rejected span. In the asymptotic proof we will hold $\delta$ fixed while making the separate testing errors small.

\subsubsection{Correct-message acceptance}\label{subsec:two-correct}
The geometric lemma (Lemma~\ref{lem:gram}) supplies the step that is unavailable for arbitrary intersections on $B$. We now apply it to the correct message pair $m$, for which each binary rejection probability is already included in $a$, defined in \eqref{eq:two-a}. There are two possible sources of rejection by $\Gamma_m$: a component outside the local label space and a component in the labeled rejected span. The first will be exactly zero, and the second will be bounded by the Gram estimate from Lemma~\ref{lem:gram}.

For the correct pair, \eqref{eq:two-exact-table} gives $K_m\ket{t_m}=\ket{t_m}$. Substitute the definition $\rhot m=\rho_m\otimes\proj{t_m}$ from \eqref{eq:labeled-states} to obtain
\begin{align}
\Tr[(I_{BL}-I_B\otimes K_m)\rhot m]
 &=\Tr[(I_B\otimes(I_L-K_m))(\rho_m\otimes\proj{t_m})]\label{eq:correct-support-zero}\\
&=\Tr[\rho_m]\,\Tr[(I_L-K_m)\proj{t_m}]\label{eq:line-4-80-2}\\
&=\langle t_m|(I_L-K_m)|t_m\rangle=0.\label{eq:line-4-80-3}
\end{align}
The only possible rejection of the correct tagged state is therefore its overlap with $\cW_m$. Substituting \eqref{eq:two-Gamma} and then \eqref{eq:two-bad-bound}, we obtain
\begin{align}
\Tr[(I_{BL}-\Gamma_m)\rhot m]
 &=\Tr[(I_{BL}-I_B\otimes K_m+\Pi_m)\rhot m]\label{eq:two-correct-bound}\\
&=\Tr[\Pi_m(\rho_m\otimes\proj{t_m})]\label{eq:line-4-81-2}\\
&\le C_2(\delta)\sum_s
     \Tr[(Q_{s,m}\otimes\proj{w_{m,s}})
                         (\rho_m\otimes\proj{t_m})]\label{eq:line-4-81-3}\\
&=C_2(\delta)\sum_s
      \Tr[Q_{s,m}\rho_m]\,
      \Tr[\proj{w_{m,s}}\proj{t_m}]\label{eq:line-4-81-4}\\
&=C_2(\delta)\sum_s
      |\langle w_{m,s}|t_m\rangle|^2\Tr[Q_{s,m}\rho_m]\label{eq:line-4-81-5}\\
&\le C_2(\delta)\sum_s\Tr[(I_B-P_{s,m})\rho_m].\label{eq:line-4-81-6}
\end{align}
The fourth line factors the trace over $B$ and $L$. In the last line, the squared inner products are at most one because both vectors are normalized. Their exact values are $1-\delta$, $1-\delta$, and $(1-\delta)^2$ for $s=1,2,12$, respectively.

Only the overlap with the correct tagged state is needed in this estimate. The present Gram bound controls a larger set of label directions and is therefore conservative. In \sref{subsubsec:sharp-packing}, rescaling the same spanning vectors by their overlaps with the correct label improves the estimate without changing the rejected space or the decoder.

Every term $\Tr[(I_B-P_{s,m})\rho_m]=\Tr[Q_{s,m}\rho_m]$ is a \emph{rejection} probability for a test applied to the correct message state. Such rejection is an error, so it is desirable that these terms be small. The calculation in \eqref{eq:two-correct-bound}--\eqref{eq:line-4-81-6} shows that their labeled span also has small rejection, at a cost depending only on $\delta$. Equivalently, the correct tagged state's acceptance obeys
\begin{equation}\label{eq:correct-acceptance-lower}
 \Tr[\Gamma_m\rhot m]
 \ge1-C_2(\delta)\sum_s\Tr[Q_{s,m}\rho_m].
\end{equation}
This is the controlled replacement for the classical intersection guarantee: a small sum of individual rejection probabilities yields acceptance close to one, even though the original tests do not commute.

\subsubsection{A worked noncommuting example}\label{subsec:qubit-example}
We return to the instability example from \sref{sec:two-setup} and compute the effect of the labels exactly. This example explains together the rejected span, the subtraction in $\Gamma_\ell$, and the need for the Gram-matrix bound. Fix one candidate message pair $\ell$; this is an illustration of its detection projector, not a full communication code.

Let $B$ be a qubit with correct output $\rho_\ell\coloneqq\proj0$. The basis vectors $\ket0,\ket1$ here play the roles of $\ket{e_0},\ket{e_1}$ in that earlier example. For $0<\theta<\pi/4$, set
\begin{equation}\label{eq:example-q}
 \ket{q_+}=\sin\theta\ket0+\cos\theta\ket1,
 \qquad
 \ket{q_-}=-\sin\theta\ket0+\cos\theta\ket1,
\end{equation}
and choose
\begin{equation}\label{eq:example-tests}
 Q_{1,\ell}\coloneqq \proj{q_+},\qquad
 Q_{2,\ell}\coloneqq \proj{q_-},\qquad Q_{12,\ell}\coloneqq 0.
\end{equation}
Each nontrivial test rejects the correct output with probability $\sin^2\theta$, yet $\ket{q_+}$ and $\ket{q_-}$ span $B$. Rejecting their unlabelled span would therefore reject $\ket0$ with probability one, however small the positive angle is.

The labeled rejected span is instead
\begin{equation}\label{eq:example-r}
 \cW_\ell=\spanop\{\ket{r_1},\ket{r_2}\},\qquad
 \ket{r_1}\coloneqq\ket{q_+}\otimes\ket{w_{\ell,1}},\quad
 \ket{r_2}\coloneqq\ket{q_-}\otimes\ket{w_{\ell,2}}.
\end{equation}
The two vectors are normalized, and their overlap is
\begin{equation}\label{eq:example-overlap}
 \zeta\coloneqq\langle r_1|r_2\rangle
 =\langle q_+|q_-\rangle\langle w_{\ell,1}|w_{\ell,2}\rangle
 =(1-\delta)\cos(2\theta).
\end{equation}
Here $0<\zeta<1$, so an orthonormal basis of the rejected span is
\begin{equation}\label{eq:example-chi}
 \ket{\chi_+}\coloneqq
       \frac{\ket{r_1}+\ket{r_2}}{\sqrt{2(1+\zeta)}},
 \qquad
 \ket{\chi_-}\coloneqq
       \frac{\ket{r_1}-\ket{r_2}}{\sqrt{2(1-\zeta)}}.
\end{equation}
Indeed, the squared norms of the two numerators are $2(1+\zeta)$ and $2(1-\zeta)$, and their mutual inner product is $1-\zeta+\zeta-1=0$. Thus $\Pi_\ell=\proj{\chi_+}+\proj{\chi_-}$.

The correct tagged vector is $\ket{\psi_\ell}\coloneqq\ket0\otimes\ket{t_\ell}$. Its overlaps with the spanning vectors are
\begin{align}
\langle r_1|\psi_\ell\rangle
 &=\langle q_+|0\rangle\langle w_{\ell,1}|t_\ell\rangle
 =\sqrt{1-\delta}\sin\theta,\label{eq:example-r-overlaps}\\
\langle r_2|\psi_\ell\rangle
 &=\langle q_-|0\rangle\langle w_{\ell,2}|t_\ell\rangle
 =-\sqrt{1-\delta}\sin\theta.\label{eq:line-4-88-2}
\end{align}
It is therefore orthogonal to $\ket{\chi_+}$. Substitution into the other overlap gives the exact rejection probability
\begin{align}
\langle\psi_\ell|\Pi_\ell|\psi_\ell\rangle
 &=|\langle\chi_-|\psi_\ell\rangle|^2\label{eq:example-exact-rejection}\\
&=\frac{(2\sqrt{1-\delta}\sin\theta)^2}{2(1-\zeta)}\label{eq:line-4-89-2}\\
&=\frac{2(1-\delta)\sin^2\theta}
         {1-(1-\delta)\cos(2\theta)}\label{eq:line-4-89-3}\\
&=\frac{2(1-\delta)\sin^2\theta}
         {\delta+2(1-\delta)\sin^2\theta}.\label{eq:line-4-89-4}
\end{align}
Because $\ket{\psi_\ell}$ is in the local space associated with $\ell$, subtracting $\Pi_\ell$ from $I_B\otimes K_\ell$ leaves acceptance probability
\begin{align}
\langle\psi_\ell|\Gamma_\ell|\psi_\ell\rangle
 &=1-\langle\psi_\ell|\Pi_\ell|\psi_\ell\rangle\label{eq:example-exact-acceptance}\\
&=\frac{\delta}{\delta+2(1-\delta)\sin^2\theta}\label{eq:line-4-90-2}\\
&=\frac{1}{1+\dfrac{2(1-\delta)}{\delta}\sin^2\theta}
 \longrightarrow1\quad(\theta\to0,\ \delta>0\text{ fixed}).\label{eq:line-4-90-3}
\end{align}
The third line divides numerator and denominator by $\delta$. Thus, for every fixed $\delta>0$, the acceptance probability tends to one as $\theta\to0$, whereas the unlabelled intersection accepts with probability zero for every positive $\theta$. Quantitatively, \eqref{eq:example-exact-rejection}--\eqref{eq:line-4-89-4} is at most $2(1-\delta)\sin^2\theta/\delta$. This particular two-vector calculation is sharper than the general three-vector constant, but illustrates the same mechanism. The denominator has acquired a positive term $\delta$ that prevents arbitrarily small individual errors from producing complete rejection. \fref{fig:qubit-regularization} displays this comparison.

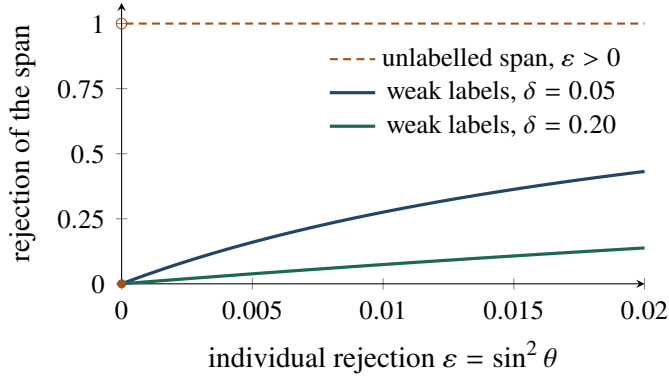
\begin{figure}[!t]
\centering
\begin{adjustbox}{max width=\linewidth}
\begin{tikzpicture}[font=\small]
\begin{axis}[
 width=8.5cm,height=5.3cm,
 xmin=0,xmax=.02,ymin=0,ymax=1.08,scaled x ticks=false,
 axis lines=left,
 xlabel={individual rejection $\varepsilon=\sin^2\theta$},
 ylabel={rejection of the span},
 xtick={0,.005,.01,.015,.02},xticklabels={$0$,$0.005$,$0.01$,$0.015$,$0.02$},
 ytick={0,.25,.5,.75,1},
 tick label style={font=\footnotesize},
 label style={font=\small},
 legend style={draw=none,font=\footnotesize,at={(.98,.68)},anchor=east},
 domain=0:.02,samples=120,clip=false]
 \addplot[inkorange,thick,densely dashed] coordinates {(.00015,1)(.02,1)};
 \addlegendentry{unlabelled span, $\varepsilon>0$}
 \addplot[inkblue,very thick] {2*(1-.05)*x/(.05+2*(1-.05)*x)};
 \addlegendentry{weak labels, $\delta=0.05$}
 \addplot[inkgreen,very thick] {2*(1-.2)*x/(.2+2*(1-.2)*x)};
 \addlegendentry{weak labels, $\delta=0.20$}
 \addplot[forget plot,only marks,mark=o,inkorange,mark size=2pt] coordinates {(0,1)};
 \addplot[forget plot,only marks,mark=*,inkorange,mark size=1.5pt] coordinates {(0,0)};
\end{axis}
\node[align=left,anchor=west,text width=4.15cm] at (9.1,2.95)
 {The unlabelled span rejects the correct vector completely for every positive $\varepsilon$.};
\node[align=left,anchor=west,text width=4.15cm] at (9.1,.95)
 {For fixed $\delta>0$, the labeled span's rejection tends to zero with the separate testing errors.};
\end{tikzpicture}
\end{adjustbox}
\caption{Exact calculation for the noncommuting qubit example. The solid curves are $2(1-\delta)\varepsilon/[\delta+2(1-\delta)\varepsilon]$, from \eqref{eq:example-exact-rejection}--\eqref{eq:line-4-89-4}. The unlabelled value is $1$ for $\varepsilon>0$ and $0$ at $\varepsilon=0$; the open circle marks the discontinuity. This figure concerns the fictitious tagged state's rejection by a candidate projector, before the state-based pretty-good-measurement analysis and removal of the labels. It does not compare complete coding-error bounds.}
\label{fig:qubit-regularization}
\end{figure}
\FloatBarrier

The order of limits is important in this example and in the proof. Taking $\delta\to0$ first, with $\theta>0$ fixed, recovers the unlabelled rejection probability one. Taking the individual testing errors to zero at fixed $\delta>0$ avoids that instability. Only after obtaining small testing errors do we choose a smaller label perturbation. The asymptotic argument in the next section implements precisely this order.

\subsection{The tagged-state decoder and removal of the labels}\label{subsec:two-remove}
The geometric argument has produced a good binary test $\Gamma_m$, defined in \eqref{eq:two-Gamma}, for every proposed message pair. We now use these tests to analyze a single decoding measurement. It is important to distinguish their roles: $\Gamma_m$ certifies performance, whereas the \emph{tagged states themselves} define the pretty-good measurement. Thus we do not require the $\Gamma_m$ projectors to be jointly measurable or to sum to the identity. Cheng's trace inequality makes the connection between the tests and the state-based decoder.

Two tasks remain. We first bound the decoder's error on the fictitious tagged ensemble. We then compress the completed POVM against the fixed reference label, obtaining a measurement on $B$ alone. The second step removes the message-dependent auxiliary states; it is not an extra operation that the sender must implement.

\paragraph{The states define the measurement.}
Apply \eqref{eq:weighted-pgm-definition} with $\tau_m=\rhot m$. Explicitly, on $BL$ set
\begin{equation}\label{eq:two-pgm-explicit}
 \widetilde Z\coloneqq\sum_m\rhot m,\qquad
 \widetilde\Lambda_m\coloneqq
       \widetilde Z^{-1/2}\rhot m\widetilde Z^{-1/2},\qquad
 \widetilde\Lambda_{\mathrm f}\coloneqq I_{BL}-\Pi_{\widetilde Z}.
\end{equation}
All inverses are on the support. The effects are positive and sum to the identity, as in \eqref{eq:state-pgm-complete}. The definition uses the realized output states, the message labels, and $\delta$, but not the separately selected projectors $P_{s,m}$. Different choices of binary tests can therefore certify the same decoder. 

\paragraph{The projectors certify its error.}
For a fixed $m$, let $B_m\coloneqq\sum_{\ell\ne m}\rhot\ell$. Substituting $A=\rhot m$, $B=B_m$, and $T=\Gamma_m$ into \eqref{eq:cheng-test} gives
\begin{align}
1-\Tr[\widetilde\Lambda_m\rhot m]
 &=\Tr[\rhot m\widetilde Z^{-1/2}B_m\widetilde Z^{-1/2}]\label{eq:two-pgm-cheng}\\
&\le\Tr[(I_{BL}-\Gamma_m)\rhot m]
       +\sum_{\ell\ne m}\Tr[\Gamma_m\rhot\ell].\label{eq:line-4-92-2}
\end{align}
The equality uses $\rhot m$ being supported on $\supp\widetilde Z$, exactly as in \eqref{eq:pgm-one-versus-rest}--\eqref{eq:line-2-33-4}. The inequality has coefficient one on each term.

The false-acceptance estimate previously proved has $\Gamma_\ell$ acting on $\rhot m$. In \eqref{eq:two-pgm-cheng}--\eqref{eq:line-4-92-2} the indices occur in the opposite order. The uniform-message average resolves this distinction by the exact reindexing
\begin{equation}\label{eq:two-reindex}
 \frac1M\sum_m\sum_{\ell\ne m}\Tr[\Gamma_m\rhot\ell]
 =\frac1M\sum_m\sum_{\ell\ne m}\Tr[\Gamma_\ell\rhot m].
\end{equation}
This is a relabeling of ordered pairs, not an equality between the two individual summands. Now use \eqref{eq:two-correct-bound}--\eqref{eq:line-4-81-6} for the first sum and \eqref{eq:two-wrong-unified} for the reindexed second sum:
\begin{align}
\widetilde\Pe
 &\coloneqq1-\frac1M\sum_m\Tr[\widetilde\Lambda_m\rhot m]\label{eq:two-tagged-error}\\
&\le\frac1M\sum_m\Tr[(I_{BL}-\Gamma_m)\rhot m]
       +\frac1M\sum_m\sum_{\ell\ne m}\Tr[\Gamma_\ell\rhot m]\label{eq:line-4-94-2}\\
&\le\frac{C_2(\delta)}M\sum_m\sum_s\Tr[(I_B-P_{s,m})\rho_m]
       +\frac1M\sum_m\sum_{\ell\ne m}\Tr[P_{s(m,\ell),\ell}\rho_m]\label{eq:line-4-94-3}\\
&=C_2(\delta)a+b.\label{eq:line-4-94-4}
\end{align}
Figure~\ref{fig:decoder-certificate} summarizes how the binary analysis tests certify the state-based decoder. The last equality uses the definitions of $a$ and $b$ in \eqref{eq:two-a} and \eqref{eq:two-b}. The geometry supplies both estimates, while Cheng's trace inequality converts them into a bound on the state-based decoder. 

\begin{figure}[!t]
\centering
\begin{adjustbox}{max width=\linewidth}
\begin{tikzpicture}[font=\small]
\node[panel,text width=3.0cm,minimum height=1.2cm] (states) at (0,2.8)
 {Tagged states\\$\rhot m$};
\node[panel,text width=4.0cm,minimum height=1.2cm] (pgm) at (6.0,2.8)
 {Define the measurement\\$\widetilde Z^{-1/2}\rhot m\widetilde Z^{-1/2}$};
\node[panel,fill=softgreen,text width=3.0cm,minimum height=1.2cm] (tests) at (0,0)
 {Separate tests\\$P_{1,m},P_{2,m},P_{12,m}$};
\node[panel,fill=softgreen,text width=4.0cm,minimum height=1.2cm] (gamma) at (6.0,0)
 {Construct analysis tests\\$\Gamma_m=I_B\otimes K_m-\Pi_m$};
\draw[flow] (states)--(pgm);
\draw[flow] (tests)--(gamma);
\draw[->,densely dashed,inkorange,thick] (gamma.north)--node[right,text width=3.5cm,align=left,font=\footnotesize]{Cheng's trace inequality\\bounds the decoder's error}(pgm.south);
\end{tikzpicture}
\end{adjustbox}
\caption{The states define the decoder; the geometric tests certify its performance. Solid arrows indicate a construction. The dashed arrow indicates an error estimate, not a physical operation or a dependency in the definition of the POVM. The projectors $P_{s,m}$ and $\Gamma_m$ are not its outcome effects. After the tagged-state analysis, the completed POVM is compressed against the fixed reference label.}
\label{fig:decoder-certificate}
\end{figure}
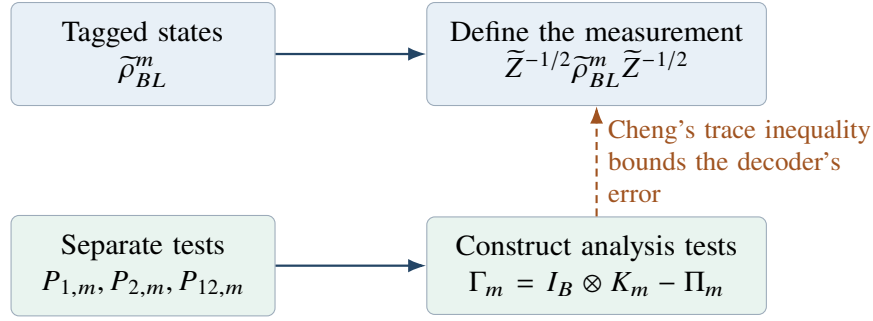

\paragraph{Compressing the completed measurement.}

We now turn the completed POVM into a measurement on the actual output. Define the fixed isometry
\begin{equation}\label{eq:J-two}
 J\colon \cH_B\longrightarrow\cH_B\otimes\cH_L,
 \qquad J\ket h=\ket h\otimes\ket{00}.
\end{equation}
This operation prepares a locally available reference state and is independent of the unknown message. Compress every effect, including failure:
\begin{equation}\label{eq:physical-POVM}
 \Lambda_m\coloneqq J^\dagger\widetilde\Lambda_mJ,
 \qquad
 \Lambda_{\mathrm f}\coloneqq J^\dagger\widetilde\Lambda_{\mathrm f}J.
\end{equation}
For every $\ket h$, positivity follows from
$\langle h|\Lambda_m|h\rangle=\langle Jh|\widetilde\Lambda_m|Jh\rangle\ge0$,
and likewise for failure. Completeness follows by substituting all definitions:
\begin{align}
\sum_m\Lambda_m+\Lambda_{\mathrm f}
 &=J^\dagger\left(\sum_m\widetilde\Lambda_m
                      +\widetilde\Lambda_{\mathrm f}\right)J\label{eq:physical-POVM-sum}\\
&=J^\dagger\bigl(\widetilde Z^{-1/2}\widetilde Z\widetilde Z^{-1/2}+I_{BL}-\Pi_{\widetilde Z}\bigr)J\label{eq:line-4-97-2}\\
&=J^\dagger I_{BL}J
 =I_B\langle00|00\rangle=I_B.\label{eq:line-4-97-3}
\end{align}
Thus the compressed effects form a POVM on the original system $B$.

For each message pair $m$, its actual success probability equals the success probability of the completed extended-space measurement on the fixed-label state. Explicit substitution gives
\begin{align}
\Tr[\Lambda_m\rho_m]
 &=\Tr[J^\dagger\widetilde\Lambda_mJ\rho_m]\label{eq:physical-success-identity}\\
&=\Tr[\widetilde\Lambda_mJ\rho_mJ^\dagger]\label{eq:line-4-98-2}\\
&=\Tr[\widetilde\Lambda_m(\rho_m\otimes\proj{00})]
 =\Tr[\widetilde\Lambda_m\rhoh m].\label{eq:line-4-98-3}
\end{align}
The second equality uses cyclicity of the trace. The third follows from the action of $J$ in \eqref{eq:J-two}, and the last is the definition in \eqref{eq:labeled-states}. Since $\widetilde\Lambda_m$ is an effect, \eqref{eq:effect-distance} and \eqref{eq:two-label-distance}--\eqref{eq:line-4-11-3} yield
\begin{align}
1-\Tr[\Lambda_m\rho_m]
 &=1-\Tr[\widetilde\Lambda_m\rhoh m]\label{eq:two-final-transfer}\\
&=1-\Tr[\widetilde\Lambda_m\rhot m]
   +\Tr[\widetilde\Lambda_m(\rhot m-\rhoh m)]\label{eq:line-4-99-2}\\
&\le1-\Tr[\widetilde\Lambda_m\rhot m]+\Delta(\rhot m,\rhoh m)\label{eq:line-4-99-3}\\
&=1-\Tr[\widetilde\Lambda_m\rhot m]+d_2(\delta).\label{eq:line-4-99-4}
\end{align}
The normalized trace distance in this comparison is exact. The estimate can nevertheless be sharpened by retaining the information that the tagged error is already small; \sref{subsubsec:sharp-packing} gives this error-sensitive comparison for the same completed POVM.

Average over $m$ and insert \eqref{eq:two-tagged-error}--\eqref{eq:line-4-94-4}. This proves \eqref{eq:two-packing} and completes the proof of \thref{thm:two-packing}.

\fref{fig:transfer} emphasizes the order of the last step. The comparison is made once, using the completed success effect for the transmitted message. Replacing the labeled state separately in every false-acceptance term would instead allow an additive error to accumulate over all competitors. We never make that replacement. A failure outcome can finally be assigned to an arbitrary message without increasing the average error.

\begin{figure}[!t]
\centering
\begin{adjustbox}{max width=\linewidth}
\begin{tikzpicture}[font=\small]
\node[panel,text width=4.0cm] (s1) at (0,2.1) {$\rho_m\otimes\proj{t_m}$};
\node[operator,text width=2.1cm] (d1) at (5.4,2.1) {$\widetilde\Lambda$};
\node[anchor=west] at (-2.2,2.95) {\textbf{Analysis only:} message-dependent labels};
\draw[flow] (s1)--(d1);
\draw[flow] (d1.east)--(8.2,2.1) node[right] {estimate};
\node[panel,text width=4.0cm,fill=softgreen] (s2) at (0,0) {$\rho_m\otimes\proj{00}$};
\node[operator,text width=2.1cm] (d2) at (5.4,0) {$\widetilde\Lambda$};
\node[anchor=west] at (-2.2,-.95) {\textbf{Equivalent physical decoder:} $\Lambda_m=J^\dagger\widetilde\Lambda_mJ$ on $B$};
\draw[flow] (s2)--(d2);
\draw[flow] (d2.east)--(8.2,0) node[right] {estimate};
\draw[<->,dashed,inkorange] (s1.south)--node[right,align=left,font=\footnotesize] {normalized trace\\distance $d_2(\delta)$} (s2.north);
\node[align=center,font=\footnotesize] at (5.4,1.08) {the same completed POVM};
\end{tikzpicture}
\end{adjustbox}
\caption{The labels are removed after the pretty-good measurement of the tagged states has been completed. Replacing the fictitious labeled state by the fixed-label state changes the success probability of that measurement by at most $d_2(\delta)$. This is one comparison for the transmitted message, not a separate perturbation estimate for every competitor. The lower row is implemented by a POVM on $B$ alone, with $J\ket h=\ket h\otimes\ket{00}$.}
\label{fig:transfer}
\end{figure}
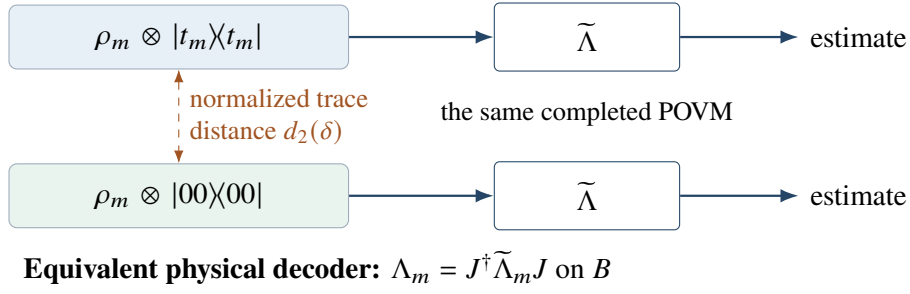

\FloatBarrier
\section{Random coding and the two-sender rate region}\label{sec:two-coding}
The geometric construction is now complete. To use it for a channel, we choose the separate tests so that the expectation of every false-acceptance term is a binary testing probability. We first carry this out for one use and then apply it to a block of $n$ uses.

\subsection{The one-shot random-coding bound}
For the channel in \eqref{eq:two-channel} and the distributions $p_1$ and $p_2$ defined thereafter, define the conditionally averaged states
\begin{equation}\label{eq:two-averages}
 \overline\rho_{x_2}\coloneqq\sum_{x_1}p_1(x_1)\rho_{x_1,x_2},\qquad
 \overline\rho_{x_1}\coloneqq\sum_{x_2}p_2(x_2)\rho_{x_1,x_2},\qquad
 \overline\rho\coloneqq\sum_{x_1,x_2}p_1(x_1)p_2(x_2)\rho_{x_1,x_2}.
\end{equation}
Choose projectors $P_{s,x_1,x_2}$, for $s\in\{1,2,12\}$, before drawing the codebook. Define
\begin{equation}\label{eq:two-alpha}
 \alpha_s\coloneqq\sum_{x_1,x_2}p_1(x_1)p_2(x_2)
       \Tr[(I-P_{s,x_1,x_2})\rho_{x_1,x_2}]
\end{equation}
and
\begin{align}
\beta_1&\coloneqq\sum_{x_1,x_2}p_1(x_1)p_2(x_2)
                  \Tr[P_{1,x_1,x_2}\overline\rho_{x_2}],\label{eq:two-betas}\\
\beta_2&\coloneqq\sum_{x_1,x_2}p_1(x_1)p_2(x_2)
                  \Tr[P_{2,x_1,x_2}\overline\rho_{x_1}],\label{eq:line-5-3-2}\\
\beta_{12}&\coloneqq\sum_{x_1,x_2}p_1(x_1)p_2(x_2)
                  \Tr[P_{12,x_1,x_2}\overline\rho].\label{eq:line-5-3-3}
\end{align}
These quantities are averages over candidate inputs. They need not be small for every individual input pair.

Generate all $x_i(m_i)$ independently according to $p_i$, and set
\begin{equation}\label{eq:codebook-tests}
 \rho_m\coloneqq\rho_{x_1(m_1),x_2(m_2)},\qquad
 P_{s,m}\coloneqq P_{s,x_1(m_1),x_2(m_2)}.
\end{equation}
For each realization, apply \thref{thm:two-packing}. The following proposition evaluates the expectation of its right-hand side.

\begin{proposition}[Two-sender random coding]\label{prop:two-random}
Let $(x_1,x_2)\mapsto\rho_{x_1,x_2}$ be the finite-alphabet, finite-dimensional channel in \eqref{eq:two-channel}, with product input distribution $p_1p_2$. Choose projectors $P_{s,x_1,x_2}$ for $s\in\{1,2,12\}$ and define their average errors $\alpha_s,\beta_s$ by \eqref{eq:two-alpha} and \eqref{eq:two-betas}--\eqref{eq:line-5-3-3}. For integers $M_1,M_2\ge1$, sample every $x_i(m_i)$ independently according to $p_i$, as in \eqref{eq:codebook-tests}. For every $0<\delta<1$, there are decoding POVMs depending on these codebooks whose uniform-message average error $\Pe$, defined by \eqref{eq:ensemble-error} with $\rho_m\coloneqq\rho_{x_1(m_1),x_2(m_2)}$, satisfies
\begin{equation}\label{eq:two-random}
\begin{aligned}
\ee_{\cC}\Pe
 \le {}&d_2(\delta)+C_2(\delta)(\alpha_1+\alpha_2+\alpha_{12})\\
 &+\bigl[(M_1-1)\beta_1+(M_2-1)\beta_2
            +(M_1-1)(M_2-1)\beta_{12}\bigr].
\end{aligned}
\end{equation}
The expectation is over the independent codebook samples, and $d_2(\delta),C_2(\delta)$ are defined in \eqref{eq:two-constants}.
\end{proposition}
\begin{proof}
For a fixed transmitted message pair $m$, its input pair has distribution $p_1p_2$. Thus, by \eqref{eq:two-a} and \eqref{eq:two-alpha},
\begin{equation}
 \ee_{\cC}a=\alpha_1+\alpha_2+\alpha_{12}.
\end{equation}
For the first error class, fix $\ell_1\ne m_1$ and set $\ell_2\coloneqq m_2$. The candidate input $x_1(\ell_1)$, the true input $x_1(m_1)$, and their shared input $x_2(m_2)$ are independent random variables. Therefore
\begin{equation}\label{eq:average-type1}
\begin{aligned}
\ee_{\cC}\Tr[P_{1,(\ell_1,m_2)}\rho_m]& =\sum_{x_1,x_2,y_1}p_1(x_1)p_2(x_2)p_1(y_1)
       \Tr[P_{1,x_1,x_2}\rho_{y_1,x_2}]\\
 &=\sum_{x_1,x_2}p_1(x_1)p_2(x_2)
       \Tr[P_{1,x_1,x_2}\overline\rho_{x_2}]=\beta_1.
\end{aligned}
\end{equation}
Similarly, for $\ell_1=m_1$ and $\ell_2\ne m_2$, the relevant independent variables give
\begin{equation}\label{eq:average-type2}
\begin{aligned}
\ee_{\cC}\Tr[P_{2,(m_1,\ell_2)}\rho_m]
 &=\sum_{x_1,x_2,y_2}p_1(x_1)p_2(x_2)p_2(y_2)
       \Tr[P_{2,x_1,x_2}\rho_{x_1,y_2}]\\
 &=\sum_{x_1,x_2}p_1(x_1)p_2(x_2)
       \Tr[P_{2,x_1,x_2}\overline\rho_{x_1}]=\beta_2.
\end{aligned}
\end{equation}
When both indices differ, all four input variables are independent, so
\begin{equation}\label{eq:average-type12}
\begin{aligned}
\ee_{\cC}\Tr[P_{12,(\ell_1,\ell_2)}\rho_m]
 &=\sum_{x_1,x_2,y_1,y_2}p_1(x_1)p_2(x_2)p_1(y_1)p_2(y_2)
       \Tr[P_{12,x_1,x_2}\rho_{y_1,y_2}]\\
 &=\sum_{x_1,x_2}p_1(x_1)p_2(x_2)
       \Tr[P_{12,x_1,x_2}\overline\rho]=\beta_{12}.
\end{aligned}
\end{equation}
The numbers of competitors in these classes are $M_1-1$, $M_2-1$, and $(M_1-1)(M_2-1)$; see \fref{fig:grid}. Linearity of expectation now yields
\begin{equation}
 \ee_{\cC}b=(M_1-1)\beta_1+(M_2-1)\beta_2
                 +(M_1-1)(M_2-1)\beta_{12}.
\end{equation}
Taking the expectation of \eqref{eq:two-packing} proves the claim. Distinct competitors are generally dependent, but independence between them was not used.
\end{proof}

\begin{figure}[!t]
\centering
\begin{adjustbox}{max width=\linewidth}
\begin{tikzpicture}[font=\small,x=.9cm,y=.75cm]
\foreach \i in {1,...,4}{\foreach \j in {1,...,5}{
\def\cellcolor{softorange}
\ifnum\i=2\def\cellcolor{softblue}\fi
\ifnum\j=3\def\cellcolor{softgreen}\fi
\ifnum\i=2\ifnum\j=3\def\cellcolor{white}\fi\fi
\fill[\cellcolor] (\j-1,\i-1) rectangle (\j,\i);
\draw[gray!55] (\j-1,\i-1) rectangle (\j,\i);
}}
\fill[inkblue] (2.5,1.5) circle (2.2pt);
\node[above right,font=\footnotesize] at (2.5,1.5) {$m$};
\draw[->] (0,-.4)--(5.4,-.4) node[right] {Sender 2};
\draw[->] (-.4,0)--(-.4,4.45) node[above] {Sender 1};
\node[below] at (2.5,-.42) {$\ell_2=m_2$};
\node[anchor=east] at (-.5,1.5) {$\ell_1=m_1$};
\fill[softgreen] (6.2,3.1) rectangle (6.65,3.55);
\node[anchor=west] at (6.85,3.32) {Only first message differs: $M_1-1$};
\fill[softblue] (6.2,1.9) rectangle (6.65,2.35);
\node[anchor=west] at (6.85,2.12) {Only second message differs: $M_2-1$};
\fill[softorange] (6.2,.7) rectangle (6.65,1.15);
\node[anchor=west,align=left] at (6.85,.92) {Both messages differ:\\$(M_1-1)(M_2-1)$};
\end{tikzpicture}
\end{adjustbox}
\caption{The three disjoint classes of competing message pairs for a fixed true pair $m$. Rows index Sender~1 and columns index Sender~2. The illustrated $4\times5$ grid has $3$, $4$, and $12$ competitors in the respective classes. The random-codebook average of each term in a class is the same binary false-acceptance probability. No independence between different competitors is required.}
\label{fig:grid}
\end{figure}
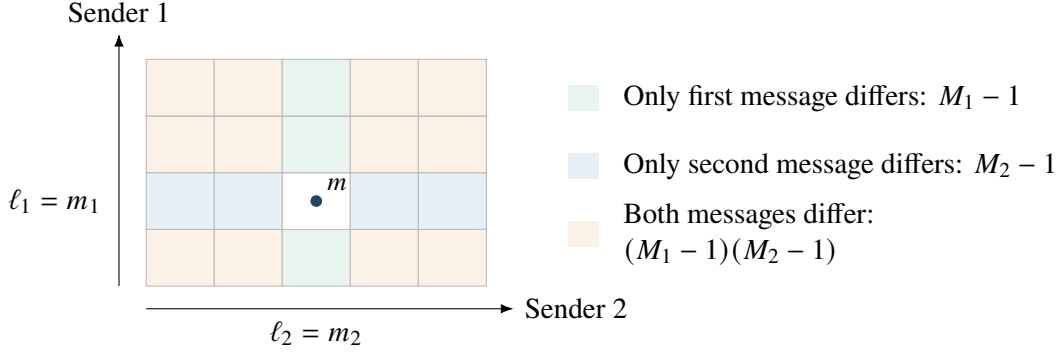

\begin{remark}[Repeated codewords]
Independently sampled codewords need not have distinct values. Equations~\eqref{eq:average-type1}--\eqref{eq:average-type12} use independence of the random variables assigned to different message indices, not an assumption that their realizations differ. The auxiliary label identity uses those indices themselves and thus remains exact in the presence of repeated codewords.
\end{remark}

\subsection{The three binary alternatives}
The three false-acceptance quantities are ordinary binary testing probabilities for the null state $\omega_{X_1X_2B}$ in \eqref{eq:two-omega}. Define alternative states, with tensor factors placed in the order $X_1X_2B$, by
\begin{align}
\sigma^{(1)}_{X_1X_2B}&\coloneqq\omega_{X_1}\otimes\omega_{X_2B},\label{eq:two-sigma}\\
\sigma^{(2)}_{X_1X_2B}&\coloneqq\omega_{X_2}\otimes\omega_{X_1B},\label{eq:line-5-11-2}\\
\sigma^{(12)}_{X_1X_2B}&\coloneqq\omega_{X_1}\otimes\omega_{X_2}\otimes\omega_B.\label{eq:line-5-11-3}
\end{align}
The notation in the second line includes the canonical reordering of tensor factors. Their classical blocks are
\begin{align}
\sigma^{(1)}_{X_1X_2B}
 &=\sum_{x_1,x_2}p_1(x_1)p_2(x_2)\proj{x_1,x_2}\otimes\overline\rho_{x_2},\label{eq:two-sigma-blocks}\\
\sigma^{(2)}_{X_1X_2B}
 &=\sum_{x_1,x_2}p_1(x_1)p_2(x_2)\proj{x_1,x_2}\otimes\overline\rho_{x_1},\label{eq:line-5-12-2}\\
\sigma^{(12)}_{X_1X_2B}
 &=\sum_{x_1,x_2}p_1(x_1)p_2(x_2)\proj{x_1,x_2}\otimes\overline\rho.\label{eq:line-5-12-3}
\end{align}
A block-diagonal test $T_s=\sum_{x_1,x_2}\proj{x_1,x_2}\otimes P_{s,x_1,x_2}$ therefore has errors exactly $\alpha_s$ and $\beta_s$ from \eqref{eq:two-alpha}--\eqref{eq:line-5-3-3}.

The relative entropies of these pairs are the desired three information quantities. The logarithm of a tensor product is the sum of its factor logarithms on the support. Hence
\begin{align}
D(\omega\Vert\sigma^{(1)})
 &=-H(X_1X_2B)_\omega+H(X_1)_\omega+H(X_2B)_\omega\label{eq:two-divergences}\\
&=H(X_1X_2)_\omega+H(X_2B)_\omega-H(X_2)_\omega-H(X_1X_2B)_\omega\label{eq:line-5-13-2}\\
&=I(X_1;B|X_2)_\omega,\label{eq:line-5-13-3}\\
D(\omega\Vert\sigma^{(2)})
 &=-H(X_1X_2B)_\omega+H(X_2)_\omega+H(X_1B)_\omega\label{eq:line-5-13-4}\\
&=H(X_1X_2)_\omega+H(X_1B)_\omega-H(X_1)_\omega-H(X_1X_2B)_\omega\label{eq:line-5-13-5}\\
&=I(X_2;B|X_1)_\omega,\label{eq:line-5-13-6}\\
D(\omega\Vert\sigma^{(12)})
 &=-H(X_1X_2B)_\omega+H(X_1)_\omega+H(X_2)_\omega+H(B)_\omega\label{eq:line-5-13-7}\\
&=H(X_1X_2)_\omega+H(B)_\omega-H(X_1X_2B)_\omega\label{eq:line-5-13-8}\\
&=I(X_1X_2;B)_\omega.\label{eq:line-5-13-9}
\end{align}
The replacements of $H(X_1)+H(X_2)$ by $H(X_1X_2)$ use independence of the inputs. Support inclusion also holds: for a pair of positive input probability, $\overline\rho_{x_2}\ge p_1(x_1)\rho_{x_1,x_2}$, with analogous inequalities for the other alternatives. Zero-probability classical blocks can be ignored. The relative entropies in \eqref{eq:two-divergences}--\eqref{eq:line-5-13-9} are therefore finite.

\subsection{Projective Stein tests}\label{subsec:stein-tests}
For clarity, we justify both projectivity and the exponential false-acceptance bound. This avoids assuming that an arbitrary optimal fixed-error binary test is a projector.

\begin{lemma}[Projective threshold tests]\label{lem:projective-stein}
Let $\rho,\sigma$ be states on the same finite-dimensional Hilbert space and let $r$ be a real number satisfying $r<D(\rho\Vert\sigma)$. Using the binary error probabilities $\alpha(T;\rho)=\Tr[(I-T)\rho]$ and $\beta(T;\sigma)=\Tr[T\sigma]$ from \eqref{eq:binary-errors}, define, for each positive integer $n$,
\begin{equation}\label{eq:NP-test}
 T_n\coloneqq\{\rho^{\otimes n}-2^{nr}\sigma^{\otimes n}\ge0\},
\end{equation}
where $\{Z\ge0\}$ denotes the spectral projector onto the nonnegative eigenspaces of a Hermitian operator $Z$. Then
\begin{equation}\label{eq:NP-errors}
 \alpha(T_n;\rho^{\otimes n})\longrightarrow0,
 \qquad
 \beta(T_n;\sigma^{\otimes n})\le2^{-nr}.
\end{equation}
If both states are block diagonal in the same classical basis, $T_n$ is block diagonal in the corresponding tensor-product classical basis.
\end{lemma}
\begin{proof}
Set $t\coloneqq 2^{nr}$. Since $T_n$ is the positive spectral test,
\begin{equation}
 t\Tr[T_n\sigma^{\otimes n}]
 \le\Tr[T_n\rho^{\otimes n}]\le1,
\end{equation}
which proves the second claim. To prove the first, choose $r'$ satisfying $r<r'<D(\rho\Vert\sigma)$ and apply \eqref{eq:stein-standard} at rate $r'$ to obtain effects $U_n$.

For an arbitrary effect $T$,
\begin{equation}
 \alpha(T;\rho^{\otimes n})+t\beta(T;\sigma^{\otimes n})
 =1-\Tr[T(\rho^{\otimes n}-t\sigma^{\otimes n})].
\end{equation}
The spectral projector $T_n$ maximizes the last trace: in an eigenbasis of $\rho^{\otimes n}-t\sigma^{\otimes n}$, each diagonal entry of an effect is in $[0,1]$, and accepting precisely the nonnegative eigenvalues is optimal. Thus
\begin{align}
\alpha(T_n;\rho^{\otimes n})
 &\le\alpha(T_n;\rho^{\otimes n})+t\beta(T_n;\sigma^{\otimes n})\label{eq:line-5-18-1}\\
&\le\alpha(U_n;\rho^{\otimes n})+t\beta(U_n;\sigma^{\otimes n})\label{eq:line-5-18-2}\\
&\le\alpha(U_n;\rho^{\otimes n})+2^{-n(r'-r)}\longrightarrow0.\label{eq:line-5-18-3}
\end{align}
Finally, a spectral projector of a block-diagonal Hermitian operator is block diagonal. Each of its classical blocks is itself a projector.
\end{proof}

\subsection{Achieving the strictly interior rates}
Suppose first that $R_1,R_2>0$ and the three inequalities in \eqref{eq:two-region}--\eqref{eq:line-3-5-3} are strict. Choose thresholds $r_1,r_2,r_{12}$ satisfying
\begin{align}
R_1&<r_1<I(X_1;B|X_2)_\omega,\label{eq:two-thresholds}\\
R_2&<r_2<I(X_2;B|X_1)_\omega,\label{eq:line-5-19-2}\\
R_1+R_2&<r_{12}<I(X_1X_2;B)_\omega.\label{eq:line-5-19-3}
\end{align}
Apply \lemref{lem:projective-stein} separately to $\omega^{\otimes n}$ versus $(\sigma^{(s)})^{\otimes n}$ at threshold $r_s$. The resulting test has the form
\begin{equation}\label{eq:two-block-tests}
 T_{s,n}=\sum_{x_1^n,x_2^n}\proj{x_1^n,x_2^n}
                      \otimes P^{(n)}_{s,x_1^n,x_2^n}.
\end{equation}
The classical blocks give exactly the projectors required for \prref{prop:two-random} applied to the block channel.

To check the alternatives explicitly, independence over channel uses gives
\begin{align}
\sum_{x_1^n}p_1^{\otimes n}(x_1^n)
       \bigotimes_{t=1}^n\rho_{x_{1,t},x_{2,t}}
 &=\bigotimes_{t=1}^n\left(\sum_{x_{1,t}}p_1(x_{1,t})\rho_{x_{1,t},x_{2,t}}\right)\label{eq:block-average}\\
&=\bigotimes_{t=1}^n\overline\rho_{x_{2,t}}.\label{eq:line-5-21-2}
\end{align}
The other conditional average and the full average factor in the same way, by summing over $x_2^n$ or both input strings. Accordingly, the errors in the block version of \eqref{eq:two-alpha}--\eqref{eq:line-5-3-3} satisfy
\begin{equation}\label{eq:two-block-errors}
 \alpha_{s,n}\longrightarrow0,
 \qquad\beta_{s,n}\le2^{-nr_s},
 \qquad s\in\{1,2,12\}.
\end{equation}
The three tests were chosen separately. We have not asserted the existence of one common binary test on $X_1^nX_2^nB^n$ satisfying all three bounds.

Take $M_i\coloneqq\lfloor2^{nR_i}\rfloor$ and write $R_{12}\coloneqq R_1+R_2$, with $s\in\{1,2,12\}$ in the sums below. Since $M_i-1\le2^{nR_i}$, \prref{prop:two-random} gives
\begin{equation}\label{eq:two-final-asymptotic}
 \ee_{\cC}\Pe^{(n)}
 \le d_2(\delta)+C_2(\delta)\sum_s\alpha_{s,n}
       +\sum_s2^{-n(r_s-R_s)}.
\end{equation}
For every fixed $\delta>0$, all terms except $d_2(\delta)$ tend to zero. Given $\eps>0$, first choose $\delta$ so that $d_2(\delta)<\eps/2$, and then choose $n$ sufficiently large so that the remaining terms are below $\eps/2$. This proves that, for every sufficiently large blocklength, the ensemble has a decoder with expected average error at most $\eps$.

For completeness, one may select a single vanishing-error sequence as follows. For each integer $j\ge2$, choose a fixed $\delta_j$ with $d_2(\delta_j)<1/(2j)$ and a blocklength $N_j$ after which the other terms in \eqref{eq:two-final-asymptotic} are below $1/(2j)$. Increase the $N_j$ if necessary so that they are strictly increasing. For $N_j\le n<N_{j+1}$, use $\delta_j$. The expected average error is at most $1/j$ and hence tends to zero. The codeword distribution remains $p_1^{\otimes n}p_2^{\otimes n}$ for every choice of $\delta_j$.

For each $n$, at least one realized codebook has average error no larger than the ensemble average. Selecting such a codebook gives a deterministic code. Also, $\lfloor2^{nR_i}\rfloor\ge2^{nR_i-1}$ implies
\begin{equation}
 R_i-\frac1n\le\frac1n\log_2 M_i\le R_i.
\end{equation}
Thus the asymptotic rates are the prescribed $R_i$.

\subsection{Zero rates and boundary points}\label{subsec:two-boundary}
The zero-rate case must not be obtained by imposing a strict inequality of the form $0<I=0$. Suppose, for example, that $R_1=0$ and $R_2>0$. Take $M_1=1$. Error types $1$ and $12$ are impossible, so choose their projectors to be $I$. Their correct-message errors are zero, and their false-acceptance contributions have zero multiplicity. Only the test for type $2$ is needed. If $R_2<I(X_2;B|X_1)_\omega$, the preceding argument applies to that test alone. Sender~1's single codeword is still sampled from $p_1^{\otimes n}$ and is known as part of the codebook. The case $R_2=0<R_1$ is identical with the sender roles interchanged. When both rates are zero, each sender has one message and the decoder always outputs that pair.

Now let $(R_1,R_2)$ be an arbitrary nonnegative pair in the closed region \eqref{eq:two-region}--\eqref{eq:line-3-5-3}. For $j\ge2$, set $R_i^{(j)}\coloneqq (1-1/j)R_i$. Every error class that can occur for positive-rate senders has a positive rate sum. Since its mutual-information bound is at least that sum at the target rate, its inequality is strict at the scaled rate. The previous arguments therefore give, for each $j$, codes of every sufficiently large blocklength with expected average error at most $1/j$ and rates approaching $(R_1^{(j)},R_2^{(j)})$. Choose increasing blocklength thresholds and use the $j$th construction until the next threshold. The error tends to zero and the rates tend to $(R_1,R_2)$. This is a choice of a sequence of whole-block codes, not time sharing within a code. It completes the proof of \thref{thm:two-capacity}.
\FloatBarrier
\section{An arbitrary fixed number of senders}\label{sec:k}
The two-sender proof used no property specific to three tests. Its essential ingredients were an exact coordinatewise projection identity and a uniformly positive Gram matrix. We now verify both ingredients for $k$ senders and carry out the entire packing and random-coding calculation.

\subsection{Channel and coding theorem}
Fix a positive integer $k$, independent of the blocklength. Write $[k]\coloneqq \{1,\ldots,k\}$ and $S^c\coloneqq[k]\setminus S$. For $S\subseteq[k]$, the notation $x_S$ or $X_S$ denotes the tuple of coordinates or classical registers indexed by $S$, in increasing order. Empty tuples have one value and probability one. Set
\begin{equation}\label{eq:k-products}
 p(x)\coloneqq\prod_{i=1}^kp_i(x_i),\qquad
 p_S(x_S)\coloneqq\prod_{i\in S}p_i(x_i),\qquad
 R(S)\coloneqq\sum_{i\in S}R_i.
\end{equation}
Consider a classical--quantum channel $x=(x_1,\ldots,x_k)\mapsto\rho_x$, where $x_i\in\cX_i$ for finite alphabets $\cX_i$ and $B$ is finite dimensional, and the state
\begin{equation}\label{eq:k-omega}
 \omega_{X_{[k]}B}
 =\sum_xp(x)\proj{x}_{X_{[k]}}\otimes\rho_x.
\end{equation}
A code has independent uniform messages $m_i\in[M_i]$, one codebook $x_i^n(m_i)$ for each sender, and a single POVM indexed by the whole message tuple. Its average error is
\begin{equation}\label{eq:k-code-error}
 \Pe^{(n)}\coloneqq1-\frac{1}{\prod_i M_i}\sum_{m_1,\ldots,m_k}
  \Tr\!\left[\Lambda_m\bigotimes_{t=1}^n
       \rho_{x_{1,t}(m_1),\ldots,x_{k,t}(m_k)}\right].
\end{equation}
Achievability is defined by vanishing average error and $\liminf_n n^{-1}\log_2M_i\ge R_i$ for every $i$.

\begin{theorem}[Simultaneous decoding for $k$ senders; Sen~\cite{Sen2021Joint,Sen2021Network}]\label{thm:k-capacity}
Fix a positive integer $k$. Let $x=(x_1,\ldots,x_k)\mapsto\rho_x$ be a channel with finite input alphabets and finite-dimensional output $B$. Fix the product distribution $p(x)=\prod_i p_i(x_i)$ and the state $\omega_{X_{[k]}B}$ in \eqref{eq:k-omega}. Write $S^c\coloneqq [k]\setminus S$ and $R(S)=\sum_{i\in S}R_i$, as in \eqref{eq:k-products}. With the average-error criterion \eqref{eq:k-code-error}, every nonnegative rate tuple satisfying
\begin{equation}\label{eq:k-region}
 R(S)\le I(X_S;B|X_{S^c})_\omega
 \quad\text{for all }\varnothing\ne S\subseteq[k]
\end{equation}
is achievable with independently generated codebooks and one joint decoding POVM. Sender~$i$'s codewords are sampled independently from $p_i^{\otimes n}$. For positive rates with all inequalities strict, the message sizes may be $M_i=\lfloor2^{nR_i}\rfloor$ and the expected average error tends to zero. Zero rates and boundary points are included by the construction below. No time sharing or rate splitting is used.
\end{theorem}
When $S=[k]$, the right-hand side means $I(X_{[k]};B)_\omega$. We prove the theorem by a deterministic bound that applies even when the states have no channel structure.

\subsection{Deterministic packing for message tuples}
Let $\cM\coloneqq\prod_{i=1}^k[M_i]$, let $M\coloneqq\prod_iM_i$, and let $(\rho_m)_{m\in\cM}$ be arbitrary states on $B$. For distinct tuples $m,\ell$, define their exact error set by
\begin{equation}\label{eq:k-error-set}
 S(m,\ell)\coloneqq\{i\in[k]:m_i\ne\ell_i\}.
\end{equation}
For every nonempty $S\subseteq[k]$ and every candidate message tuple $\ell$, let $P_{S,\ell}$ be a projector on $B$. Define
\begin{align}
a&\coloneqq\frac1M\sum_m\sum_{\varnothing\ne S\subseteq[k]}
                \Tr[(I-P_{S,m})\rho_m],\label{eq:k-ab}\\
b&\coloneqq\frac1M\sum_m\sum_{\ell\ne m}
                \Tr[P_{S(m,\ell),\ell}\rho_m].\label{eq:line-6-6-2}
\end{align}
As for two senders, $a$ and $b$ are sums of binary-testing error probabilities, not necessarily probabilities themselves: $0\le a\le2^k-1$ and $0\le b\le M-1$. Each term in $b$ uses the projector for the full proposed tuple $\ell$ and its exact set of nonmatching coordinates relative to the true tuple~$m$.

\begin{theorem}[Deterministic $k$-sender packing]\label{thm:k-packing}
Let $k$ and $M_1,\ldots,M_k$ be positive integers, let $\cM\coloneqq\prod_i[M_i]$ and $M\coloneqq\prod_iM_i$, and let $(\rho_m)_{m\in\cM}$ be states on a finite-dimensional system $B$. For every $m\in\cM$ and nonempty $S\subseteq[k]$, let $P_{S,m}$ be a projector. Define the exact error set $S(m,\ell)$ by \eqref{eq:k-error-set} and the sums $a,b$ by \eqref{eq:k-ab}--\eqref{eq:line-6-6-2}. For every $0<\delta<1$, there is a decoding POVM on $B$, with an optional failure outcome, whose uniform-message error
$\Pe=1-M^{-1}\sum_m\Tr[\Lambda_m\rho_m]$
(the convention of \eqref{eq:ensemble-error}) satisfies
\begin{equation}\label{eq:k-packing}
 \Pe\le d_k(\delta)+C_k(\delta)a+b,
\end{equation}
where
\begin{equation}\label{eq:k-constants}
 d_k(\delta)\coloneqq\sqrt{1-(1-\delta)^k},
 \qquad
 C_k(\delta)\coloneqq\bigl(1-\sqrt{1-\delta}\bigr)^{-k}.
\end{equation}
In particular, $d_k(\delta)\le\sqrt{k\delta}$ and $C_k(\delta)\le(2/\delta)^k$.
\end{theorem}
The constants in this basic statement are chosen for the transparent tensor-product proof below. They suffice for all first-order rate results. Section~\ref{subsubsec:sharp-packing} gives a stronger finite-error guarantee for the same decoder; Appendix~\ref{app:quantitative} contains its proof. Neither refinement is needed to follow the present argument.

\begin{proof}
For each sender, introduce $L_i$ with basis $\ket0,\ket1,\ldots,\ket{M_i}$, and set $L\coloneqq L_1\cdots L_k$. Define
\begin{equation}\label{eq:k-labels}
 \ket{u_{m_i}}\coloneqq \sqrt{1-\delta}\ket0+\sqrt\delta\ket{m_i},\qquad
 \ket{t_m}_L\coloneqq\bigotimes_{i=1}^k\ket{u_{m_i}},\qquad
 \rhot m\coloneqq\rho_m\otimes\proj{t_m}.
\end{equation}
For a candidate message tuple $\ell$, set
\begin{equation}\label{eq:k-K}
 K_\ell\coloneqq\bigotimes_{i=1}^k(\proj0+\proj{\ell_i}),
\end{equation}
and define, for every $S\subseteq[k]$, the normalized vector
\begin{equation}\label{eq:k-w}
 \ket{w_{\ell,S}}\coloneqq\bigotimes_{i=1}^k\ket{v_i(S)},
 \qquad
 \ket{v_i(S)}\coloneqq
 \begin{cases}
   \ket0,&i\in S,\\
   \ket{u_{\ell_i}},&i\notin S.
 \end{cases}
\end{equation}
As in the two-sender case, $K_\ell$ preserves the reference and candidate directions in each register and discards all other message basis vectors. Its range has dimension $2^k$, independent of the message sizes. For each coordinate, the two cases in \eqref{eq:one-label-projection} still hold. Multiplying the surviving factors gives
\begin{equation}\label{eq:k-exact}
 K_\ell\ket{t_m}=(1-\delta)^{|S(m,\ell)|/2}\ket{w_{\ell,S(m,\ell)}}.
\end{equation}
For $m=\ell$, this identity is interpreted with $S=\varnothing$, and $\ket{w_{\ell,\varnothing}}=\ket{t_\ell}$. For nonempty $S$, the letter $w$ again denotes a wrong-message pattern; only the labels of the matching senders survive unchanged.

Set $Q_{S,\ell}\coloneqq I-P_{S,\ell}$ and construct
\begin{equation}\label{eq:k-bad-Gamma}
 \cW_\ell\coloneqq\sum_{S\ne\varnothing}
       \ran Q_{S,\ell}\otimes\spanop\{\ket{w_{\ell,S}}\},
 \qquad
 \Gamma_\ell\coloneqq I_B\otimes K_\ell-\Pi_{\cW_\ell}.
\end{equation}
The sum defining $\cW_\ell$ is a linear span, exactly as in \eqref{eq:two-bad-explicit}--\eqref{eq:line-4-24-2}: it contains all vectors $\sum_{S\ne\varnothing}\ket{h_S}\otimes\ket{w_{\ell,S}}$ with $\ket{h_S}\in\ran Q_{S,\ell}$. We again remove that span only inside the local space associated with $\ell$. Since every $\ket{w_{\ell,S}}$ lies in $\ran K_\ell$, the containment $\cW_\ell\subseteq\ran(I_B\otimes K_\ell)$ holds. Therefore $\Gamma_\ell$ is a projector and $\Gamma_\ell(I_B\otimes K_\ell)=(I_B\otimes K_\ell)\Gamma_\ell=\Gamma_\ell$.

\paragraph{Incorrect-message acceptance.}
Fix $\ell\ne m$ and let $S\coloneqq S(m,\ell)$. Define the isometry $V_{\ell,S}\ket h\coloneqq\ket h\otimes\ket{w_{\ell,S}}$ and the effect $A_{\ell,S}\coloneqq V_{\ell,S}^\dagger\Gamma_\ell V_{\ell,S}$. The definition of $\cW_\ell$ implies $\Gamma_\ell V_{\ell,S}Q_{S,\ell}=0$. Multiplying by $V_{\ell,S}^\dagger$ gives $A_{\ell,S}Q_{S,\ell}=0$, and taking adjoints gives $Q_{S,\ell}A_{\ell,S}=0$. The isometry also implies $0\le A_{\ell,S}\le I_B$. Expanding $I_B=P_{S,\ell}+Q_{S,\ell}$ on both sides, exactly as in \eqref{eq:A-support-expansion}--\eqref{eq:line-4-48-3}, therefore yields
\begin{equation}\label{eq:k-A-bound}
 0\le A_{\ell,S}=P_{S,\ell}A_{\ell,S}P_{S,\ell}
 \le P_{S,\ell}I_BP_{S,\ell}=P_{S,\ell}.
\end{equation}
The label identity \eqref{eq:k-exact} gives
\begin{equation}\label{eq:k-projected-label}
 K_\ell\proj{t_m}K_\ell=(1-\delta)^{|S|}\proj{w_{\ell,S}}.
\end{equation}
Using the support of $\Gamma_\ell$ and substituting this identity gives
\begin{align}
\Tr[\Gamma_\ell\rhot m]
 &=\Tr[\Gamma_\ell(I_B\otimes K_\ell)\rhot m(I_B\otimes K_\ell)]\label{eq:k-wrong}\\
&=\Tr[\Gamma_\ell(\rho_m\otimes K_\ell\proj{t_m}K_\ell)]\label{eq:line-6-16-2}\\
&=(1-\delta)^{|S|}\Tr[\Gamma_\ell(\rho_m\otimes\proj{w_{\ell,S}})]\label{eq:line-6-16-3}\\
&=(1-\delta)^{|S|}\Tr[V_{\ell,S}^\dagger\Gamma_\ell V_{\ell,S}\rho_m]\label{eq:line-6-16-4}\\
&=(1-\delta)^{|S|}\Tr[A_{\ell,S}\rho_m]\label{eq:line-6-16-5}\\
&\le(1-\delta)^{|S|}\Tr[P_{S,\ell}\rho_m]
 \le\Tr[P_{S,\ell}\rho_m].\label{eq:line-6-16-6}
\end{align}
This establishes the required estimate for every one of the $2^k-1$ possible error sets.

\paragraph{The Gram matrix.}
For two subsets $S,T$, each tensor factor in \eqref{eq:k-w} has inner product $1$ when its two choices agree and $\sqrt{1-\delta}$ when they differ. Thus
\begin{equation}\label{eq:k-overlaps}
 \langle w_{\ell,S}|w_{\ell,T}\rangle=(1-\delta)^{|S\mathbin\triangle T|/2},
\end{equation}
where $S\mathbin\triangle T\coloneqq(S\setminus T)\cup(T\setminus S)$ is the symmetric difference: the set of indices belonging to exactly one of $S$ and $T$. Each such index contributes one factor $\sqrt{1-\delta}$ to the inner product, while each remaining index contributes one. In the order of the binary indicator strings of subsets, the Gram matrix of all $2^k$ vectors is
\begin{equation}\label{eq:k-Gram}
 G_{\mathrm{full}}=
 \begin{pmatrix}1&\sqrt{1-\delta}\\\sqrt{1-\delta}&1\end{pmatrix}^{\otimes k}
 \ge\bigl(1-\sqrt{1-\delta}\bigr)^k I_{2^k}.
\end{equation}
The inequality follows because the eigenvalues of a tensor product are products of the eigenvalues of its factors, and the smaller factor eigenvalue is $1-\sqrt{1-\delta}>0$.

The Gram matrix $G_{\mathrm{err}}$ for nonempty subsets is obtained by deleting the row and column for $\varnothing$. For $\ket z\in\cc^{2^k-1}$, insert a zero in that coordinate to obtain $\ket{\widetilde z}$. Then
\begin{equation}
 \langle z|G_{\mathrm{err}}|z\rangle
 =\langle\widetilde z|G_{\mathrm{full}}|\widetilde z\rangle
 \ge\bigl(1-\sqrt{1-\delta}\bigr)^k\norm{\ket{\widetilde z}}^2
 =\bigl(1-\sqrt{1-\delta}\bigr)^k\norm{\ket z}^2.
\end{equation}
Consequently, \lemref{lem:gram} gives
\begin{equation}\label{eq:k-bad-bound}
 \Pi_{\cW_\ell}\le C_k(\delta)
      \sum_{S\ne\varnothing}Q_{S,\ell}\otimes\proj{w_{\ell,S}}.
\end{equation}
\fref{fig:lattice} illustrates the subset structure when $k=3$.

\paragraph{Correct-message acceptance.}
For the correct tuple, $K_m\ket{t_m}=\ket{t_m}$. Also, \eqref{eq:k-w} implies $\langle w_{m,S}|t_m\rangle=(1-\delta)^{|S|/2}$. Therefore
\begin{align}
\Tr[(I-\Gamma_m)\rhot m]
 &=\Tr[\Pi_{\cW_m}(\rho_m\otimes\proj{t_m})]\label{eq:k-correct}\\
&\le C_k(\delta)\sum_{S\ne\varnothing}
       \Tr[Q_{S,m}\rho_m]\,|\langle w_{m,S}|t_m\rangle|^2\label{eq:line-6-21-2}\\
&=C_k(\delta)\sum_{S\ne\varnothing}(1-\delta)^{|S|}
                       \Tr[(I-P_{S,m})\rho_m]\label{eq:line-6-21-3}\\
&\le C_k(\delta)\sum_{S\ne\varnothing}
                       \Tr[(I-P_{S,m})\rho_m].\label{eq:line-6-21-4}
\end{align}

\paragraph{The state-based measurement and its error.}
Set
\begin{equation}\label{eq:k-state-pgm}
 \widetilde Z\coloneqq \sum_m\rhot m,\qquad
 \widetilde\Lambda_m\coloneqq \widetilde Z^{-1/2}\rhot m\widetilde Z^{-1/2},
 \qquad\widetilde\Lambda_{\mathrm f}=I_{BL}-\Pi_{\widetilde Z}.
\end{equation}
This is the pretty-good measurement of the tagged states, completed on the kernel. For each $m$, define $B_m\coloneqq \sum_{\ell\ne m}\rhot\ell$. Since $\widetilde Z=\rhot m+B_m$, \lemref{lem:cheng} with the analysis test $\Gamma_m$ gives
\begin{equation}\label{eq:k-cheng-error}
 1-\Tr[\widetilde\Lambda_m\rhot m]
 \le\Tr[(I_{BL}-\Gamma_m)\rhot m]
     +\sum_{\ell\ne m}\Tr[\Gamma_m\rhot\ell].
\end{equation}
After averaging uniformly over the message tuples, interchange $m,\ell$ in the double sum. Equations~\eqref{eq:k-correct}--\eqref{eq:line-6-21-4} and~\eqref{eq:k-wrong}--\eqref{eq:line-6-16-6} then yield
\begin{align}
\widetilde\Pe
 &\le\frac1M\sum_m\Tr[(I_{BL}-\Gamma_m)\rhot m]
   +\frac1M\sum_m\sum_{\ell\ne m}\Tr[\Gamma_\ell\rhot m]\label{eq:k-tagged-error}\\
&\le\frac{C_k(\delta)}M\sum_m\sum_{S\ne\varnothing}
      \Tr[(I_B-P_{S,m})\rho_m]
   +\frac1M\sum_m\sum_{\ell\ne m}\Tr[P_{S(m,\ell),\ell}\rho_m]\label{eq:line-6-24-2}\\
&=C_k(\delta)a+b.\label{eq:line-6-24-3}
\end{align}
The equality uses \eqref{eq:k-ab}--\eqref{eq:line-6-6-2}. As in the two-sender case, the measurement depends on the tagged states, not on the selected binary tests.

\paragraph{Compression to the original output.}
Define $J\ket h\coloneqq \ket h\otimes\ket{0^k}$ and compress every effect by $J^\dagger(\cdot)J$. Positivity is preserved, and the effects sum to $J^\dagger I_{BL}J=I_B$. The success probability on $\rho_m$ is the success probability of $\widetilde\Lambda_m$ on $\rho_m\otimes\proj{0^k}$. The two compared states have normalized trace distance
\begin{align}
\Delta(\rho_m\otimes\proj{t_m},\rho_m\otimes\proj{0^k})
 &=\sqrt{1-|\langle0^k|t_m\rangle|^2}\label{eq:line-6-25-1}\\
&=\sqrt{1-(1-\delta)^k}=d_k(\delta).\label{eq:line-6-25-2}
\end{align}
To make the last comparison explicit, denote the compressed success effect by $\Lambda_m\coloneqq J^\dagger\widetilde\Lambda_mJ$. Cyclicity of the trace and the definition of $J$ give
\begin{align}
 \Tr[\Lambda_m\rho_m]
 &=\Tr[\widetilde\Lambda_mJ\rho_mJ^\dagger]\label{eq:k-physical-success}\\
 &=\Tr[\widetilde\Lambda_m(\rho_m\otimes\proj{0^k})].\label{eq:k-physical-success-final}
\end{align}
We compare this one completed success effect with its probability on the tagged state. Adding and subtracting that probability, and then applying \eqref{eq:effect-distance}, gives
\begin{align}
 1-\Tr[\Lambda_m\rho_m]
 &=1-\Tr[\widetilde\Lambda_m\rhot m]
   +\Tr[\widetilde\Lambda_m(\rhot m-\rho_m\otimes\proj{0^k})]\label{eq:k-transfer-expanded}\\
 &\le1-\Tr[\widetilde\Lambda_m\rhot m]+d_k(\delta).\label{eq:k-transfer-expanded-final}
\end{align}
The same constant applies to every message. Averaging uniformly and substituting \eqref{eq:k-tagged-error}--\eqref{eq:line-6-24-3} therefore yields
\begin{equation}\label{eq:k-transfer-average}
 \Pe\le\widetilde\Pe+d_k(\delta)
 \le d_k(\delta)+C_k(\delta)a+b.
\end{equation}
This proves \eqref{eq:k-packing}. The label comparison is used once per true message, not separately for each competing message.

It remains to verify the convenient elementary estimates for the constants. The finite geometric-sum identity gives
\begin{equation}\label{eq:k-d-elementary}
 1-(1-\delta)^k
 =\delta\sum_{j=0}^{k-1}(1-\delta)^j\le k\delta,
 \qquad d_k(\delta)\le\sqrt{k\delta},
\end{equation}
because each of the $k$ summands is at most one. For $C_k$, multiply the numerator and denominator by $1+\sqrt{1-\delta}$ to eliminate the square root from the numerator:
\begin{align}
 1-\sqrt{1-\delta}
 &=\frac{(1-\sqrt{1-\delta})(1+\sqrt{1-\delta})}
          {1+\sqrt{1-\delta}}\label{eq:k-C-difference-squares}\\
 &=\frac{1-(1-\delta)}{1+\sqrt{1-\delta}}
 =\frac{\delta}{1+\sqrt{1-\delta}}\ge\frac\delta2.\label{eq:k-C-difference-squares-final}
\end{align}
All quantities are positive. Taking reciprocals reverses the last inequality, and taking the $k$th power gives
\begin{equation}\label{eq:k-C-elementary}
 C_k(\delta)=\left(\frac{1}{1-\sqrt{1-\delta}}\right)^k
 \le\left(\frac2\delta\right)^k.
\end{equation}
\end{proof}

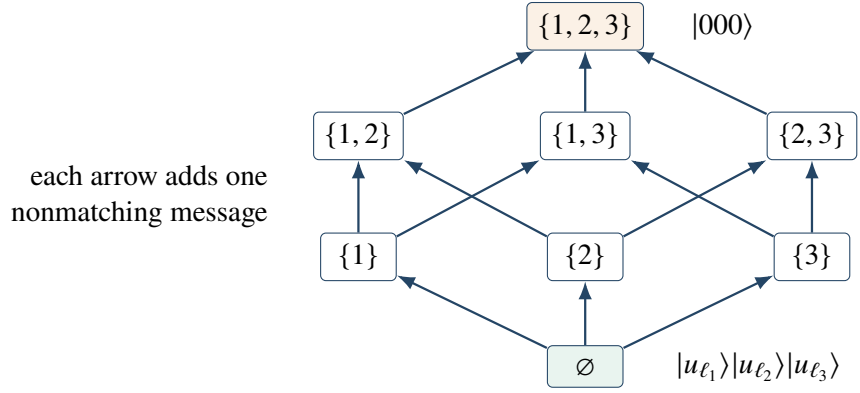
\begin{figure}[!t]
\centering
\begin{adjustbox}{max width=\linewidth}
\begin{tikzpicture}[font=\small]
\tikzset{subset/.style={draw=inkblue,fill=white,rounded corners=2pt,minimum width=1.0cm,minimum height=.55cm}}
\node[subset,fill=softgreen] (e) at (0,0) {$\varnothing$};
\node[subset] (a) at (-3,1.45) {$\{1\}$};
\node[subset] (b) at (0,1.45) {$\{2\}$};
\node[subset] (c) at (3,1.45) {$\{3\}$};
\node[subset] (ab) at (-3,3.05) {$\{1,2\}$};
\node[subset] (ac) at (0,3.05) {$\{1,3\}$};
\node[subset] (bc) at (3,3.05) {$\{2,3\}$};
\node[subset,fill=softorange] (abc) at (0,4.5) {$\{1,2,3\}$};
\begin{scope}[on background layer]
\foreach \x in {a,b,c}{\draw[flow] (e)--(\x);}
\draw[flow] (a)--(ab);\draw[flow] (a)--(ac);
\draw[flow] (b)--(ab);\draw[flow] (b)--(bc);
\draw[flow] (c)--(ac);\draw[flow] (c)--(bc);
\foreach \x in {ab,ac,bc}{\draw[flow] (\x)--(abc);}
\end{scope}
\node[anchor=west] at (1.05,0) {$\ket{u_{\ell_1}}\ket{u_{\ell_2}}\ket{u_{\ell_3}}$};
\node[anchor=west] at (1.25,4.5) {$\ket{000}$};
\node[anchor=east,align=right] at (-4.05,2.25) {each arrow adds one\\nonmatching message};
\end{tikzpicture}
\end{adjustbox}
\caption{Error patterns for three senders, illustrating the general tensor-product structure. A node is the set $S$ of senders whose proposed messages differ from the true messages; its vector is $\ket{w_{\ell,S}}$. Along an arrow, one factor $\ket{u_{\ell_i}}$ is replaced by $\ket0$. The eight vectors have Gram matrix $\left(\begin{smallmatrix}1&\sqrt{1-\delta}\\\sqrt{1-\delta}&1\end{smallmatrix}\right)^{\otimes3}$. The seven nonempty patterns form a principal submatrix. No new geometric mechanism is needed beyond the two-sender case.}
\label{fig:lattice}
\end{figure}

\subsection{Random-codebook averaging for every error set}
Return to the channel $x\mapsto\rho_x$ and the product distribution $p$. For nonempty $S\subseteq[k]$, define
\begin{equation}\label{eq:k-average}
 \overline\rho_{x_{S^c}}\coloneqq\sum_{x_S}p_S(x_S)\rho_x.
\end{equation}
For $S=[k]$, this is the unconditional average output. Fix projectors $P_{S,x}$ before drawing the codebooks, and define
\begin{equation}\label{eq:k-test-errors}
 \alpha_S\coloneqq\sum_xp(x)\Tr[(I-P_{S,x})\rho_x],
 \qquad
 \beta_S\coloneqq\sum_xp(x)\Tr[P_{S,x}\overline\rho_{x_{S^c}}].
\end{equation}
Independently draw every $x_i(m_i)$ according to $p_i$ and set $P_{S,\ell}\coloneqq P_{S,x(\ell)}$, $\rho_m\coloneqq\rho_{x(m)}$. The following is the general version of \prref{prop:two-random}.

\begin{proposition}[$k$-sender random coding]\label{prop:k-random}
Let $x\mapsto\rho_x$ be a channel with $k$ finite classical input alphabets and finite-dimensional output $B$, and fix $p(x)=\prod_i p_i(x_i)$ as in \eqref{eq:k-products}. Choose projectors $P_{S,x}$ for every nonempty $S\subseteq[k]$ and define $\alpha_S,\beta_S$ by \eqref{eq:k-test-errors}, using the conditional averages in \eqref{eq:k-average}. For positive integers $M_i$, sample every $x_i(m_i)$ independently according to $p_i$. For every $0<\delta<1$, there are decoding POVMs depending on the realized codebooks whose uniform-message error $\Pe=1-(\prod_iM_i)^{-1}\sum_m\Tr[\Lambda_m\rho_{x(m)}]$ satisfies
\begin{equation}\label{eq:k-random}
 \ee_{\cC}\Pe\le d_k(\delta)
    +C_k(\delta)\sum_{S\ne\varnothing}\alpha_S
    +\sum_{S\ne\varnothing}N_S\beta_S,
 \qquad N_S\coloneqq\prod_{i\in S}(M_i-1).
\end{equation}
The sums run over nonempty subsets of $[k]$, the expectation is over the codebooks, and $d_k,C_k$ are defined in \eqref{eq:k-constants}.
\end{proposition}
\begin{proof}
Each true input tuple has distribution $p$, so $\ee a=\sum_{S\ne\varnothing}\alpha_S$. Fix $m,\ell$ with exact error set $S$. Conditional on the candidate tuple $x(\ell)=x$, every differing true coordinate $x_i(m_i)$ is an independent fresh sample from $p_i$, while the coordinates in $S^c$ equal $x_{S^c}$. Consequently,
\begin{equation}\label{eq:k-pair-average}
\begin{aligned}
\ee_{\cC}\Tr[P_{S,x(\ell)}\rho_{x(m)}]
 &=\sum_xp(x)\sum_{y_S}p_S(y_S)
             \Tr[P_{S,x}\rho_{y_S,x_{S^c}}]\\
 &=\sum_xp(x)\Tr[P_{S,x}\overline\rho_{x_{S^c}}]
 =\beta_S.
\end{aligned}
\end{equation}
The first equality is a finite sum over independent draws, so it does not require conditioning on a zero-probability value. For fixed $m$, each $i\in S$ has $M_i-1$ allowed differing indices, and every $i\notin S$ has exactly one allowed index. Thus there are precisely $N_S$ competitors with exact error set $S$. Summing \eqref{eq:k-pair-average} over all candidates and then over $m$ gives $\ee b=\sum_S N_S\beta_S$. Taking the expectation of \eqref{eq:k-packing} proves the proposition.
\end{proof}

\subsection{Binary alternatives and the asymptotic limit}
The random-codebook calculation has reduced decoding error to separate correct-rejection and false-acceptance probabilities. We now identify the binary hypotheses that realize these probabilities. Their relative entropies will give the desired conditional-mutual-information rates, and the binary Stein lemma will make each false-acceptance probability decay faster than the number of corresponding competitors grows. After establishing these identities, we take the blocklength limit and include zero-rate senders and boundary points.

For a nonempty error set $S$, define the alternative state
\begin{equation}\label{eq:k-sigma}
 \sigma^{(S)}_{X_{[k]}B}
 \coloneqq\omega_{X_S}\otimes\omega_{X_{S^c}B}
 =\sum_xp(x)\proj{x}\otimes\overline\rho_{x_{S^c}},
\end{equation}
with the tensor factors reordered canonically. Independence of the inputs is needed for the block expression. Taking the logarithm of the product in the first expression gives
\begin{align}
D(\omega\Vert\sigma^{(S)})
 &=-H(X_{[k]}B)_\omega+H(X_S)_\omega+H(X_{S^c}B)_\omega\label{eq:k-relative-entropy}\\
&=H(X_{[k]})_\omega+H(X_{S^c}B)_\omega
      -H(X_{S^c})_\omega-H(X_{[k]}B)_\omega\label{eq:line-6-32-2}\\
&=I(X_S;B|X_{S^c})_\omega.\label{eq:line-6-32-3}
\end{align}
Here $H(X_{[k]})=H(X_S)+H(X_{S^c})$ by product inputs. For every $x$ with $p(x)>0$, \eqref{eq:k-average} implies $\overline\rho_{x_{S^c}}\ge p_S(x_S)\rho_x$. This verifies the support condition and finiteness of the relative entropy.

We now prove \thref{thm:k-capacity}, including zero-rate senders. Given a rate tuple, let
\begin{equation}\label{eq:active-set}
 J\coloneqq\{i:R_i>0\}
\end{equation}
be the set of positive-rate senders. First suppose that
\begin{equation}\label{eq:active-strict}
 R(S)<I(X_S;B|X_{S^c})_\omega
 \qquad\text{for every nonempty }S\subseteq J.
\end{equation}
For $i\notin J$, set $M_i\coloneqq 1$; for $i\in J$, set $M_i\coloneqq \lfloor2^{nR_i}\rfloor$. The only possible error sets are nonempty subsets of $J$. If $J$ is empty, the decoder always outputs the unique tuple and has zero error.

For each nonempty $S\subseteq J$, choose
\begin{equation}\label{eq:k-thresholds}
 R(S)<r_S<I(X_S;B|X_{S^c})_\omega.
\end{equation}
Apply \lemref{lem:projective-stein} to $\omega^{\otimes n}$ and $(\sigma^{(S)})^{\otimes n}$ at threshold $r_S$. After grouping the classical registers by sender, the projective test decomposes as
\begin{equation}
 T_{S,n}=\sum_{x_1^n,\ldots,x_k^n}\proj{x_1^n,\ldots,x_k^n}
                 \otimes P^{(n)}_{S,x_1^n,\ldots,x_k^n}.
\end{equation}
For any fixed $x_{S^c}^n$, the alternative block is precisely the conditional average of the block channel, because
\begin{equation}
 \sum_{x_S^n}\left(\prod_{i\in S}p_i^{\otimes n}(x_i^n)\right)
      \bigotimes_{t=1}^n\rho_{x_t}
 =\bigotimes_{t=1}^n\left(\sum_{x_{S,t}}p_S(x_{S,t})
                              \rho_{x_{S,t},x_{S^c,t}}\right)
 =\bigotimes_{t=1}^n\overline\rho_{x_{S^c,t}}.
\end{equation}
Thus these blocks satisfy
\begin{equation}
 \alpha_{S,n}\longrightarrow0,
 \qquad\beta_{S,n}\le2^{-nr_S}.
\end{equation}
For an impossible error set, namely one not contained in $J$, set every $P^{(n)}_{S,x^n}=I$. Then $\alpha_{S,n}=0$ and $N_S=0$, so that set contributes nothing to the bound.

For the possible sets, $N_S\le\prod_{i\in S}M_i\le2^{nR(S)}$. The block version of \prref{prop:k-random} therefore yields
\begin{equation}\label{eq:k-asymptotic}
 \ee_{\cC}\Pe^{(n)}
 \le d_k(\delta)
 +C_k(\delta)\sum_{\varnothing\ne S\subseteq J}\alpha_{S,n}
 +\sum_{\varnothing\ne S\subseteq J}2^{-n(r_S-R(S))}.
\end{equation}
There are finitely many error sets because $k$ is fixed. Given $\eps>0$, choose $\delta$ so that $d_k(\delta)<\eps/2$. With this $\delta$ held fixed, choose $n$ large enough to make the remaining finite sum smaller than $\eps/2$. As in the explicit diagonal construction after \eqref{eq:two-final-asymptotic}, this gives a sequence with expected average error tending to zero. The rate of each positive-rate sender $i$ lies between $R_i-1/n$ and $R_i$, and an inactive sender has rate zero.

Finally, consider a tuple in the closed region \eqref{eq:k-region}. Keep $J$ as in \eqref{eq:active-set}, and scale the rates by $R_i^{(j)}=(1-1/j)R_i$, for $j\ge2$. Every nonempty $S\subseteq J$ obeys
\begin{equation}
 I(X_S;B|X_{S^c})_\omega-R^{(j)}(S)
 \ge R(S)-R^{(j)}(S)=\frac{R(S)}j>0.
\end{equation}
Hence the strict construction applies at each scaled tuple. Choose successively larger blocklength thresholds for error targets $1/j$, and use the $j$th whole-block construction between its threshold and the next. The expected errors tend to zero and the rates tend to the target tuple. At each blocklength, select a codebook whose error is no larger than the expectation. This proves \thref{thm:k-capacity} for all nonnegative rates in the stated region.

\begin{remark}[Typical good codebooks]
The expectation bound also implies a high-probability statement. If the chosen sequence has $\ee_{\cC}\Pe^{(n)}\le\eta_n$ with $\eta_n>0$ and $\eta_n\to0$, Markov's inequality gives
\begin{equation}
 \Pr_{\cC}\{\Pe^{(n)}>\sqrt{\eta_n}\}\le\sqrt{\eta_n}.
\end{equation}
Thus the conclusion is not based on an unusually favorable codebook outside the ordinary independent ensemble.
\end{remark}
\FloatBarrier
\subsection{General effects, one-shot bounds, and nonuniform priors}\label{subsec:general-packing}\label{app:effects}
So far the geometric proof has used projectors, which suffice for asymptotic threshold tests. The next two applications will use general effects and, for source coding, a nonuniform prior. The deterministic and random-coding bounds extend to both settings without changing the geometry. We begin with arbitrary binary effects. A common two-dimensional extension of the output makes this precise and yields a formulation in terms of hypothesis-testing relative entropy.

\subsubsection{A common extension for all effects}
For an effect $0\le E\le I$, set $L_E\coloneqq\sqrt{E(I-E)}$. On the direct sum $\cH_B\oplus\cH_B$, define
\begin{equation}\label{eq:dilation-effect}
 \mathcal P(E)\coloneqq
 \begin{pmatrix}
  E&L_E\\
  L_E&I-E
 \end{pmatrix}.
\end{equation}
The operator $L_E$ is a function of $E$, so it commutes with $E$. Consequently,
\begin{align}
\mathcal P(E)^2
 &=\begin{pmatrix}
 E^2+L_E^2&EL_E+L_E(I-E)\\
 L_EE+(I-E)L_E&L_E^2+(I-E)^2
 \end{pmatrix}\label{eq:line-6-43-1}\\
&=\begin{pmatrix}E&L_E\\L_E&I-E\end{pmatrix}
 =\mathcal P(E).\label{eq:line-6-43-2}
\end{align}
It is self-adjoint and therefore is a projector. The same isometry $J_0\ket h\coloneqq\ket h\otimes\ket0$ works for every effect, under the identification $\cH_B\oplus\cH_B\simeq\cH_B\otimes\cc^2$:
\begin{equation}\label{eq:common-compression}
 J_0^\dagger\mathcal P(E)J_0=E.
\end{equation}
Different effects generally give noncommuting projectors; no simultaneous commutation is asserted or required.

Now replace each state $\rho_m$ by $J_0\rho_mJ_0^\dagger$ and each effect $E_{S,\ell}$ by $\mathcal P(E_{S,\ell})$. For every pair of messages and every error set,
\begin{align}
\Tr[\mathcal P(E_{S,\ell})J_0\rho_mJ_0^\dagger]
 &=\Tr[E_{S,\ell}\rho_m],\label{eq:line-6-45-1}\\
\Tr[(I-\mathcal P(E_{S,m}))J_0\rho_mJ_0^\dagger]
 &=\Tr[(I-E_{S,m})\rho_m].\label{eq:line-6-45-2}
\end{align}
Apply \thref{thm:k-packing} on the extended output space and compress every resulting POVM element with $J_0$. This preserves all actual success probabilities and introduces no additional error. Thus \thref{thm:k-packing}, and by the same averaging calculation \prref{prop:k-random}, hold with arbitrary effects in place of the projectors. The auxiliary two-dimensional space in this paragraph is distinct from the weak-label spaces used in the main proof. The common dilation also leaves the final state-based decoder unchanged after compression. Here is the algebra. Set
\begin{equation}\label{eq:pgm-dilation-covariance}
 Z=\sum_m\rho_m,\qquad \widehat Z=J_0ZJ_0^\dagger.
\end{equation}
If $Z=\sum_{j:\lambda_j>0}\lambda_j\proj{f_j}$ is its positive spectral decomposition, then $\widehat Z$ has the same positive eigenvalues with eigenvectors $J_0\ket{f_j}$. Its inverse on the support therefore obeys
\begin{equation}
 \widehat Z^{-1/2}J_0=J_0Z^{-1/2}.
\end{equation}
Consequently, the compressed outcome effect is
\begin{align}
J_0^\dagger\widehat Z^{-1/2}(J_0\rho_mJ_0^\dagger)
                      \widehat Z^{-1/2}J_0
 &=(J_0^\dagger J_0)Z^{-1/2}\rho_mZ^{-1/2}(J_0^\dagger J_0)\label{eq:line-6-48-1}\\
&=Z^{-1/2}\rho_mZ^{-1/2}.\label{eq:line-6-48-2}
\end{align}
The completed failure effect is preserved as well, because $J_0^\dagger\Pi_{\widehat Z}J_0=\Pi_Z$. Apply this calculation with the identity on the weak-label registers included in the isometry. Thus using a common dilation to analyze arbitrary effects does not introduce a dependence of the final decoder on those effects.

\subsubsection{Separately optimized hypothesis-testing effects}
Projective threshold tests suffice for the asymptotic region, but a finite-blocklength bound should also allow the best binary effect at a prescribed rejection probability. The common dilation just proved permits exactly this choice. We express the resulting bound in terms of the hypothesis-testing relative entropy from \eqref{eq:DH}; the remaining point is to choose its optimizing effects in the classical-block form required by the random-codebook calculation.

For the product-distribution state $\omega$ in \eqref{eq:k-omega} and the alternatives $\sigma^{(S)}$ in \eqref{eq:k-sigma}, choose an optimizing effect $E_S$ for each $D_H^{\eps_S}(\omega\Vert\sigma^{(S)})$. Both hypotheses are block diagonal in the classical registers. Let $\overline{\Delta}$ denote dephasing in their classical basis. Since $\overline{\Delta}$ is self-adjoint for the trace inner product and fixes both hypotheses,
\begin{equation}
 \Tr[\overline{\Delta}(E_S)\omega]=\Tr[E_S\omega],
 \qquad
 \Tr[\overline{\Delta}(E_S)\sigma^{(S)}]=\Tr[E_S\sigma^{(S)}].
\end{equation}
Also $0\le\overline{\Delta}(E_S)\le I$. We may therefore choose the optimizing effect to have the form
\begin{equation}
 E_S=\sum_x\proj{x}\otimes E_{S,x},\qquad 0\le E_{S,x}\le I.
\end{equation}
Its conditional effects satisfy the exact testing relations
\begin{equation}
 \alpha_S\le\eps_S,
 \qquad
 \beta_S=2^{-D_H^{\eps_S}(\omega\Vert\sigma^{(S)})}.
\end{equation}
The effect version of \prref{prop:k-random} now proves the following statement.

\begin{corollary}[One-shot hypothesis-testing bound]\label{cor:one-shot}
Let $x\mapsto\rho_x$ be a finite-alphabet $k$-sender channel with finite-dimensional output $B$, let $p(x)=\prod_i p_i(x_i)$, and let $\omega$ and $\sigma^{(S)}=\omega_{X_S}\otimes\omega_{X_{S^c}B}$ be the states in \eqref{eq:k-omega} and \eqref{eq:k-sigma}. Use the hypothesis-testing relative entropy $D_H^{\eps_S}$ defined in \eqref{eq:DH}. Fix positive integers $M_i$, $0<\delta<1$, and parameters $0\le\eps_S<1$ for every nonempty $S\subseteq[k]$. If every codeword $x_i(m_i)$ is sampled independently from $p_i$, there are decoding POVMs depending on the codebooks whose uniform-message average error $\Pe=1-(\prod_iM_i)^{-1}\sum_m\Tr[\Lambda_m\rho_{x(m)}]$ satisfies
\begin{equation}\label{eq:one-shot-DH}
 \ee_{\cC}\Pe
 \le d_k(\delta)+C_k(\delta)\sum_{S\ne\varnothing}\eps_S
 +\sum_{S\ne\varnothing}
     \left[\prod_{i\in S}(M_i-1)\right]
       2^{-D_H^{\eps_S}(\omega\Vert\sigma^{(S)})}.
\end{equation}
Here all sums run over nonempty subsets of $[k]$, $S^c\coloneqq[k]\setminus S$, the expectation is over the independent codebooks, and $d_k,C_k$ are defined in \eqref{eq:k-constants}.
\end{corollary}
When some error sets are impossible, their effects may instead be chosen to be the identity, and those sets omitted from the two sums. This is the same convention used for zero-rate senders in the main proof. Section~\ref{subsubsec:sharp-packing} and Appendix~\ref{app:combine} give sharper constants and error-sensitive versions for these same binary tests. The basic form shows explicitly that the packing construction can use separately optimized binary effects, even though such effects need not be projectors on the original output space.

The sharpened versions below can be applied with the same individually optimized tests. For independent codebooks, the refined competing-message average is the weighted expression in \eqref{eq:sharp-channel-substitution}; no modification of the coding distribution is required.

\subsubsection{A packing bound for nonuniform message priors}\label{subsubsec:weighted-packing}
Uniform messages are natural for channel coding, but a classical source generally has a nonuniform, correlated distribution. Only the averaging step changes. The preceding common dilation allows arbitrary effects, and the weighted version of \prref{prop:pgm-tests} already includes the required probabilities.

\begin{corollary}[Deterministic packing with arbitrary priors]\label{cor:weighted-packing}
Let $k,M_1,\ldots,M_k$ be positive integers, set $\mathcal M\coloneqq \prod_{i=1}^k[M_i]$, and let $(p_m,\rho_m)_{m\in\mathcal M}$ be an arbitrary ensemble on a finite-dimensional system $B$, and let $0\le E_{S,m}\le I_B$ for every $m$ and every nonempty $S\subseteq[k]$. The prior $p_m$ need not factor across coordinates. Write $S(m,\ell)\coloneqq \{i:m_i\ne\ell_i\}$ and define
\begin{align}
a_p&\coloneqq \sum_m p_m\sum_{S\ne\varnothing}\Tr[(I-E_{S,m})\rho_m],\label{eq:weighted-ab}\\
b_p&\coloneqq \sum_m p_m\sum_{\ell\ne m}\Tr[E_{S(m,\ell),\ell}\rho_m].\label{eq:line-6-53-2}
\end{align}
For each $0<\delta<1$, there is a POVM on $B$ with an optional failure outcome such that
\begin{equation}\label{eq:weighted-packing}
 P_{\mathrm e}=1-\sum_m p_m\Tr[\Lambda_m\rho_m]
 \le d_k(\delta)+C_k(\delta)a_p+b_p,
\end{equation}
where $d_k,C_k$ are given by \eqref{eq:k-constants}. Specifically, use the labels in \eqref{eq:k-labels}, form the PGM of the positive semidefinite operators
\begin{equation}\label{eq:weighted-tagged-states}
 A_m=p_m\rho_m\otimes\proj{t_m},\qquad Z=\sum_m A_m,
\end{equation}
and compress its effects against $\ket{0^k}$ as in \thref{thm:k-packing}. The resulting decoder depends on the ensemble and $\delta$, not on the effects $E_{S,m}$ used to analyze it.
\end{corollary}
\begin{proof}
First assume the effects are projectors and use the same $\Gamma_m$ as in the proof of \thref{thm:k-packing}. The correct-rejection estimate \eqref{eq:k-correct}--\eqref{eq:line-6-21-4} holds for every $m$. The false-acceptance estimate \eqref{eq:k-wrong}--\eqref{eq:line-6-16-6} holds for every ordered pair of distinct tuples. Apply \eqref{eq:weighted-pgm-error} to the tagged ensemble and multiply these two estimates by $p_m$. Summing gives
\begin{equation}\label{eq:weighted-tagged-error}
 \widetilde P_{\mathrm e}
 \le\sum_m p_m\Tr[(I-\Gamma_m)\rhot m]
 +\sum_m p_m\sum_{\ell\ne m}\Tr[\Gamma_\ell\rhot m]
 \le C_k(\delta)a_p+b_p.
\end{equation}
For each $m$, the label-removal comparison costs at most $d_k(\delta)$. Its weighted average is the same constant because $\sum_m p_m=1$. This proves \eqref{eq:weighted-packing}. For general effects, use the common dilation \eqref{eq:dilation-effect}. It preserves both testing sums and the compressed state-based PGM, by the covariance calculation above, which also holds with $Z=\sum_m p_m\rho_m$. Zero-prior states cause no difficulty: their PGM numerators and their contributions to the actual error vanish.
\end{proof}

The conclusion applies to any ensemble indexed by a Cartesian product, not only to channel output ensembles. The next two sections exploit that flexibility. In entanglement-assisted communication, the coordinates select positions of pre-shared entangled states. In distributed compression, they are the source symbols and the prior is their joint distribution. Only the evaluation of the competing-state terms changes.

\subsubsection{Sharper error estimates for the same decoder}\label{subsubsec:sharp-packing}
The preceding bounds were chosen to expose the construction rather than to optimize its constants. We now record a sharper guarantee without changing any label state, analysis projector, or decoding measurement. The first improvement uses the overlaps of the error-label vectors with the correct label, rather than bounding all label directions at once. The second compares a small tagged error with the actual error, rather than adding the trace distance regardless of the tagged error. The full calculations are deferred to Appendix~\ref{app:quantitative}; the basic proof and its first-order applications do not depend on them.

For $0<\delta<1$, define a matrix indexed by nonempty subsets of $[k]$ and its associated coefficient by
\begin{equation}\label{eq:sharp-gram-definition}
 \mathsf L_k(\delta)\coloneqq
 \left[(1-\delta)^{-|S\cap T|}\right]_{\varnothing\ne S,T\subseteq[k]},
 \qquad
 \kappa_k(\delta)\coloneqq\lambda_{\min}(\mathsf L_k(\delta))^{-1}.
\end{equation}
The matrix is positive definite. It is the Gram matrix obtained by rescaling each error-label vector to have overlap one with the correct label. In particular, the rescaling changes the proof's coordinates but not the rejected subspace. An explicit estimate, avoiding any eigenvalue calculation, is
\begin{equation}\label{eq:sharp-simple-coefficient}
 \kappa_k(\delta)
 \le \frac{(k+1)2^{k-2}(1-\delta)}{\delta^{k-1}},
 \qquad
 \kappa_k(\delta)\le(1-\delta)C_k(\delta).
\end{equation}
These statements include $k=1$, for which $\kappa_1(\delta)=1-\delta$.

For $0<d<1$ and $0\le x\le1$, define
\begin{equation}\label{eq:transfer-function}
 \mathcal T_d(x)\coloneqq
 \begin{cases}
 \left(\sqrt{(1-d^2)x}+d\sqrt{1-x}\right)^2,
       &0\le x\le1-d^2,\\[2pt]
 1,    &1-d^2\le x\le1.
 \end{cases}
\end{equation}
Here $d$ is a scalar parameter. If the tagged ensemble has decoding error at most $x$, then $\mathcal T_{d_k(\delta)}(x)$ bounds the error after its labels are removed. The function is nondecreasing and concave. Notably, $\mathcal T_d(0)=d^2$, whereas the elementary trace-distance comparison would give the larger value $d$.

\Needspace{13\baselineskip}
\begin{corollary}[Refined packing bound]\label{cor:sharp-packing}
Let $(p_m,\rho_m)_{m\in\mathcal M}$, the analysis effects, and the error sums $a_p,b_p$ satisfy the hypotheses of \corref{cor:weighted-packing}. For every $0<\delta<1$, the same compressed tilted-state PGM satisfies
\begin{equation}\label{eq:sharp-main}
 P_{\mathrm e}\le
 \mathcal T_{d_k(\delta)}\!\left(
  \min\{1,\kappa_k(\delta)a_p+(1-\delta)b_p\}\right).
\end{equation}
The quantities $d_k$, $\kappa_k$, and $\mathcal T_d$ are defined in
\eqref{eq:k-constants}, \eqref{eq:sharp-gram-definition}, and \eqref{eq:transfer-function}, respectively. Under a uniform prior, $a_p=a$ and $b_p=b$ as in \eqref{eq:k-ab}--\eqref{eq:line-6-6-2}. The right-hand side of \eqref{eq:sharp-main} is no larger than the basic guarantee $\min\{1,d_k(\delta)+C_k(\delta)a_p+b_p\}$.
\end{corollary}
\begin{proof}
Appendix~\ref{app:quantitative} proves that the tagged error is at most
$\kappa_k(\delta)a_p+b_{p,\delta}$, where $b_{p,\delta}$ retains the factor
$(1-\delta)^{|S(m,\ell)|}$ in every competing-message term. Thus
$b_{p,\delta}\le(1-\delta)b_p$. The exact probability-transfer estimate then gives \eqref{eq:sharp-main}. For the comparison with the basic bound, the same transfer estimate implies $\mathcal T_d(x)\le\min\{1,x+d\}$, and \eqref{eq:sharp-simple-coefficient} implies
$\kappa_k(\delta)a_p+(1-\delta)b_p\le C_k(\delta)a_p+b_p$.
\end{proof}

The refined statement changes no aspect of the protocol. It can be used after the familiar testing errors have been evaluated. For random codebooks or hashes, concavity permits the same formula with upper bounds on $\ee a_p$ and $\ee b_p$; Appendix~\ref{app:combine} proves this averaging step. The most detailed uniform-prior bound proved here, which retains every error-pattern weight and the quadratic correction from Cheng's trace inequality, is stated in \thref{thm:sharp-uniform}. The appendix also distinguishes optimality of the geometric and transfer steps from the unresolved question of an optimal complete coding bound.

\FloatBarrier

\section{Entanglement-assisted classical communication}\label{sec:ea}
We first apply the packing theorem without changing the communication objective: several senders transmit independent classical messages to one receiver. However, now let us suppost that entanglement is shared in advance between each sender and the receiver. Entanglement-assisted multiple-access regions were established in~\cite{HsiehDevetakWinter2008}, with simultaneous-decoding developments in~\cite{XuWilde2013} and an arbitrary-sender capacity treatment in~\cite{ShiEtAl2021}. Our purpose is a one-shot reduction to the theorem already proved. We use position-based coding, introduced by Anshu, Jain, and Warsi~\cite{AnshuJainWarsi2019}, in the multiple-access form developed by Qi \textit{et al}.~\cite[Section~4.1]{QiWangWilde2018}. The receiver identifies which reference position becomes correlated with the channel output. Separately chosen tests handle the different subsets of incorrect positions.

\subsection{Position-based coding and the receiver's states}\label{subsec:ea-position}
Let $\mathcal N\colon A_1\cdots A_k\to B$ be a finite-dimensional quantum channel, and fix pure bipartite states $\phi^{(i)}_{R_iA_i}$, with $R_i$ held by the receiver. Define
\begin{equation}\label{eq:ea-omega}
 \omega_{R_{[k]}B}
 \coloneqq \left(\operatorname{id}_{R_{[k]}}\otimes\mathcal N\right)
       \left(\bigotimes_{i=1}^k\phi^{(i)}_{R_iA_i}\right),
 \qquad R_{[k]}\coloneqq R_1\cdots R_k.
\end{equation}
The channel is trace preserving, so
\begin{equation}\label{eq:ea-product-reference}
 \omega_{R_{[k]}}=\bigotimes_i\phi^{(i)}_{R_i}.
\end{equation}
The input shares need not be classical. The classical coordinates in the packing theorem will be the message positions, not physical input bases.

Figure~\ref{fig:ea-position} illustrates position selection for one sender. For message sizes $M_i$, Sender~$i$ and the receiver share $M_i$ independent copies of $\phi^{(i)}$, with systems $A_{i,j}$ and $R_{i,j}$ at position $j$. To send $m_i\in[M_i]$, Sender~$i$ transmits $A_{i,m_i}$ through input $A_i$ of $\mathcal N$. All unused input shares are discarded. Let $\mathcal R$ denote the tensor product of all reference systems $R_{i,j}$ retained by the receiver. Its state for message tuple $m$ is
\begin{equation}\label{eq:ea-position-state}
 \rho_m^{\mathrm{pos}}
 \coloneqq\omega_{R_{1,m_1}\cdots R_{k,m_k}B}
 \otimes\bigotimes_{i=1}^k\bigotimes_{j\ne m_i}\phi^{(i)}_{R_{i,j}}.
\end{equation}
Tensor factors are placed in a fixed order. This is a finite ensemble on $\mathcal R B$, so the deterministic packing theorem applies directly. No random codebook is required.

For a proposed tuple $\ell$, keep $R_{1,\ell_1},\ldots,R_{k,\ell_k},B$ and trace out the other reference systems. Relabel the retained reference systems as $R_1,\ldots,R_k$. If $m=\ell$, the retained state is $\omega$. If $S=S(m,\ell)$ is nonempty, the retained state is
\begin{equation}\label{eq:ea-alternative}
 \sigma_S^{\mathrm{EA}}
 \coloneqq\omega_{R_S}\otimes\omega_{R_{S^c}B}.
\end{equation}
To verify this, trace out the correctly transmitted references $R_{i,m_i}$ for $i\in S$ in \eqref{eq:ea-position-state}; the correlated factor becomes $\omega_{R_{S^c}B}$. Each selected reference $R_{i,\ell_i}$ for $i\in S$ instead belongs to an unused pair, and contributes the independent factor $\phi^{(i)}_{R_i}$. Their product is $\omega_{R_S}$ by \eqref{eq:ea-product-reference}. For two senders, the three alternatives are therefore
\begin{equation}\label{eq:ea-two-alternatives}
 \omega_{R_1}\otimes\omega_{R_2B},\qquad
 \omega_{R_2}\otimes\omega_{R_1B},\qquad
 \omega_{R_1}\otimes\omega_{R_2}\otimes\omega_B,
\end{equation}
corresponding to types $1$, $2$, and $12$. This exact marginal identity replaces the independent-codeword averaging calculation of the unassisted proof.

\begin{figure}[!t]
\centering
\begin{adjustbox}{max width=\linewidth}
\begin{tikzpicture}[font=\small]
\node[anchor=west,font=\bfseries] at (0,3.5) {One sender's reference positions};
\foreach \j in {1,2,3,4}{
 \node[operator,minimum width=1.6cm,fill=softgreen] (r\j) at (2.4*\j-1.2,2.35) {$R_{i,\j}$};
 \node[operator,minimum width=1.6cm] (a\j) at (2.4*\j-1.2,.55) {$A_{i,\j}$};
 \draw[quantum] (r\j.south)--node[right,font=\footnotesize] {$\phi^{(i)}$}(a\j.north);
}
\node[operator,fill=softblue,text width=2.2cm] (ch) at (12,-1.0) {Channel output $B$};
\draw[flow] (a2.south)--(3.6,-1.0)--node[above] {transmit position $m_i=2$}(ch.west);
\draw[inkorange,densely dashed,thick,rounded corners] (2.6,1.65) rectangle (4.6,3.0);
\node[anchor=west,align=left,text width=4cm] at (9.6,2.35)
 {The receiver retains every reference position.};
\node[align=left,anchor=west,text width=12.8cm] at (0,-2.85)
 {Selecting the transmitted position retains its correlation with $B$. Selecting an unused position gives an independent reference state. The argument applies separately to every sender.};
\end{tikzpicture}
\end{adjustbox}
\caption{Entanglement-assisted position-based coding. The vertical connections represent pre-shared entangled states, not communication. The physical transmitted system is the selected input share; the other input shares are discarded. For a proposed message tuple, each incorrect reference position is independent of the channel output, giving \eqref{eq:ea-alternative}. These physical reference systems are distinct from the fictitious weak-label registers used only to analyze the decoder.}
\label{fig:ea-position}
\end{figure}

\subsection{The one-shot bound}\label{subsec:ea-one-shot}
The position-state identities reduce decoding to the same separate binary tests used by the packing theorem. The only task left is to evaluate their two errors and count the alternative positions. Choose an effect $T_S$ on $R_{[k]}B$ for each nonempty $S$ and define
\begin{equation}\label{eq:ea-errors}
 \alpha_S\coloneqq \Tr[(I-T_S)\omega],\qquad
 \beta_S\coloneqq \Tr[T_S\sigma_S^{\mathrm{EA}}].
\end{equation}
For proposed tuple $\ell$, let $E_{S,\ell}$ apply $T_S$ to its selected references and $B$, and the identity to all other reference systems. Equations~\eqref{eq:ea-position-state}--\eqref{eq:ea-alternative} give
\begin{align}
\Tr[(I-E_{S,m})\rho_m^{\mathrm{pos}}]&=\alpha_S,\label{eq:ea-test-identities}\\
\Tr[E_{S,\ell}\rho_m^{\mathrm{pos}}]&=\beta_S
       \quad\text{when }S(m,\ell)=S.\label{eq:line-7-7-2}
\end{align}
Thus the application needs no new intersection or disturbance estimate.

\begin{theorem}[One-shot entanglement-assisted simultaneous decoding]\label{thm:ea-one-shot}
Fix the finite-dimensional channel $\mathcal N$ and shared states in \eqref{eq:ea-omega}, positive message sizes $M_i$, and $0<\delta<1$. Define $\sigma_S^{\mathrm{EA}}$ by \eqref{eq:ea-alternative} and choose arbitrary effects with errors $\alpha_S,\beta_S$ from \eqref{eq:ea-errors}. Position-based coding using $M_i$ copies of each $\phi^{(i)}$, followed by one POVM on $\mathcal R B$, has uniform-message average error satisfying
\begin{equation}\label{eq:ea-one-shot}
 P_{\mathrm e}^{\mathrm{EA}}
 \le d_k(\delta)+C_k(\delta)\sum_{S\ne\varnothing}\alpha_S
 +\sum_{S\ne\varnothing}\left[\prod_{i\in S}(M_i-1)\right]\beta_S.
\end{equation}
Here $d_k,C_k$ are from \eqref{eq:k-constants}, and $P_{\mathrm e}^{\mathrm{EA}}=1-(\prod_iM_i)^{-1}\sum_m\Tr[\Lambda_m\rho_m^{\mathrm{pos}}]$. In particular, for $0\le\eps_S<1$, one may replace $\alpha_S$ by $\eps_S$ and $\beta_S$ by $2^{-D_H^{\eps_S}(\omega\Vert\sigma_S^{\mathrm{EA}})}$.
\end{theorem}
\begin{proof}
Apply the effect version of \thref{thm:k-packing} to \eqref{eq:ea-position-state}, or equivalently use \corref{cor:weighted-packing} with the uniform prior. The first identity in \eqref{eq:ea-test-identities}--\eqref{eq:line-7-7-2} gives $a=\sum_S\alpha_S$. For each true tuple, exactly $\prod_{i\in S}(M_i-1)$ alternatives have error set $S$, so the second identity gives $b=\sum_S[\prod_{i\in S}(M_i-1)]\beta_S$. Substitution proves \eqref{eq:ea-one-shot}. Choose an optimizing effect in \eqref{eq:DH} for each binary pair to obtain the last assertion. The general-effects version permits separately optimized one-shot tests. Unlike the classical-input application, no dephasing in the reference registers is assumed: these registers are genuinely quantum. Projective threshold tests will suffice for the asymptotic statement below.
\end{proof}

For a sharper finite-error guarantee, apply \corref{cor:sharp-packing} or \thref{thm:sharp-uniform} with the exact substitutions in \eqref{eq:sharp-channel-substitution}. The reference-position calculation and the physical protocol do not change.

The physical decoder is the compressed PGM of the labeled versions of \eqref{eq:ea-position-state}. Its weak labels are not supplied by the entanglement resource. Entanglement supplies the genuine reference systems $\mathcal R$; the separate, fictitious label registers are removed just as in the unassisted theorem. The amount of shared entanglement required by this position-based coding construction is not optimized.

\subsection{The asymptotic rate region}\label{subsec:ea-asymptotic}
To identify the rates, only one entropy calculation remains. Product reference marginals and the logarithm of a tensor product give
\begin{align}
D(\omega\Vert\sigma_S^{\mathrm{EA}})
 &=-H(R_{[k]}B)_\omega+H(R_S)_\omega+H(R_{S^c}B)_\omega\label{eq:ea-divergence}\\
&=H(R_{[k]})_\omega-H(R_{S^c})_\omega
       +H(R_{S^c}B)_\omega-H(R_{[k]}B)_\omega\label{eq:line-7-9-2}\\
&=I(R_S;B|R_{S^c})_\omega.\label{eq:line-7-9-3}
\end{align}
Support inclusion holds for a state relative to the product of its marginals; hence this relative entropy is finite in the present setting.

\begin{corollary}[Fixed-input entanglement-assisted rates]\label{cor:ea-rates}
For the channel and product shared input states in \eqref{eq:ea-omega}, every nonnegative rate tuple $(r_i)_{i=1}^k$ satisfying
\begin{equation}\label{eq:ea-region}
 r(S)\coloneqq\sum_{i\in S}r_i\le I(R_S;B|R_{S^c})_\omega
 \quad(\varnothing\ne S\subseteq[k])
\end{equation}
is achievable for classical communication with vanishing uniform-message average error, pairwise sender--receiver entanglement, and one joint decoding measurement. The shared input states and number $k$ are fixed; no bound on entanglement consumption is asserted.
\end{corollary}
\begin{proof}
For a strictly interior positive rate tuple, choose $r(S)<t_S<I(R_S;B|R_{S^c})_\omega$ and apply \lemref{lem:projective-stein} to $\omega^{\otimes n}$ versus $(\sigma_S^{\mathrm{EA}})^{\otimes n}$. The resulting tests satisfy $\alpha_{S,n}\to0$ and $\beta_{S,n}\le2^{-nt_S}$. Position-based coding is applied to the block channel, using $M_i=\lfloor2^{nr_i}\rfloor$ copies of $(\phi^{(i)})^{\otimes n}$. Equation~\eqref{eq:ea-one-shot} gives
\begin{equation}\label{eq:ea-asymptotic-bound}
 P_{\mathrm e}^{\mathrm{EA},(n)}
 \le d_k(\delta)+C_k(\delta)\sum_{S\ne\varnothing}\alpha_{S,n}
       +\sum_{S\ne\varnothing}2^{-n(t_S-r(S))}.
\end{equation}
Hold $\delta$ fixed while $n\to\infty$, then let $\delta\to0$, using the diagonal construction already given after \eqref{eq:two-final-asymptotic}. A zero-rate sender uses one position. Tests for impossible error sets can be the identity, so only subsets of positive-rate senders need strict inequalities. Scaling the positive rates toward the origin and diagonalizing includes boundary points, exactly as in \thref{thm:k-capacity}.
\end{proof}

This is a fixed-input achievable region, not an assertion that optimizing one use gives the unrestricted entanglement-assisted capacity region. The known capacity results and their regularization questions remain those of~\cite{HsiehDevetakWinter2008,ShiEtAl2021}. The point here is that the same one-shot theorem proves the simultaneous-decoding statement without a separate geometric argument.

\section{Distributed classical compression with quantum side information}\label{sec:compression}
The second application changes the direction of the coding problem. Several senders observe correlated classical data and communicate shorter descriptions to a receiver who already holds quantum side information. The case of one classical sender and one receiver with quantum side information was studied asymptotically by Devetak and Winter~\cite{DevetakWinter2003} and in the one-shot setting by Renes and Renner~\cite{RenesRenner2012}. The distributed setting with several classical senders is treated in~\cite{Winter1999Thesis,LiuXiFan2023}. In particular, Salek and Winter~\cite[Appendix~A, Lemma~19]{SalekWinter2022} prove simultaneous decoding from independent random bins using Sen's joint typicality construction. Our purpose is not to introduce simultaneous compression, but to recover its region through a direct one-shot application of the packing theorem already proved. The compression application of Cheng's state-based PGM analysis is the corresponding single-source precedent~\cite[Section~4.3]{Cheng2023}.

\subsection{The source and the transmitted hash values}\label{subsec:compression-model}
Figure~\ref{fig:compression-hashing} summarizes the encoding and decoding steps. Each sender transmits a hash value; the receiver uses these classical values together with the quantum side information.
Fix a positive integer $k$, let $\mathcal X_i$ be nonempty finite alphabets, and let $B$ be finite dimensional. The source is
\begin{equation}\label{eq:source-state}
 \omega_{X_{[k]}B}\coloneqq \sum_{x\in\prod_i\mathcal X_i}p(x)\proj{x}_{X_{[k]}}\otimes\rho_x.
\end{equation}
The probability distribution $p$ may be correlated; the conditional states $\rho_x$ are normalized. Sender~$i$ observes $x_i$. The receiver holds $B$ and must reconstruct the whole tuple. Unlike channel messages, source tuples are generally not uniform, which is why the arbitrary-prior corollary is needed.

Fix positive integers $L_i$. Sender~$i$ sends the value of a function $f_i\colon \mathcal X_i\to[L_i]$. For the existence proof, select the functions independently from families satisfying
\begin{equation}\label{eq:hash-collision}
 \Pr_{f_i}\{f_i(x_i)=f_i(z_i)\}\le L_i^{-1}
 \quad\text{for every }x_i\ne z_i.
\end{equation}
Such two-universal families exist; independent uniform assignments of outputs already suffice~\cite{CarterWegman1979}. The receiver knows the selected functions. Write $f(x)\coloneqq (f_1(x_1),\ldots,f_k(x_k))$ and let $C$ record this tuple of transmitted values. For fixed $f$, the receiver's ensemble is
\begin{equation}\label{eq:source-receiver-ensemble}
 (p(x),\rho_x^f)_x,\qquad
 \rho_x^f\coloneqq \proj{f(x)}_C\otimes\rho_x.
\end{equation}
An alternative tuple $z$ can be accepted only if it has the same transmitted hash values as the true tuple. The independent hashes suppress precisely the error classes used by the packing theorem.

For nonempty $S\subseteq[k]$, set $d_{X_S}\coloneqq \prod_{i\in S}|\mathcal X_i|$ and define the normalized alternative
\begin{equation}\label{eq:source-alternative}
 \zeta_S\coloneqq\pi_{X_S}\otimes\omega_{X_{S^c}B},\qquad
 \pi_{X_S}\coloneqq I_{X_S}/d_{X_S}.
\end{equation}
The uniform state is intentional: summing competing source strings counts their possible values, rather than drawing them independently from the source marginal. Choose a block-diagonal effect
\begin{equation}\label{eq:source-tests}
 T_S\coloneqq \sum_z\proj z\otimes T_{S,z},\quad 0\le T_{S,z}\le I_B,
 \qquad
 \alpha_S\coloneqq \Tr[(I-T_S)\omega],\quad \beta_S\coloneqq \Tr[T_S\zeta_S].
\end{equation}
Optimizing binary effects may be chosen block diagonal by dephasing the classical registers, as in \corref{cor:one-shot}.

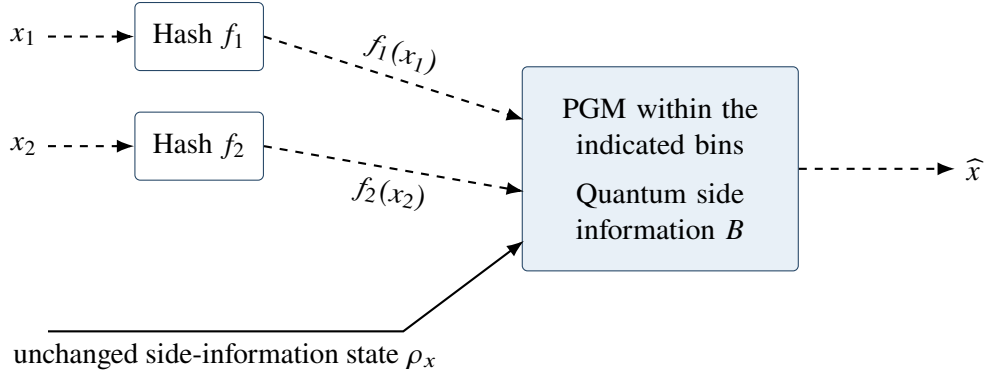
\begin{figure}[!t]
\centering
\begin{adjustbox}{max width=\linewidth}
\begin{tikzpicture}[font=\small]
\node[anchor=east] at (0,2.1) {$x_1$};
\node[anchor=east] at (0,.65) {$x_2$};
\node[operator,minimum width=1.7cm] (h1) at (2,2.1) {Hash $f_1$};
\node[operator,minimum width=1.7cm] (h2) at (2,.65) {Hash $f_2$};
\node[operator,text width=3.3cm,minimum height=2.7cm,fill=softblue] (dec) at (8.1,.35)
 {PGM within the indicated bins\\[5pt]Quantum side information $B$};
\draw[classical,->] (0,2.1)--(h1.west);
\draw[classical,->] (0,.65)--(h2.west);
\draw[classical,->] (h1.east)--node[above,sloped] {$f_1(x_1)$}(6.3,1.0);
\draw[classical,->] (h2.east)--node[below,sloped] {$f_2(x_2)$}(6.3,.05);
\draw[quantum,->] (0,-1.8)--node[below] {unchanged side-information state $\rho_x$}(4.7,-1.8)--(6.3,-.6);
\draw[classical,->] (dec.east)--(12.0,.35) node[right] {$\widehat x$};
\node[align=left,anchor=west,text width=12.0cm] at (.3,-3.30)
 {If only the first coordinate differs, a competitor survives hashing with probability at most $1/L_1$. If both differ, it survives with probability at most $1/(L_1L_2)$.};
\end{tikzpicture}
\end{adjustbox}
\caption{Distributed classical compression with quantum side information, shown for two senders. The classical source coordinates may be correlated with each other and with $B$. The receiver already holds $B$; it is not transmitted by either encoder. The bin-conditioned PGM retains the source probabilities and uses the same fictitious tilting and final compression as the general theorem. Random hash selection proves existence; one fixed collection of functions is used by the final code.}
\label{fig:compression-hashing}
\end{figure}

\subsection{One-shot decoding by the weighted packing bound}\label{subsec:compression-one-shot}
For a proposed source tuple $z$, let the analysis effect on $CB$ be
\begin{equation}\label{eq:source-analysis-effect}
 E^f_{S,z}\coloneqq\proj{f(z)}_C\otimes T_{S,z}.
\end{equation}
For the correct tuple, the classical factor accepts with certainty. For an alternative, it enforces equality of the transmitted hash values. This is the only additional ingredient in the proof.

\begin{theorem}[One-shot distributed compression]\label{thm:compression-one-shot}
Let the finite classical--quantum source, independent two-universal hashes, and block-diagonal tests be as in \eqref{eq:source-state}--\eqref{eq:source-tests}. Set $L(S)\coloneqq \prod_{i\in S}L_i$, let $0<\delta<1$, and use $d_k,C_k$ from \eqref{eq:k-constants}. There are decoders depending on the selected functions whose error in reconstructing the full source tuple satisfies
\begin{equation}\label{eq:source-one-shot}
 \ee_f P_{\mathrm e}(f)
 \le d_k(\delta)+C_k(\delta)\sum_{S\ne\varnothing}\alpha_S
       +\sum_{S\ne\varnothing}\frac{d_{X_S}}{L(S)}\beta_S,
\end{equation}
where $P_{\mathrm e}(f)=1-\sum_xp(x)\Tr[\Lambda_{x|f(x)}\rho_x]$ and failure counts as error. In particular, for $0\le\eps_S<1$ there is one fixed collection of hashes and a decoder with
\begin{equation}\label{eq:source-DH}
 P_{\mathrm e}
 \le d_k(\delta)+C_k(\delta)\sum_{S\ne\varnothing}\eps_S
       +\sum_{S\ne\varnothing}\frac{d_{X_S}}{L(S)}
           2^{-D_H^{\eps_S}(\omega\Vert\zeta_S)}.
\end{equation}
All hypotheses in the relative entropies are normalized states. No independence assumption is imposed on the source coordinates.
\end{theorem}
\begin{proof}
Apply \corref{cor:weighted-packing} to \eqref{eq:source-receiver-ensemble} and the effects in \eqref{eq:source-analysis-effect}. Identify the finite alphabets with message-index sets. For every $f$,
\begin{equation}\label{eq:source-a}
 a_p=\sum_x p(x)\sum_S\Tr[(I-T_{S,x})\rho_x]=\sum_S\alpha_S.
\end{equation}
If $S(x,z)=S$, independence of the hash functions and \eqref{eq:hash-collision} give
\begin{equation}\label{eq:source-pair-collision}
 \Pr_f\{f(x)=f(z)\}
 =\prod_{i\in S}\Pr_{f_i}\{f_i(x_i)=f_i(z_i)\}
 \le L(S)^{-1}.
\end{equation}
For each such pair, the false-acceptance term is
\begin{equation}\label{eq:source-false-acceptance}
 \Tr[E^f_{S,z}\rho_x^f]
 =\mathbf1_{\{f(x)=f(z)\}}\Tr[T_{S,z}\rho_x].
\end{equation}
Thus the expected sum of false acceptances is bounded by
\begin{equation}\label{eq:source-b-start}
 \ee_f b_p
 \le\sum_{S\ne\varnothing}\frac1{L(S)}
       \sum_z\sum_{x:S(x,z)=S}p(x)\Tr[T_{S,z}\rho_x].
\end{equation}
For fixed $z$, retain the equalities $x_{S^c}=z_{S^c}$ but drop the restrictions $x_i\ne z_i$ for $i\in S$. Every summand is nonnegative, so
\begin{align}
\sum_z\sum_{x:S(x,z)=S}p(x)\Tr[T_{S,z}\rho_x]
 &\le\sum_z\sum_{x_S}p(x_S,z_{S^c})
                  \Tr[T_{S,z}\rho_{x_S,z_{S^c}}]\label{eq:source-b-calculation}\\
&=\Tr[T_S(I_{X_S}\otimes\omega_{X_{S^c}B})]\label{eq:line-8-13-2}\\
&=d_{X_S}\Tr[T_S\zeta_S]=d_{X_S}\beta_S.\label{eq:line-8-13-3}
\end{align}
The second equality uses the classical-block expansion
\begin{equation}\label{eq:source-marginal-blocks}
 \omega_{X_{S^c}B}
 =\sum_{z_{S^c}}\proj{z_{S^c}}\otimes
       \sum_{x_S}p(x_S,z_{S^c})\rho_{x_S,z_{S^c}}.
\end{equation}
Equations~\eqref{eq:source-a}--\eqref{eq:line-8-13-3} prove \eqref{eq:source-one-shot}. Dephasing an optimizing effect for $D_H^{\eps_S}(\omega\Vert\zeta_S)$ does not change either testing probability, so it has the form in \eqref{eq:source-tests}. This gives \eqref{eq:source-DH} after selecting a hash realization with error no larger than the average.
\end{proof}

The same hash calculation also yields the refined bound in \eqref{eq:sharp-random-average}, with \eqref{eq:sharp-source-substitution}. The use of this nonlinear bound after averaging is justified by concavity, not by replacing a random error by its expectation without proof. The source prior remains arbitrary.

\subsubsection{The explicit decoder retains the source prior}
The theorem's PGM is block diagonal in the received classical value $c$. On that block, define
\begin{equation}\label{eq:source-pgm-normalizer}
 Z_c\coloneqq \sum_{x:f(x)=c}p(x)\rho_x\otimes\proj{t_x}.
\end{equation}
For $f(z)=c$, its extended-space outcome effect is
\begin{equation}\label{eq:source-pgm}
 \widetilde\Lambda_{z|c}
 =Z_c^{-1/2}\bigl(p(z)\rho_z\otimes\proj{t_z}\bigr)Z_c^{-1/2};
 \qquad \widetilde\Lambda_{z|c}=0\quad\text{if }f(z)\ne c.
\end{equation}
Complete the POVM on each block, and compress against $\ket{0^k}$ to obtain $\Lambda_{z|c}$ on $B$. A zero-probability bin can be assigned a failure outcome. Normalizing $p(z)$ conditionally on $c$ would multiply every numerator and $Z_c$ by the same constant and hence leave the PGM unchanged. Omitting the relative source probabilities would in general change the decoder. This is the only distinction in normalization from uniform-message channel coding.

\subsection{The conditional-entropy rate region}\label{subsec:compression-rates}
We next recover the known distributed rate region from the one-shot estimate. In particular, this gives the same simultaneous-reconstruction constraints as Salek and Winter's independent-random-binning theorem~\cite[Appendix~A, Lemma~19]{SalekWinter2022}. The new calculation needed here is only the relative entropy of the source and its normalized alternative.

The alphabet-size factors in the one-shot bound cancel exactly in the asymptotic calculation. Indeed, the logarithm of \eqref{eq:source-alternative} gives
\begin{align}
D(\omega\Vert\zeta_S)
 &=-H(X_{[k]}B)_\omega+\log_2d_{X_S}+H(X_{S^c}B)_\omega\label{eq:source-divergence}\\
&=\log_2d_{X_S}-H(X_S|X_{S^c}B)_\omega.\label{eq:line-8-17-2}
\end{align}
The support condition holds because the support of each positive classical block of $\omega$ is contained in that of its marginal sum. No product assumption on $p$ is used in this identity.

\begin{corollary}[Distributed classical compression rates]\label{cor:compression-rates}
For independent copies of the source \eqref{eq:source-state}, every nonnegative rate tuple $(r_i)_{i=1}^k$ satisfying
\begin{equation}\label{eq:source-region}
 \sum_{i\in S}r_i\ge H(X_S|X_{S^c}B)_\omega
 \quad(\varnothing\ne S\subseteq[k])
\end{equation}
is achievable by separate deterministic classical encoders and one quantum decoder, with vanishing average reconstruction error. Here Sender~$i$ uses an alphabet of size $L_i^{(n)}$ with $\limsup_n n^{-1}\log_2 L_i^{(n)}\le r_i$. No communication between encoders or shared randomness during operation is required.
\end{corollary}
\begin{proof}
First suppose all inequalities in \eqref{eq:source-region} are strict. Choose $\gamma>0$ smaller than each positive gap. Write $h_S\coloneqq H(X_S|X_{S^c}B)_\omega$ and apply \lemref{lem:projective-stein} to $\omega^{\otimes n}$ and $\zeta_S^{\otimes n}$ with threshold $\log_2d_{X_S}-h_S-\gamma$. This threshold is allowed to be negative; the lemma requires only that it be strictly below the relative entropy. The tests can be chosen block diagonal and satisfy
\begin{equation}\label{eq:source-Stein}
 \alpha_{S,n}\longrightarrow0,\qquad
 \beta_{S,n}\le2^{-n(\log_2d_{X_S}-h_S-\gamma)}.
\end{equation}
Use $L_i^{(n)}=\lceil2^{nr_i}\rceil$ and independent two-universal hashes on the block alphabets. Substituting \eqref{eq:source-Stein} into \eqref{eq:source-one-shot} gives
\begin{equation}\label{eq:source-asymptotic}
 \ee_fP_{\mathrm e}^{(n)}
 \le d_k(\delta)+C_k(\delta)\sum_{S\ne\varnothing}\alpha_{S,n}
       +\sum_{S\ne\varnothing}2^{-n(r(S)-h_S-\gamma)},
 \qquad r(S)=\sum_{i\in S}r_i.
\end{equation}
There are finitely many sets $S$. Fix $\delta$, let $n$ grow, and then choose $\delta$ smaller by the same diagonal argument used for channel coding. Select one hash collection with no larger error at each blocklength. Rounding the alphabet sizes changes each rate by at most $1/n$. For a boundary tuple, apply the strict construction to $r_i+1/j$ and choose increasing blocklength thresholds for error at most $1/j$. The resulting limiting rates obey the stated upper bounds, including when a target rate is zero.
\end{proof}

For two senders, \eqref{eq:source-region} reads
\begin{equation}\label{eq:source-two-region}
 r_1\ge H(X_1|X_2B),\qquad
 r_2\ge H(X_2|X_1B),\qquad
 r_1+r_2\ge H(X_1X_2|B),
\end{equation}
with all entropies evaluated on $\omega$. These are the established distributed compression constraints~\cite{Winter1999Thesis,SalekWinter2022,LiuXiFan2023}; the result here is their derivation from the same one-shot tilted-state packing theorem. There is no separate intersection proof for source coding.

For comparison, Chou~\cite[Section~V-B, Lemma~4]{Chou2022} recalls this distributed source-coding region, attributes it to Winter's thesis, and explains its achievability from single-sender coding and time sharing. Section~V-C of that work uses distributed compression as an ingredient in channel coding. This source-coding reduction is distinct from the direct simultaneous-hashing analysis used here.

\subsection{Retaining the quantum side information}\label{subsec:compression-gentle}
Decoding need not destroy the side-information state. This is relevant when compression is part of a larger protocol. Implement a bin-conditioned POVM by the operators $\sqrt{\Lambda_{z|c}}$, record its classical answer, and retain $B$.

\begin{proposition}[Reconstruction with small disturbance]\label{prop:compression-gentle}
For the source \eqref{eq:source-state}, fix hashes and a decoder with average reconstruction error $\eps$. Include failure as a distinct decoded symbol. The preceding measurement implementation produces a state on the original classical record $X_{[k]}$, its decoded copy $\widehat X_{[k]}$, and $B$ whose normalized trace distance from
\begin{equation}\label{eq:source-ideal-reconstruction}
 \Omega=\sum_xp(x)\proj x_X\otimes\proj x_{\widehat X}\otimes\rho_x
\end{equation}
is at most $\sqrt\eps+\eps/2$. Here $X$ and $\widehat X$ abbreviate the two tuples of classical registers.
\end{proposition}
\begin{proof}
For a true tuple $x$, set $E_x\coloneqq \Lambda_{x|f(x)}$ and $e_x\coloneqq1-\Tr[E_x\rho_x]$. Let $\ket{\psi_x}$ purify $\rho_x$ and set $\ket{v_x}\coloneqq (I\otimes\sqrt{E_x})\ket{\psi_x}$. Since $\sqrt{E_x}\ge E_x$ for $0\le E_x\le I$,
\begin{equation}\label{eq:source-gentle-vector}
 \norm{\ket{\psi_x}-\ket{v_x}}^2
 =1+\Tr[E_x\rho_x]-2\Tr[\sqrt{E_x}\rho_x]\le e_x.
\end{equation}
Write the difference of the two rank-one operators as
\begin{equation}\label{eq:source-gentle-decomposition}
 \proj{\psi_x}-\ket{v_x}\bra{v_x}
 =(\ket{\psi_x}-\ket{v_x})\bra{\psi_x}
   +\ket{v_x}(\bra{\psi_x}-\bra{v_x}).
\end{equation}
The trace norm of $\ket a\bra b$ is $\norm{\ket a}\norm{\ket b}$, and $\norm{\ket{v_x}}\le1$. The triangle inequality and partial-trace contractivity therefore give
\begin{equation}\label{eq:source-gentle-bound}
 \norm{\rho_x-\sqrt{E_x}\rho_x\sqrt{E_x}}_1\le2\sqrt{e_x}.
\end{equation}
The actual reconstructed state is
\begin{equation}\label{eq:source-actual-reconstruction}
 \Omega'=\sum_{x,z}p(x)\proj x_X\otimes\proj z_{\widehat X}
 \otimes\sqrt{\Lambda_{z|f(x)}}\rho_x\sqrt{\Lambda_{z|f(x)}}.
\end{equation}
The classical blocks are orthogonal. For each $x$, the correct block differs by at most $2\sqrt{e_x}$, and all incorrect blocks have total trace $e_x$. Hence
\begin{equation}\label{eq:source-gentle-final}
 \tfrac12\norm{\Omega'-\Omega}_1
 \le\sum_xp(x)\sqrt{e_x}+\tfrac12\sum_xp(x)e_x
 \le\sqrt\eps+\eps/2,
\end{equation}
where the final step is scalar Cauchy--Schwarz. This is the usual gentle-measurement mechanism, written explicitly for the source decoder.
\end{proof}

Both new applications have used the same quantum construction. Entanglement assistance changes which reference systems are available; compression adds a classical test of matching hash values and a nonuniform prior. The next application changes the physical signal space instead, to coherent optical states, while leaving the packing argument intact.

\section{Bosonic multiple-access channels with coherent-state encoding}\label{sec:bosonic}
The same simultaneous-decoding mechanism applies to optical signals mixed by a passive interferometer. We begin with two senders at the inputs of a beamsplitter and then treat $k$ senders at the inputs of a $k$-port splitter. The receiver has access to one output mode per channel use and may make an arbitrary collective measurement over its received modes. The other output modes are inaccessible. The coherent-state rate region is known~\cite{YenShapiro2005}; our purpose is to recover its achievability from the message-label construction, including a careful passage from finite constellations to Gaussian modulation. Extensions to more general Gaussian multiple-access channels are studied in~\cite{ShiZhuang2022}. Explicit optical receiver constructions address a different implementation question~\cite{WildeGuha2012}.

\subsection{Two inputs, one received output, and the energy constraint}\label{subsec:bosonic-model}
A single bosonic mode has Fock basis $(\ket q)_{q=0}^\infty$, annihilation operator $\hat a$, and photon-number operator $\hat n=\hat a^\dagger\hat a$. A coherent state of complex amplitude $z\in\mathbb{C}$ is
\begin{equation}\label{eq:coherent-state}
 \ket z\coloneqq e^{-|z|^2/2}\sum_{q=0}^\infty\frac{z^q}{\sqrt{q!}}\ket q,
 \qquad \hat a\ket z=z\ket z,
 \qquad \langle z|\hat n|z\rangle=|z|^2.
\end{equation}
The basis vectors $\ket{m_i}$ used to construct the auxiliary labels are orthonormal; the weak labels themselves are not. Physical coherent states with different amplitudes are also generally nonorthogonal, but they are optical signals rather than auxiliary message labels. In fact,
\begin{equation}\label{eq:coherent-overlap}
 |\langle z|z'\rangle|^2=e^{-|z-z'|^2},\qquad
 \Delta(\proj z,\proj{z'})=\sqrt{1-e^{-|z-z'|^2}}\le |z-z'|.
\end{equation}
The first identity follows by summing the exponential series in \eqref{eq:coherent-state}; the second uses \eqref{eq:pure-distance} and $1-e^{-u}\le u$ for $u\ge0$.

Fix $0<\eta<1$. A lossless beamsplitter transforms the two input modes $A_1,A_2$ into output modes $B,E$ according to
\begin{equation}\label{eq:beamsplitter}
 \hat b=\sqrt\eta\,\hat a_1+\sqrt{1-\eta}\,\hat a_2,
 \qquad
 \hat e=-\sqrt{1-\eta}\,\hat a_1+\sqrt\eta\,\hat a_2.
\end{equation}
The two rows are orthonormal, so the canonical commutation relations are preserved. To keep the state notation readable, denote the two output amplitudes by
\begin{equation}\label{eq:two-output-amplitudes}
 \gamma_B\coloneqq\sqrt\eta z_1+\sqrt{1-\eta}z_2,
 \qquad
 \gamma_E\coloneqq-\sqrt{1-\eta}z_1+\sqrt\eta z_2.
\end{equation}
Their dependence on $(z_1,z_2)$ is understood. The passive unitary leaves vacuum invariant and transforms displacement amplitudes by the same matrix. Hence
\begin{equation}\label{eq:beamsplitter-coherent}
 \ket{z_1}_{A_1}\ket{z_2}_{A_2}
 \longmapsto\ket{\gamma_B}_B\ket{\gamma_E}_E.
\end{equation}
In particular, tracing out $E$ leaves a pure coherent state on $B$. The induced classical--quantum channel is
\begin{equation}\label{eq:bosonic-two-channel}
 (z_1,z_2)\longmapsto\rho_{z_1,z_2}
   =\proj{\gamma_B(z_1,z_2)}_B.
\end{equation}
These are standard passive-linear-optical identities~\cite{YenShapiro2005,Serafini2023}. The optical phases and the channel parameters are fixed and known. Phase calibration is not a message-dependent auxiliary resource.

Figure~\ref{fig:beamsplitter} shows the two-input optical geometry and the discarded output. For an $n$-use code, Sender~$i$ transmits $\bigotimes_{t=1}^n\ket{z_{i,t}(m_i)}$. We impose the uniform-message average photon-number constraint
\begin{equation}\label{eq:optical-energy}
 \mathcal E_i(\cC)\coloneqq
 \frac1{nM_i}\sum_{m_i=1}^{M_i}\sum_{t=1}^n|z_{i,t}(m_i)|^2\le N_i,
 \qquad i=1,2.
\end{equation}
Here $0\le N_i<\infty$. This is a bound on the mean input photon number, averaged over messages and uses. It is not a hard cutoff on photon number, a maximum-energy constraint for every codeword, or a maximum-error criterion. Error is the uniform-message probability in \eqref{eq:code-error}, with the output states given by \eqref{eq:bosonic-two-channel}.

\begin{figure}[!t]
\centering
\begin{adjustbox}{max width=\linewidth}
\begin{tikzpicture}[font=\small]
 % A physical two-port optical plate: left and lower inputs; right and upper outputs.
 \coordinate (C) at (3.1,0);
 \draw[fill=softblue,draw=inkblue!40,rounded corners=1pt] (2.40,-.70) rectangle (3.80,.70);
 \draw[inkblue,line width=2pt] (2.48,-.62)--(3.72,.62);
 \draw[quantum,->] (-.1,0)--(2.94,0);
 \node[anchor=east] at (-.2,0) {$A_1:\ket{z_1}$};
 \draw[quantum,->] (3.1,-2.25)--(3.1,-.16);
 \node[anchor=north] at (3.1,-2.34) {$A_2:\ket{z_2}$};
 \draw[quantum,->] (3.26,0)--(8.15,0);
 \node[anchor=south,align=center] at (5.8,.28)
 {$B:\ket{\gamma_B}$};
 \draw[quantum,->] (3.1,.16)--(3.1,2.10);
 \draw[line width=1pt] (2.88,2.14)--(3.32,2.14);
 \foreach \x in {2.90,3.02,3.14,3.26}{\draw[gray] (\x,2.14)--(\x+.12,2.34);}
 \node[anchor=west,align=left] at (3.53,1.7) {$E$: inaccessible\\output mode};
 \node[anchor=east,align=right] at (2.16,-.66) {beamsplitter\\$\eta$};
 \node[operator,text width=2.4cm,minimum height=1.0cm] (dec) at (9.5,0)
 {Collective decoder\\on $B^n$};
 \draw[classical,->] (dec.east)--(12.1,0) node[right] {$(\widehat m_1,\widehat m_2)$};
 \node[font=\footnotesize,anchor=west] at (-.1,-3.23)
 {One use is shown; the decoder acts after $n$ received modes have been stored.};
\end{tikzpicture}
\end{adjustbox}
\caption{The two-sender optical multiple-access channel, drawn with the standard diagonal beamsplitter interface. Two input paths interfere at the plate. The right output is delivered to the receiver and the upper output is discarded, so the receiver does not obtain both outputs. The diagram shows one optical use followed schematically by storage and collective decoding over $n$ uses. No weak message labels are sent through the optical channel.}
\label{fig:beamsplitter}
\end{figure}

\subsection{Finite constellations: an exact finite-dimensional reduction}\label{subsec:constellations}
First restrict each input amplitude $z_i$ to a finite constellation $\mathcal A_i\subset\cc$. For $i\in\{1,2\}$, let $p_i$ be a probability distribution on that constellation: $p_i(z_i)\ge0$ and $\sum_{z_i\in\mathcal A_i}p_i(z_i)=1$. The input choices are independent, so their joint probability is $p_1(z_1)p_2(z_2)$. All sums below run over these finite alphabets. All conditional output vectors lie in
\begin{equation}\label{eq:signal-span}
 \cH_{\mathrm{sig}}
 \coloneqq\spanop\left\{\ket{\gamma_B(z_1,z_2)}:
          z_1\in\mathcal A_1,\ z_2\in\mathcal A_2\right\}.
\end{equation}
Its dimension is at most $|\mathcal A_1||\mathcal A_2|$. This is an exact finite-dimensional space even though each coherent state has infinitely many Fock components. We do not truncate those components. All $n$-use outputs lie in $\cH_{\mathrm{sig}}^{\otimes n}$, so the finite-dimensional coding theorem applies without modification. A POVM on that space extends to the full Fock space by adding its orthogonal complement to the failure effect.

To evaluate the conditional information quantities, define
\begin{align}
\tau_1&\coloneqq \sum_{z_1}p_1(z_1)\proj{\sqrt\eta z_1},\label{eq:optical-finite-mixtures}\\
\tau_2&\coloneqq \sum_{z_2}p_2(z_2)\proj{\sqrt{1-\eta}z_2},\label{eq:line-9-8-2}\\
\tau_{12}&\coloneqq \sum_{z_1,z_2}p_1(z_1)p_2(z_2)
                    \proj{\gamma_B(z_1,z_2)}.\label{eq:line-9-8-3}
\end{align}
The displacement $D(z)\coloneqq\exp(z\hat b^\dagger-\overline z\hat b)$ satisfies $D(z)\proj wD(z)^\dagger=\proj{z+w}$. Therefore, for the finite-constellation input--output state $\omega$,
\begin{align}
\overline\rho_{z_2}
 &=D(\sqrt{1-\eta}z_2)\tau_1D(\sqrt{1-\eta}z_2)^\dagger,\label{eq:optical-conditional-entropy}\\
I(X_1;B|X_2)_\omega
 &=\sum_{z_2}p_2(z_2)H(\overline\rho_{z_2})
   -\sum_{z_1,z_2}p_1(z_1)p_2(z_2)H(\rho_{z_1,z_2})\label{eq:line-9-9-2}\\
&=H(\tau_1).\label{eq:line-9-9-3}
\end{align}
The conditional output states are pure, so the second sum is equal to zero; displacement is unitary, so every term in the first sum is equal to $H(\tau_1)$. The same calculation with the sender roles exchanged, and then with both inputs averaged, gives
\begin{equation}\label{eq:finite-three-entropies}
 I(X_2;B|X_1)_\omega=H(\tau_2),\qquad
 I(X_1X_2;B)_\omega=H(\tau_{12}).
\end{equation}
Conditioning on the other sender means that its amplitude is a known displacement in this entropy calculation. It does not mean that the physical decoder is given the other message or decodes it first.

\begin{proposition}[Simultaneous decoding for finite coherent-state constellations]\label{prop:finite-optical}
Consider the channel \eqref{eq:bosonic-two-channel} with finite constellations $\mathcal A_1,\mathcal A_2$, independent input distributions $p_1,p_2$, and mixtures $\tau_1,\tau_2,\tau_{12}$ from \eqref{eq:optical-finite-mixtures}--\eqref{eq:line-9-8-3}. Every nonnegative rate pair satisfying
\begin{equation}\label{eq:finite-optical-region}
 R_1\le H(\tau_1),\qquad R_2\le H(\tau_2),\qquad
 R_1+R_2\le H(\tau_{12})
\end{equation}
is achievable in uniform-message average error by independent unsplit codebooks and a collective measurement on $B^n$. The deterministic and random-coding bounds in \thref{thm:two-packing} and \prref{prop:two-random} hold on the exact signal space \eqref{eq:signal-span}. If $\ee_{p_i}|z_i|^2<N_i$ whenever $N_i>0$, and Sender~$i$ uses vacuum when $N_i=0$, good deterministic codebooks can also be chosen to satisfy \eqref{eq:optical-energy}.
\end{proposition}
\begin{proof}
The refined bounds of \sref{subsubsec:sharp-packing} also apply on the finite signal span, with the same codebooks and decoder. The rate and error claims follow by applying the finite-dimensional theorems to \eqref{eq:signal-span} and substituting \eqref{eq:optical-conditional-entropy}--\eqref{eq:finite-three-entropies}. For the energy claim, write $\mu_i\coloneqq \ee|z_i|^2$ and $v_i\coloneqq\operatorname{Var}(|z_i|^2)$, which are finite because the constellation is finite. For a fixed constellation and fixed distribution $p_i$, the mean $\mu_i$, variance $v_i$, and positive energy slack $N_i-\mu_i$ do not depend on $n$. The $nM_i$ sampled costs are independent, so their average has variance $v_i/(nM_i)$. Chebyshev's inequality therefore gives
\begin{equation}\label{eq:energy-chebyshev}
 \Pr\{\mathcal E_i(\cC)>N_i\}
 \le\frac{v_i}{nM_i(N_i-\mu_i)^2}\xrightarrow[n\to\infty]{}0
 \qquad(N_i>\mu_i).
\end{equation}
The message size $M_i=M_i(n)$ may grow, but $M_i(n)\ge1$ already ensures $nM_i(n)\to\infty$. This limit concerns blocklength, not constellation refinement. A vacuum sender satisfies its constraint exactly. If the expected decoding error is $\xi_n\to0$, Markov's inequality gives $\Pr\{\Pe>\sqrt{\xi_n}\}\le\sqrt{\xi_n}$ when $\xi_n>0$; if $\xi_n=0$, the error is zero almost surely. The sum of these probabilities and the energy-violation probabilities tends to zero. For every sufficiently large $n$, a realization therefore has small error and satisfies both energy constraints. This selects a deterministic realization of the independent ensemble; no shared codebook randomness is needed during communication. The same selection can be applied at each stage of the boundary-rate construction.
\end{proof}

\subsection{Gaussian mixtures and the two-sender region}\label{subsec:gaussian-optical}
Circular Gaussian modulation makes the three entropies explicit. A centered circular complex Gaussian variable $Z$ with variance $N>0$ has density $e^{-|z|^2/N}/(\pi N)$ and $\ee|Z|^2=N$. Its average coherent state is
\begin{equation}\label{eq:thermal-mixture}
 \vartheta_N\coloneqq\int_{\cc}\frac{e^{-|z|^2/N}}{\pi N}\proj z\,\mathrm d^2z
   =\sum_{q=0}^\infty\frac{N^q}{(N+1)^{q+1}}\proj q.
\end{equation}
For completeness, expansion in \eqref{eq:coherent-state} makes the angular integral equal to zero for off-diagonal Fock entries. On the diagonal, polar coordinates give
\begin{equation}\label{eq:thermal-diagonal}
 \langle q|\vartheta_N|q\rangle
 =\frac{2}{Nq!}\int_0^\infty
       e^{-(1+1/N)r^2}r^{2q+1}\,\mathrm dr
 =\frac{N^q}{(N+1)^{q+1}}.
\end{equation}
This geometric distribution has mean $N$. Define the thermal entropy function by
\begin{equation}\label{eq:g-definition}
 g(N)\coloneqq(N+1)\log_2(N+1)-N\log_2N
 \quad(N>0),\qquad g(0)\coloneqq0.
\end{equation}
Its continuous extension at zero uses $\lim_{x\downarrow0}x\log_2x=0$. The entropy of the thermal state is
\begin{align}
H(\vartheta_N)
 &=\sum_{q=0}^\infty\frac{N^q}{(N+1)^{q+1}}
    \bigl[(q+1)\log_2(N+1)-q\log_2N\bigr]\label{eq:thermal-entropy}\\
&=(N+1)\log_2(N+1)-N\log_2N=g(N).\label{eq:line-9-15-2}
\end{align}
Set $\vartheta_0\coloneqq\proj0$; its entropy is the value $g(0)=0$ already defined in \eqref{eq:g-definition}.

With independent Gaussian inputs of variances $N_1,N_2$, the mixtures corresponding to \eqref{eq:optical-finite-mixtures}--\eqref{eq:line-9-8-3} are $\vartheta_{\eta N_1}$, $\vartheta_{(1-\eta)N_2}$, and $\vartheta_{\eta N_1+(1-\eta)N_2}$. Their entropies suggest the following established coherent-state region~\cite{YenShapiro2005}. The proof below justifies achievability by finite constellations, rather than applying a finite-alphabet theorem to continuous inputs without an approximation argument.

\Needspace{18\baselineskip}
\begin{theorem}[Two-sender coherent-state achievability]\label{thm:bosonic-two}
For the beamsplitter channel \eqref{eq:bosonic-two-channel}, let $0<\eta<1$ and $0\le N_1,N_2<\infty$. Under the uniform-message average photon-number constraints \eqref{eq:optical-energy}, every nonnegative rate pair satisfying
\begin{align}
R_1&\le g(\eta N_1),\label{eq:bosonic-two-region}\\
R_2&\le g((1-\eta)N_2),\label{eq:line-9-16-2}\\
R_1+R_2&\le g(\eta N_1+(1-\eta)N_2)\label{eq:line-9-16-3}
\end{align}
is achievable with coherent-state product codewords and a single collective decoding POVM on the received modes. Here $g$ is defined in \eqref{eq:g-definition}, and error is defined by \eqref{eq:code-error}. The codes can be selected from independent finite-constellation ensembles, followed by a limiting choice of the constellation and blocklength. No rate splitting or time sharing is required.
\end{theorem}

The only analytic issue is approximation: Gaussian amplitudes have unbounded support, and von Neumann entropy is not uniformly continuous on an unconstrained infinite-dimensional state space. The following elementary lemma supplies precisely the controlled approximation needed here. Its formulation already allows several independent amplitudes, so it will also apply to the $k$-sender channel.

\begin{lemma}[Finite constellations approaching Gaussian mixture entropies]\label{lem:gaussian-approx}
Fix a positive integer $j$. Let $Z_1,\ldots,Z_j$ be independent centered circular complex Gaussian variables with $\ee|Z_i|^2=P_i\ge0$, with $P_i=0$ interpreted as $Z_i=0$. There are independent finite-valued variables $Z_i^{(r)}$, coupled to $Z_i$ and indexed by positive integers $r$, such that
\begin{equation}\label{eq:quantized-moments}
 \ee Z_i^{(r)}=0,\qquad
 \ee|Z_i^{(r)}|^2\le P_i,\qquad
 \ee|Z_i^{(r)}-Z_i|^2\xrightarrow[r\to\infty]{}0.
\end{equation}
For fixed coefficients $u_i\in\cc$, define
\begin{equation}\label{eq:approx-mixtures}
 Y_r\coloneqq\sum_i u_iZ_i^{(r)},\qquad
 \sigma_r\coloneqq\ee\proj{Y_r},\qquad
 \nu\coloneqq\sum_i|u_i|^2P_i.
\end{equation}
Then $\sigma_r$ is finite rank for every $r$, and, as $r\to\infty$ with $j$, $P_i$, and $u_i$ fixed, $\Delta(\sigma_r,\vartheta_\nu)\to0$ and $H(\sigma_r)\to g(\nu)$. Here $\vartheta_\nu$ and $g$ are defined in \eqref{eq:thermal-mixture} and \eqref{eq:g-definition}. The index $r$ refines the input constellation; it is distinct from the coding blocklength $n$.
\end{lemma}
\begin{proof}
For $P_i>0$, partition the square $[-r,r]^2\subset\cc$ into finitely many cells of side at most $1/r$, symmetrically about zero, and take the square's complement as one additional cell. Replace $Z_i$ in each cell by its conditional mean. In other words, $Z_i^{(r)}$ is the conditional expectation of $Z_i$ given the cell containing it. The resulting variable takes finitely many values. The exterior cell is symmetric about zero, so its conditional mean is zero. The law of total expectation and conditional Jensen's inequality give
\begin{align}
 \ee Z_i^{(r)}&=\ee Z_i=0,\label{eq:quantization-proof}\\
 \ee|Z_i^{(r)}|^2&\le\ee|Z_i|^2=P_i.\label{eq:line-9-19-2}
\end{align}
Inside the square, a conditional mean lies in the closure of its cell, whose squared diameter is at most $2/r^2$. Outside the square, the replacement is zero. Splitting the expectation over these two events therefore gives
\begin{align}
 \ee|Z_i-Z_i^{(r)}|^2
 &\le\frac2{r^2}\Pr\{Z_i\in[-r,r]^2\}
       +\ee\bigl[|Z_i|^2\mathbf1_{\{Z_i\notin[-r,r]^2\}}\bigr]\label{eq:quantization-cell-error}\\
 &\le\frac2{r^2}
       +\ee\bigl[|Z_i|^2\mathbf1_{\{Z_i\notin[-r,r]^2\}}\bigr]
       \xrightarrow[r\to\infty]{}0.\label{eq:line-9-19-3}
\end{align}
For the last limit, $2/r^2\to0$ and the tail integrand tends pointwise to zero while being bounded by the integrable random variable $|Z_i|^2$. Dominated convergence applies. This construction is performed separately for the independent inputs, so it preserves their independence. For $P_i=0$, use $Z_i^{(r)}=0$ for every $r$.

The conditional-mean property also implies $\ee[(Z_i-Z_i^{(r)})\overline{Z_i^{(r)}}]=0$. Expanding the squared error gives
\begin{equation}\label{eq:quantization-energy-identity}
 P_i-\ee|Z_i^{(r)}|^2
 =\ee|Z_i-Z_i^{(r)}|^2\xrightarrow[r\to\infty]{}0.
\end{equation}
Thus the input second moments converge as well. Put $Y\coloneqq\sum_i u_iZ_i$. The differences $Z_i^{(r)}-Z_i$ are independent and centered. Their cross terms vanish, so
\begin{equation}\label{eq:amplitude-L2}
 \ee|Y_r-Y|^2
 =\sum_i|u_i|^2\ee|Z_i^{(r)}-Z_i|^2
 \xrightarrow[r\to\infty]{}0.
\end{equation}
There are only $j$ summands, and $j$ and all $u_i$ remain fixed.

The variable $Y$ is circular Gaussian with variance $\nu$; when $\nu=0$, it is zero almost surely. Thus $\vartheta_\nu=\ee\proj Y$, on the same probability space as $Y_r$. The triangle inequality for the trace norm of an integral yields
\begin{align}
 \Delta(\sigma_r,\vartheta_\nu)
 &=\tfrac12\norm{\ee(\proj{Y_r}-\proj Y)}_1\label{eq:optical-trace-expectation}\\
 &\le\ee\,\Delta(\proj{Y_r},\proj Y).\label{eq:optical-trace-approx}
\end{align}
The integrand is trace-norm integrable because it has norm at most two. Now use \eqref{eq:coherent-overlap} pointwise: $\sqrt{1-e^{-|Y_r-Y|^2}}\le|Y_r-Y|$. Scalar Cauchy--Schwarz then gives
\begin{align}
 \ee\,\Delta(\proj{Y_r},\proj Y)
 &\le\ee|Y_r-Y|\label{eq:optical-distance-amplitude}\\
 &\le\sqrt{\ee|Y_r-Y|^2}\xrightarrow[r\to\infty]{}0.\label{eq:line-9-21-2}
\end{align}
This proves trace-distance convergence without changing the optical states by a Fock-space cutoff.

The mean photon number is also controlled. Using \eqref{eq:coherent-state}, centered inputs, and independence, we obtain
\begin{align}
 \nu_r\coloneqq\Tr[\hat n\sigma_r]
 &=\ee|Y_r|^2\label{eq:approx-energy}\\
 &=\sum_i|u_i|^2\ee|Z_i^{(r)}|^2\le\nu,
 \qquad \nu_r\xrightarrow[r\to\infty]{}\nu.\label{eq:approx-energy-limit}
\end{align}
The expectation of each cross term is zero, and \eqref{eq:quantization-energy-identity} proves the final limit. Each $\sigma_r$ is a finite mixture of pure states and hence has finite entropy. To obtain the entropy upper bound, the thermal-state formula gives, for $\nu_r>0$,
\begin{equation}\label{eq:thermal-logarithm}
 -\log_2\vartheta_{\nu_r}
 =\log_2(1+\nu_r)I+
   \log_2\!\left(\frac{1+\nu_r}{\nu_r}\right)\hat n.
\end{equation}
Its expectation in $\sigma_r$ is finite and equals $g(\nu_r)$ by \eqref{eq:approx-energy}--\eqref{eq:approx-energy-limit}. Nonnegativity of relative entropy therefore implies
\begin{equation}\label{eq:thermal-maximum}
 0\le D(\sigma_r\Vert\vartheta_{\nu_r})
 =-H(\sigma_r)+g(\nu_r),\qquad H(\sigma_r)\le g(\nu_r).
\end{equation}
If $\nu_r=0$, the state is supported on the zero-photon eigenspace and equals $\proj0$, so the same bound holds. Since $g$ is continuous, $\limsup_{r\to\infty}H(\sigma_r)\le g(\nu)$.

For the lower bound, list the eigenvalues $\lambda_q(\sigma_r)$ in decreasing order, padded by zeros, and define $h(x)\coloneqq-x\log_2x$ with $h(0)\coloneqq0$. Trace-norm convergence implies convergence of each ordered eigenvalue: the min--max principle bounds its change by $\norm{\sigma_r-\vartheta_\nu}$, which is at most the trace norm. For every fixed nonnegative integer $L$ and every $r$,
\begin{equation}\label{eq:entropy-partial-sum}
 H(\sigma_r)=\sum_{q=0}^{\infty}h(\lambda_q(\sigma_r))
 \ge\sum_{q=0}^{L}h(\lambda_q(\sigma_r)),
\end{equation}
because each summand is nonnegative. Taking a lower limit on the left, and using continuity of $h$ in the finite sum on the right, gives
\begin{align}
 \liminf_{r\to\infty}H(\sigma_r)
 &\ge\lim_{r\to\infty}\sum_{q=0}^{L}h(\lambda_q(\sigma_r))\label{eq:entropy-lower-approx}\\
 &=\sum_{q=0}^{L}h(\lambda_q(\vartheta_\nu)).\label{eq:entropy-lower-finite}
\end{align}
Only now let $L\to\infty$. The right side increases to $H(\vartheta_\nu)=g(\nu)$, proving $\liminf_{r\to\infty}H(\sigma_r)\ge g(\nu)$. Together with the upper bound, this establishes $H(\sigma_r)\to g(\nu)$ as $r\to\infty$. No interchange of a limit with an infinite entropy sum has been assumed.
\end{proof}

\begin{proof}[Proof of \thref{thm:bosonic-two}]
First take positive rates strictly inside \eqref{eq:bosonic-two-region}--\eqref{eq:line-9-16-3}. By continuity of $g$, choose $0<P_i<N_i$ so that the same three strict inequalities hold with $P_i$ in place of $N_i$. Apply \lemref{lem:gaussian-approx} independently to these two Gaussian variables. For the fixed coefficients $(\sqrt\eta,0)$, $(0,\sqrt{1-\eta})$, and $(\sqrt\eta,\sqrt{1-\eta})$, the lemma gives one sequence of finite input constellations for which, as $r\to\infty$ at these fixed powers and mixing coefficients,
\begin{equation}\label{eq:two-entropy-limits}
 H(\tau_1^{(r)})\to g(\eta P_1),\qquad
 H(\tau_2^{(r)})\to g((1-\eta)P_2),\qquad
 H(\tau_{12}^{(r)})\to g(\eta P_1+(1-\eta)P_2).
\end{equation}
Fix $r$ large enough that all three target rates satisfy \eqref{eq:finite-optical-region} strictly. Each constellation has mean energy at most $P_i<N_i$. Proposition~\ref{prop:finite-optical} now gives vanishing-error deterministic codes satisfying the actual energy constraints. The order is important: fix the interior powers and a finite constellation first, and only then let the blocklength grow. No approximation error is multiplied by $n$.

If one rate is zero, that sender can use vacuum, and only the positive-rate sender's individual inequality is needed. If $N_i=0$, the corresponding individual bound forces $R_i=0$. For boundary rates, scale all positive rates by $1-1/j$, choose interior powers and finite constellations for each $j$, and then choose increasing blocklength thresholds with error at most $1/j$. Between consecutive thresholds use the corresponding fixed constellation. This is the same whole-block diagonal construction as in \sref{subsec:two-boundary}; it preserves the energy constraints at every blocklength and attains the closed region.
\end{proof}

\subsection{\texorpdfstring{$k$ senders at a passive $k$-port splitter}{k senders at a passive k-port splitter}}\label{subsec:k-splitter}
Let a passive $k$-mode interferometer be specified by a unitary matrix $U$ acting on annihilation operators. The receiver obtains the first output mode,
\begin{equation}\label{eq:k-splitter}
 \hat b=\sum_{i=1}^k u_i\hat a_i,\qquad
 u_i\coloneqq U_{1i},\qquad \sum_i|u_i|^2=1.
\end{equation}
The remaining $k-1$ outputs are inaccessible. Conversely, every unit row $(u_i)_i$ can be completed to a unitary matrix. This is the meaning of a $k$-splitter here; it need not be balanced. It can be realized by two-port beamsplitters and phase shifts~\cite{Serafini2023}. Write $\eta_i\coloneqq |u_i|^2$. Known phases of $u_i$ can be absorbed into the input amplitudes, so one may use $u_i=\sqrt{\eta_i}$ when desired. All $k$ ports are sender inputs; there are no additional thermal input modes in this model.

Figure~\ref{fig:k-splitter} depicts a realization by successive two-mode beamsplitters. A product of coherent inputs remains a product of coherent states at all outputs, with amplitudes transformed by $U$. Denote the received amplitude by $\gamma(z)\coloneqq\sum_{i=1}^k u_i z_i$. The channel is therefore
\begin{equation}\label{eq:k-optical-channel}
 (z_1,\ldots,z_k)\longmapsto\rho_z=\proj{\gamma(z)}_B.
\end{equation}
Sender~$i$ obeys \eqref{eq:optical-energy}, now for every $i\in[k]$. The received block contains $n$ modes, not $kn$ accessible modes. The receiver makes one measurement on those $n$ modes and estimates the full message tuple.

\begin{figure}[!t]
\centering
\begin{adjustbox}{max width=\linewidth}
\begin{tikzpicture}[font=\small]
 % A cascade is sufficient to synthesize the one received row, with known phases.
 \draw[quantum,->] (-.1,0)--(2.15,0);
 \node[anchor=east] at (-.2,0) {$\ket{z_1}$};
 \foreach \x in {2.3,6.4}{
  \draw[fill=softblue,draw=inkblue!40] (\x-.48,-.48) rectangle (\x+.48,.48);
  \draw[inkblue,line width=1.8pt] (\x-.42,-.42)--(\x+.42,.42);
  \draw[quantum,->] (\x,-1.93)--(\x,-.15);
  \draw[quantum,->] (\x,.15)--(\x,1.67);
  \draw[line width=1pt] (\x-.18,1.72)--(\x+.18,1.72);
 }
 \draw[quantum,->] (2.45,0)--(3.80,0);
 \node at (4.39,0) {$\cdots$};
 \draw[quantum,->] (4.94,0)--(6.25,0);
 \node[anchor=north] at (2.3,-2.02) {$\ket{z_2}$};
 \node[anchor=north] at (6.4,-2.02) {$\ket{z_k}$};
 \node[align=center,font=\footnotesize] at (4.4,1.20) {other outputs\\not received};
 \node[align=center,font=\footnotesize] at (4.4,-1.08) {additional inputs\\and beamsplitters};
 \draw[quantum,->] (6.55,0)--(10.0,0);
 \node[anchor=south] at (8.28,.3) {$B:\ket{\gamma(z)}$};
 \node[operator,text width=2.3cm,minimum height=1.0cm] (dec) at (11.25,0)
 {One collective\\measurement on $B^n$};
 \draw[classical,->] (dec.east)--(13.7,0) node[right] {$\widehat m$};
 \node[align=center,anchor=north,font=\footnotesize,text width=12.4cm] at (6.2,-2.88)
 {A cascade of two-port beamsplitters and known phase shifts realizes the received mixing row.\\Phase shifters and the intermediate discarded ports are suppressed for clarity.};
\end{tikzpicture}
\end{adjustbox}
\caption{A physical realization of the received row of a passive $k$-port splitter. Sequential two-port mixing combines the inputs into one output with amplitude $\sum_i u_i z_i$; known phase shifts allow complex coefficients. The drawing depicts the first and last couplers of this realization for $k\ge3$, with the intermediate couplers suppressed. For $k=2$ there is one beamsplitter, and for $k=1$ no mixing is needed. The complete interferometer has $k-1$ other outputs, all inaccessible. For error set $S$, the corresponding independent Gaussian inputs contribute mean photon number $\sum_{i\in S}\eta_iN_i$ to the conditional mixture. This conditioning is part of the analysis, not side information supplied to the receiver.}
\label{fig:k-splitter}
\end{figure}

\FloatBarrier
\Needspace{10\baselineskip}
\begin{theorem}[Coherent-state simultaneous decoding at a $k$-splitter]\label{thm:bosonic-k}
Fix a positive integer $k$, coefficients $u_i\in\cc$ with $\sum_i|u_i|^2=1$, and finite input energy limits $N_i\ge0$. Consider the channel \eqref{eq:k-optical-channel}, with only the mode $B$ delivered to the receiver. Let $\eta_i\coloneqq|u_i|^2$ and let $g$ be given by \eqref{eq:g-definition}. Under the uniform-message average photon-number constraints \eqref{eq:optical-energy} for every sender, every nonnegative rate tuple satisfying
\begin{equation}\label{eq:k-optical-region}
 \sum_{i\in S}R_i\le g\!\left(\sum_{i\in S}\eta_iN_i\right)
 \qquad(\varnothing\ne S\subseteq[k])
\end{equation}
is achievable by coherent-state product codewords and one collective decoding POVM on $B^n$. Error is the uniform-message probability in \eqref{eq:k-code-error}. The codes can be selected from independent finite-constellation ensembles; the constellation and blocklength are then taken to their limits as in \thref{thm:bosonic-two}. Senders with $\eta_iN_i=0$ have rate zero and may use vacuum. No rate splitting or time sharing is used.
\end{theorem}
\begin{proof}
For finite independent constellations, the exact received signal span has dimension at most $\prod_i|\mathcal A_i|$. Each $p_i$ is a probability distribution on $\mathcal A_i$, and $p_S(z_S)\coloneqq\prod_{i\in S}p_i(z_i)$. Define the partial amplitude $\gamma_S(z_S)\coloneqq\sum_{i\in S}u_i z_i$, with $\gamma_\varnothing=0$. For a nonempty $S\subseteq[k]$, the mixture contributed by these senders is
\begin{equation}\label{eq:optical-subset-mixture}
 \tau_S\coloneqq\sum_{z_S}p_S(z_S)\proj{\gamma_S(z_S)},
 \qquad \gamma(z)=\gamma_S(z_S)+\gamma_{S^c}(z_{S^c}).
\end{equation}
Conditioning on $z_{S^c}$ leaves
\begin{equation}\label{eq:k-optical-conditional}
 \overline\rho_{z_{S^c}}
 =D(\gamma_{S^c}(z_{S^c}))\tau_S D(\gamma_{S^c}(z_{S^c}))^\dagger.
\end{equation}
Every fully conditioned state is pure. The entropy formula for a classical--quantum state, followed by unitary invariance, gives
\begin{align}
I(X_S;B|X_{S^c})_\omega
 &=\sum_{z_{S^c}}p_{S^c}(z_{S^c})H(\overline\rho_{z_{S^c}})
      -\sum_zp(z)H(\rho_z)\label{eq:k-optical-information}\\
&=H(\tau_S).\label{eq:line-9-31-2}
\end{align}
Hence \thref{thm:k-capacity} applies to all the subset entropies simultaneously. The state-based decoder is the one proved earlier, restricted to this exact finite signal span and extended by a failure effect on its complement.

First assume that every inequality indexed by a nonempty subset of positive-rate senders is strict. Choose $0<P_i<N_i$ for those senders so that these finitely many inequalities remain strict. All zero-rate senders may use vacuum. Independent centered Gaussian inputs give
\begin{equation}\label{eq:k-optical-thermal}
 \sum_{i\in S}u_i Z_i\ \text{circular Gaussian with variance }\ 
 \nu_S\coloneqq\sum_{i\in S}\eta_iP_i,
 \qquad \tau_S=\vartheta_{\nu_S}.
\end{equation}
In particular, $H(\tau_S)=g(\nu_S)$. Apply \lemref{lem:gaussian-approx} to the coefficient choices $u_i\mathbf1_{\{i\in S\}}$ for all nonempty subsets of the positive-rate senders. There are finitely many such subsets because $k$ is fixed. Thus one sufficiently fine finite constellation for each sender makes \emph{all} of their entropies close enough to $g(\nu_S)$ to preserve the strict rate inequalities. No continuous classical-register entropy is used in this step.

Now hold the chosen finite constellations, their probability distributions, $k$, and all channel parameters fixed, and let the coding blocklength $n\to\infty$. Each selected input distribution has mean cost $\mu_i\le P_i<N_i$ and finite variance $v_i\coloneqq\operatorname{Var}(|z_i|^2)$. The union bound and the same Chebyshev calculation as in \eqref{eq:energy-chebyshev} show that the probability of any positive-rate sender violating its energy constraint is at most
\begin{equation}\label{eq:k-energy-selection}
 \sum_{i:R_i>0}\frac{v_i}{nM_i(N_i-\mu_i)^2}
 \xrightarrow[n\to\infty]{}0.
\end{equation}
Combining this estimate with the vanishing expected decoding error, as in \eqref{eq:energy-chebyshev}, selects deterministic codebooks meeting all energy constraints. If $R_i=0$, a single vacuum codeword has zero cost. Constraints involving these senders are no tighter than the corresponding constraints on the positive-rate senders alone, because $g$ is nondecreasing.

Finally, scale each positive target rate by $1-1/j$. Every nonempty subset of positive-rate senders then has a strict rate inequality at the limiting powers. Choose interior powers, finite constellations, and increasing blocklength thresholds for each $j$, exactly as in the proof of \thref{thm:bosonic-two}. This achieves the closed region \eqref{eq:k-optical-region}. If there are no positive rates, the unique message tuple is decoded without error.
\end{proof}

For a balanced splitter, $\eta_i=1/k$. When all senders have the same energy limit $N$, the symmetric rate $R_i=g(N)/k$ satisfies every inequality: concavity and $g(0)=0$ give
\begin{equation}\label{eq:balanced-optical}
 g(|S|N/k)\ge\frac{|S|}{k}g(N)=\sum_{i\in S}R_i.
\end{equation}
For the centered independent Gaussian ensemble, the received mean photon number is $N$, not $kN$, because only one output is received. This illustrates why both the mixing coefficients and the receiver's accessible modes must be specified.

The results in this section are achievability statements for coherent-state encoding. They do not claim a new optical capacity formula, optimality over arbitrary nonclassical inputs, an efficient optical implementation of the decoder, or a result under a stronger per-codeword photon constraint. The phase-space amplitudes of the optical signals and the weak auxiliary labels play different roles: the former are physically transmitted, while the latter are removed from the completed measurement. The application shows that the same geometric proof applies to familiar nonorthogonal pure-state signals as well as to finite-dimensional mixed-state channels.

\section{Relation to earlier approaches and scope}\label{sec:relation}
The applications illustrate why it is useful to separate the state-based decoder from its binary analysis tests. We conclude the technical discussion by identifying the sources of these ingredients and the scope of the particular formulation proved here. 

\subsection{A special form of tilted-subspace geometry}
The construction has deliberately been proved without invoking a general joint typicality lemma. Nevertheless, its geometry is closely related to Sen's tilting construction, including the upper-triangular matrix tilting developed in~\cite[Section~3]{Sen2021Joint}. We make this relationship explicit rather than leaving it as an informal analogy.

Fix a candidate $\ell$ and an error set $S$, and write $U\coloneqq S^c$. For $T\subseteq[k]$, define the auxiliary basis vector $\ket{\ell_T}$ to have coordinate $\ket{\ell_i}$ when $i\in T$ and $\ket0$ otherwise. Expanding the product defining $\ket{w_{\ell,S}}$ gives
\begin{equation}\label{eq:tilting-expansion}
 \ket{w_{\ell,S}}
 =\sum_{T\subseteq U}
       (1-\delta)^{(|U|-|T|)/2}\delta^{|T|/2}\ket{\ell_T}.
\end{equation}
Thus the isometry $\ket h\mapsto\ket h\otimes\ket{w_{\ell,S}}$ embeds a copy of the output space into a reference copy and several mutually orthogonal auxiliary copies. The nonreference squared coefficients are
\begin{equation}
 a_{T,U}\coloneqq
 \begin{cases}
  \delta^{|T|}(1-\delta)^{|U|-|T|},&\varnothing\ne T\subseteq U,\\
  0,&\text{otherwise}.
 \end{cases}
\end{equation}
Order the nonempty subsets by a linear extension of inclusion. The matrix $(a_{T,U})_{T,U}$ is upper triangular and obeys
\begin{align}
a_{T,U}&\le a_{T,T},\label{eq:line-10-3-1}\\
\sum_{\varnothing\ne T\subseteq U}a_{T,U}
 &=\sum_{j=1}^{|U|}\binom{|U|}{j}\delta^j(1-\delta)^{|U|-j}
 =1-(1-\delta)^{|U|}\le1.\label{eq:line-10-3-2}
\end{align}
The remaining reference amplitude is $(1-\delta)^{|U|/2}$. These are the triangularity, entrywise diagonal bounds, and column-sum conditions in the matrix-tilting construction. When $U=\varnothing$, the vector is the untilted reference vector. Our Gram-matrix argument is an elementary estimate adapted to this particular family of embeddings.

Sen's multiple-access proof already associates auxiliary labels with message-indexed codewords, uses augmented states, and transfers the resulting decoder to the original channel~\cite[Section~4]{Sen2021Joint}. Neither tilting nor that final transfer should be regarded as a new ingredient here. The specialization in the present proof is the deterministic assignment of a distinct auxiliary basis vector to every message index, followed by a pairwise error calculation before averaging over the codebook. Equation~\eqref{eq:k-exact} is therefore exact, including when the sampled channel input sequences coincide. No estimate for an average over random auxiliary labels is needed.

The candidate detector is the full projector onto the complement of the labeled rejection space inside the local space associated with $\ell$. This choice allows the correct-message estimate to follow directly from \lemref{lem:gram}. Together with the exact pairwise identity, it gives a short packing argument for independent multiple-access codebooks. We view this as a direct specialization of the geometric method, not as a replacement for the more general results in~\cite{Sen2021Joint,Sen2021Network}.

\subsection{The role of separately chosen binary tests}
The connection with position-based coding, introduced by Anshu, Jain, and Warsi~\cite{AnshuJainWarsi2019} and developed for multiple-access channels by Qi \textit{et al}.~\cite{QiWangWilde2018}, lies in the separation of binary testing from the construction of a full decoder. Sen's displayed MAC decoder normalizes acceptance effects obtained from tilted subspaces~\cite[Section~4]{Sen2021Joint}. Here the final step uses the PGM of the tilted signal density operators themselves, analyzed using Cheng's trace inequality~\cite{Cheng2023}; the geometric projectors serve only as analysis tests. Cheng's original MAC bound instead normalizes the original channel output states. This identifies the particular synthesis studied here without asserting historical priority for every specialization of these methods. Cheng's trace bound does not require the binary analysis tests for different error events to commute. What is needed is a candidate operator with both a high correct-message acceptance probability and an appropriate bound for every error event.

Here the individual tests are chosen for the ordinary binary alternatives $\sigma^{(S)}$. Each has its own Stein exponent $D(\omega\Vert\sigma^{(S)})$. The auxiliary labels allow these tests to be used together without reducing any of their false-acceptance exponents. The only loss is in the correct-message term and in the final perturbation term, both of which vanish by taking the limits in the order explained after~\eqref{eq:k-asymptotic}.

This argument does not construct a single test on the original input--output state that simultaneously has all the binary testing exponents. It constructs one decoding POVM for a codebook. The codebook's Cartesian product structure and the dependence of the auxiliary label on the candidate are essential to the proof. Distinguishing these two tasks explains why separate binary Stein lemmas can suffice here without a general assertion about composite quantum hypothesis testing.

\subsection{Subsequent uses and related state-discrimination constructions}
Sen's joint typicality framework has been used for multiparty packing and relay communication~\cite{DingEtAl2020}, computation of sums over classical--quantum multiple-access channels~\cite{SohailEtAl2022}, simultaneous distributed compression~\cite[Appendix~A]{SalekWinter2022}, and structured-code constructions for three-user broadcast and interference channels~\cite{GouiaaPadakandlaBroadcast2025,GouiaaPadakandla2025}. Sen's later work applies tilting, flattening, and smoothing to multipartite covering and decoupling~\cite{Sen2025Covering}. These applications use more general structures than the independent-codebook proof treated here. In particular, the present packing theorem should not be read as a replacement for all the intersection, union, and smoothing conclusions of the joint typicality lemma.

Geometrically lifted signal ensembles and pretty-good measurements also occur outside network coding. Chitambar, Duan, and Hsieh analyze the PGM of a lifted-trine ensemble in their study of local state discrimination~\cite[Section~IV-B3]{ChitambarDuanHsieh2014}. This is a related state-discrimination precedent, not an application of Sen's later construction. Accordingly, neither the use of auxiliary directions nor applying a PGM to an enlarged signal ensemble is asserted to be a new general principle here. The present manuscript develops a particular MAC realization, its deterministic error bound, and reductions of three communication tasks to that bound.

The quantitative refinement in \sref{subsubsec:sharp-packing} is a more precise analysis of the same labeled-state PGM, not a replacement for Sen's geometric principle or Cheng's measurement inequality. The new rescaling exploits the correct-label overlaps, while the transfer calculation is a two-dimensional binary-testing argument. Appendix~\ref{app:constants} proves sharpness only for the fixed-label geometric coefficient and for the information-limited transfer problem. We do not claim that the full decoding bound is optimal or uniformly better than all prior coding bounds.

\subsection{Quantitative dependence on the number of senders}\label{subsec:sen-comparison}
The explicit error constants provide a further reason to specialize the general joint typicality framework to the present packing problem. A comparison must nevertheless fix both the number of classical registers and the meaning of the auxiliary parameters. We first compare displayed estimates, and then state the limits of this comparison.

Sen's general one-quantum-output joint typicality statement~\cite[Section~3, Fact~4]{Sen2021Network} has $c$ classical registers. For the $k$-sender classical-input channel we set $c=k$. Write $t$ for its auxiliary parameter, reserving $\delta$ for the labels in this paper. For $0<\epsilon<1$ and $0<t<1$, when the separate binary-test rejection tolerances are at most $\epsilon$, the displayed correct-rejection estimate and trace-norm perturbation estimate are, respectively,
\begin{equation}\label{eq:sen-displayed-estimates}
 r_{\mathrm{Sen}}\le A_k\epsilon/t^2+B_kt,
 \qquad\norm{\omega'-\omega_{\mathrm{emb}}}_1\le B_kt,
\end{equation}
where $\omega_{\mathrm{emb}}$ is the stated fixed-auxiliary embedding and
\begin{equation}\label{eq:sen-constants}
 A_k\coloneqq3^k2^{\,2^{k+5}},\qquad B_k\coloneqq2^{(k+3)/2}.
\end{equation}
In particular, $A_k$ is doubly exponential in $k$. This is a bound in a general joint typicality lemma, not a lower bound on the error of Sen's decoder.

For comparison on the same channel and input distribution, let
\begin{equation}\label{eq:comparison-beta-sum}
 b_\epsilon\coloneqq\sum_{\varnothing\ne S\subseteq[k]}
       N_S2^{-D_H^\epsilon(\omega\Vert\sigma^{(S)})},
 \qquad N_S\coloneqq\prod_{i\in S}(M_i-1).
\end{equation}
For the same product input distribution, message sizes, and testing tolerance, a direct use of \eqref{eq:sen-displayed-estimates} in the independent-codebook packing reduction, with the standard Hayashi--Nagaoka error analysis~\cite{HayashiNagaoka2003}, gives
\begin{equation}\label{eq:sen-direct-packing-comparison}
 \ee\Pe^{\mathrm{Sen}}
 \le 2A_k\epsilon/t^2+\tfrac52 B_kt+4b_\epsilon.
\end{equation}
Here twice the correct-rejection estimate contributes $2B_kt$, and transferring a completed measurement costs at most half the trace norm, namely $B_kt/2$. For an error set $S$, choose the product alternative in that lemma with the input coordinates in $S$ separated from the other coordinates and the output. Its alternative-acceptance guarantee then gives $b_\epsilon$ after codebook averaging. We quote this direct consequence only to compare explicit guarantees; it is not asserted to be the best consequence of Sen's method. The present one-shot corollary instead gives
\begin{equation}\label{eq:our-direct-packing-comparison}
 \ee\Pe\le d_k(\delta)+(2^k-1)C_k(\delta)\epsilon+b_\epsilon.
\end{equation}
At fixed $\delta$, the coefficient $(2^k-1)C_k(\delta)$ grows exponentially in $k$, rather than doubly exponentially. The refined coefficient in \sref{subsubsec:sharp-packing} further reduces the correct-rejection term.

Holding $\delta$ fixed is not a fair comparison at arbitrarily large $k$, however: $d_k(\delta)$ then tends to one. To keep the additive terms fixed, choose an arbitrary $0<\zeta<1$. Set $\delta=\zeta^2/(4k)$ in our estimate and $t=\zeta/(5B_k)$ in \eqref{eq:sen-direct-packing-comparison}. The elementary bounds proved after \thref{thm:k-packing} give
\begin{align}
 \ee\Pe&\le\frac\zeta2+
       (2^k-1)\left(\frac{8k}{\zeta^2}\right)^k\epsilon+b_\epsilon,
       \label{eq:matched-budget-ours}\\
 \ee\Pe^{\mathrm{Sen}}&\le\frac\zeta2+
       \frac{50A_kB_k^2}{\zeta^2}\epsilon+4b_\epsilon.
       \label{eq:matched-budget-sen}
\end{align}
For fixed $\zeta$, the coefficient of $\epsilon$ is therefore $2^{O(k\log k)}$ in the first estimate and $2^{\,2^{k+5}+O(k)}$ in the second. Thus the present reduction avoids the doubly exponential loss of the displayed general joint typicality estimate even after controlling label removal. This comparison does not establish a coding theorem in which the number of senders grows with blocklength; the coding theorems in this paper keep $k$ fixed.

The dependence on $\epsilon$ is a separate issue. For fixed $k$, optimizing the displayed Sen estimate over $t$ gives a term of order $\epsilon^{1/3}$. Our refined analysis gives a guarantee of order $\epsilon^{1/k}$ for $k\ge2$, up to $k$-dependent constants and the false-acceptance term. For example, this follows from $\kappa_k(\delta)=O_k(\delta^{-(k-1)})$ and $\mathcal T_d(x)\le2x+2d^2$, by choosing $\delta$ proportional to $\epsilon^{1/k}$. The two statements therefore do not imply uniform dominance over all tolerances and all $k$. Specialized estimates, including Sen's separate two-sender bound~\cite[Section~4]{Sen2021Joint}, also need their own comparisons.

Later uses retain large general joint-typicality constants: for example, Salek and Winter's simultaneous-compression argument quotes an error parameter $2^{k+2^{k+5}}\sqrt\epsilon$~\cite[Appendix~A, Lemma~20]{SalekWinter2022}. Their random-binning task is covered by \sref{sec:compression}, so it admits the present one-shot guarantees without that factor. More broadly, the reductions proved here---independent-codebook communication, position-based entanglement assistance, distributed hashing, and finite optical constellations---all inherit the present constants. An arbitrary later use of Sen's lemma does not automatically do so: union statements, nonunique decoding, other codebook dependencies, covering, and decoupling can require additional estimates.

Sen's later covering and decoupling papers~\cite{Sen2025Covering,Sen2025Telescoping} develop different multipartite tasks; the soft-covering paper still invokes the earlier simultaneous decoder in its interference-channel application. We have not identified in those works a replacement for the $k$-dependent classical-input packing estimate compared above. This is a comparison with the stated results and reductions, not a claim that Sen's underlying geometry must incur a doubly exponential loss.

\subsection{Limitations of the statement}
For communication, the error criterion is average error under independent uniform messages; the weighted packing corollary and compression application allow general source priors. We do not establish maximal-error achievability. The packing theorem is proved for finite alphabets and finite-dimensional outputs, with a fixed number of senders. The bosonic application in \sref{sec:bosonic} uses exact finite-dimensional signal spans and an explicitly controlled constellation limit. We do not assert a polynomial-time implementation of the decoding POVM.

The auxiliary dimension $\prod_i(M_i+1)$ can grow exponentially with blocklength. This does not represent additional communication, entanglement, or a supplied message-dependent state. After compression, the POVM acts on the receiver's actual systems: $B^n$ for unassisted and optical communication, the reference systems together with $B^n$ for entanglement-assisted communication, or $B^n$ conditioned on the transmitted hashes for compression. No fictitious message-dependent label remains. The reason the auxiliary dimension causes no difficulty in the error bound is that a fixed candidate message tuple uses only a $2^k$-dimensional auxiliary space, and its Gram matrix depends only on $k$ and $\delta$.

Finally, the finite-dimensional geometric proof does not establish all the intersection and union statements of the quantum joint typicality lemma or their applications to more general network codebooks. Those questions involve additional structure. For the independent-codebook MAC and the two additional reductions proved here, the message-index identity gives the needed pairwise estimate without an additional average over random auxiliary labels.

\FloatBarrier

\Needspace{13\baselineskip}
\section{Conclusion}\label{sec:conclusion}

\subsection{Summary of contributions}
The proof has two complementary components. Sen's tilted-subspace geometry supplies binary tests with controlled correct-message rejection and exact competing-message estimates~\cite{Sen2021Joint,Sen2021Network}. Cheng's trace inequality~\cite{Cheng2023}, related to Renes's binary near-optimality bound~\cite{Renes2017}, then certifies the pretty-good measurement of the tilted states. Removing the labels from the completed measurement incurs one perturbation cost per transmitted message, rather than one per competitor.

The resulting deterministic packing bound applies before a channel or a random codebook is specified. Ordinary independent codebooks, entanglement-assisted position-based coding, and classical hashing produce its competing-state terms in three different ways. Binary Stein tests yield the corresponding conditional-mutual-information or conditional-entropy constraints. Finite coherent-state constellations and an energy-controlled Gaussian limit extend the communication statement to a beamsplitter and a passive multiport splitter. These reductions make the same geometric idea usable without repeating its proof in each setting.

The explicit packing constants avoid the doubly exponential factor in the displayed general joint typicality estimate, with the parameter qualifications explained in \sref{subsec:sen-comparison}. This benefit carries through the reductions actually proved here, not through every possible network application without additional work. The basic theorem is supplemented by sharper finite-error bounds for the same decoder. Rescaling the error-label vectors gives the optimal universal coefficient for the specified correct-rejection inequality, and the sharp label-transfer function uses the known small tagged error. These component optimality statements do not imply global optimality of the complete coding bound. The contribution is a detailed development of this particular tilted-state PGM construction and its one-shot consequences, with the established geometry and coding regions attributed to their original sources.

\subsection{Future directions}
The separation between the decoder and its analysis tests provides several concrete directions for further work. The questions below concern quantitative refinements and extensions of the construction, rather than a new claim about the already established first-order rate regions.

\paragraph{Achievable error exponents.}
One can choose binary tests whose rejection and false-acceptance probabilities decay exponentially with blocklength, and let the label strength depend on blocklength. Optimizing these choices could give useful bounds on the exponential decay of decoding error at rates strictly inside the multiple-access region. A careful comparison should distinguish improvements to the individual testing exponents from losses incurred when combining the tests. The error-exponent analysis of position-based coding in~\cite{QiWangWilde2018} provides a natural point of comparison.

\paragraph{Second-order coding rates.}
At a fixed nonzero target error, the next question is how closely finite-blocklength rates can approach a boundary point of the region. The one-shot bounds permit separately chosen testing tolerances, but sharp second-order rates may require keeping track of the relations among several active error constraints, especially at corners. Determining achievable offsets of order $\sqrt n$ in the logarithms of message sizes, and comparing them with suitable converses, would clarify the quantitative cost of this simultaneous-decoding analysis. The classical multiple-access dispersion analysis of Tan and Kosut~\cite{TanKosut2014} provides a direct network-coding precedent, while \cite{QiWangWilde2018} develops the connection between binary testing and quantum second-order rates.

\paragraph{Implementable receivers and resource use.}
The existence of the compressed PGM does not establish an efficient circuit or an optical implementation. For structured signal ensembles, it would be useful to identify approximations that retain a controlled error bound while reducing measurement complexity. Coherent-state receivers such as those considered in~\cite{WildeGuha2012} provide a concrete setting for this investigation. The original position-based coding work already reduces entanglement consumption for point-to-point communication~\cite{AnshuJainWarsi2019}. Extending such resource savings to the present multiple-access construction while retaining its simultaneous-decoding estimates is a separate question.

\paragraph{Correlated codebooks and network protocols.}
The deterministic geometry applies to arbitrary message-indexed states, whereas the simple evaluation of false-acceptance terms uses the structure of each protocol. Superposition coding, structured codebooks, and several receivers introduce dependencies that require new averaging identities or different binary alternatives. It would be valuable to determine which of these settings admit comparably direct reductions, and to relate them to the broader network constructions of~\cite{Sen2021Network,DingEtAl2020,GouiaaPadakandla2025}.

\Needspace{10\baselineskip}
\section*{Acknowledgments and AI assistance}
ChatGPT Pro 6 (Astra) was used extensively and interactively for proof formulation and exploration, preparation and revision of the exposition, preparation of LaTeX, TikZ, and numerical-checking code, literature searches, and bibliographic checks. The research problem and successive revisions were directed by the author, who supplied detailed comments on the mathematical arguments and their presentation.

AI-assisted mathematical and editorial review passes were followed by further revisions. These were internal self-reviews, not independent external peer review or formal verification. The author is responsible for the final manuscript, its mathematical claims, and its attributions.

\begingroup
\small
\setlength{\bibsep}{5pt}
\bibliographystyle{unsrtnat}
\bibliography{references}
\endgroup
\clearpage
\appendix
\section{Quantitative refinements of the packing bound}\label{app:quantitative}
The main proof establishes simultaneous decoding using estimates that are transparent to recognize and combine. This appendix extracts more information from the same construction. We first adapt the Gram estimate to the correct label, then derive the sharp transfer from tagged error to actual error, and finally combine these results. The last subsection gives explicit constants and identifies precisely which component estimates are optimal. Throughout, the basic decoder remains the pretty-good measurement of the tilted signal states, compressed against the reference label.

\subsection{A Gram estimate adapted to the correct label}\label{app:rescaled-gram}
Fix a message tuple $m$ and the label vectors in \eqref{eq:k-labels} and \eqref{eq:k-w}. The correct-rejection calculation needs only the expectation of $\Pi_{\cW_m}$ on $\rho_m\otimes\proj{t_m}$. It need not bound every label direction equally well. This observation suggests using the overlap with $\ket{t_m}$ to scale the vectors spanning $\cW_m$.

For each nonempty $S\subseteq[k]$, define
\begin{equation}\label{eq:rescaled-label}
 \ket{z_{m,S}}\coloneqq(1-\delta)^{-|S|/2}\ket{w_{m,S}}.
\end{equation}
The scalar is finite and nonzero because $0<\delta<1$. Consequently,
\begin{equation}\label{eq:rescaled-span-overlap}
 \spanop\{\ket{z_{m,S}}\}=\spanop\{\ket{w_{m,S}}\},
 \qquad \langle z_{m,S}|t_m\rangle=1.
\end{equation}
The vectors $\ket{z_{m,S}}$ need not be normalized, but \lemref{lem:gram} allows this. Their Gram entries are
\begin{align}
\langle z_{m,S}|z_{m,T}\rangle
 &=(1-\delta)^{-(|S|+|T|)/2}
       \langle w_{m,S}|w_{m,T}\rangle\label{eq:rescaled-gram-calculation}\\
&=(1-\delta)^{-(|S|+|T|)/2+|S\triangle T|/2}
  =(1-\delta)^{-|S\cap T|}.\label{eq:line-A-3-2}
\end{align}
We used $|S\triangle T|=|S|+|T|-2|S\cap T|$. This proves the interpretation of $\mathsf L_k(\delta)$ in \eqref{eq:sharp-gram-definition}. The original error vectors are linearly independent, and nonzero rescaling preserves linear independence; hence this Gram matrix is positive definite.

Applying \lemref{lem:gram} with these rescaled vectors gives
\begin{equation}\label{eq:rescaled-projector-bound}
 \Pi_{\cW_m}\le\kappa_k(\delta)
       \sum_{S\ne\varnothing}Q_{S,m}\otimes\proj{z_{m,S}}.
\end{equation}
Taking its expectation in the correct tagged state and using \eqref{eq:rescaled-span-overlap} yields
\begin{align}
\Tr[\Pi_{\cW_m}(\rho_m\otimes\proj{t_m})]
 &\le\kappa_k(\delta)\sum_{S\ne\varnothing}
       \Tr[Q_{S,m}\rho_m]|
          \langle z_{m,S}|t_m\rangle|^2\label{eq:sharp-correct-rejection}\\
&=\kappa_k(\delta)\sum_{S\ne\varnothing}\Tr[Q_{S,m}\rho_m].\label{eq:line-A-5-2}
\end{align}
This is the only change to the correct-message geometry. The subspace $\cW_m$, its projector, and every decoder effect are unchanged.

For comparison with the basic coefficient, let $G_{\mathrm{err}}$ be the Gram matrix of the normalized error vectors, and let $D_\delta$ denote the diagonal matrix with entries $(1-\delta)^{|S|/2}$. Then
\begin{equation}\label{eq:gram-congruence}
 \mathsf L_k=D_\delta^{-1}G_{\mathrm{err}}D_\delta^{-1},
 \qquad D_\delta^2\le(1-\delta)I.
\end{equation}
The basic estimate $G_{\mathrm{err}}\ge C_k(\delta)^{-1}I$ therefore implies
\begin{equation}\label{eq:sharp-coefficient-comparison}
 \mathsf L_k\ge C_k(\delta)^{-1}D_\delta^{-2}
 \ge\frac{I}{(1-\delta)C_k(\delta)},
 \qquad \kappa_k(\delta)\le(1-\delta)C_k(\delta).
\end{equation}
The stronger explicit estimate in \eqref{eq:sharp-simple-coefficient} is proved in Appendix~\ref{app:constants}.

\subsection{Transferring a small error between two label states}\label{app:transfer}
The trace-distance comparison in the main text is exact as a bound for an arbitrary effect. It is not generally sharp when the effect's probability on one state is already known to be small. We now retain that information. The relevant optimization is only two dimensional, regardless of the output dimension.

\begin{lemma}[Sharp probability transfer]\label{lem:sharp-transfer}
Let $\rho$ be a state, and let $\ket v$ and $\ket w$ be normalized vectors with $|\langle v|w\rangle|=\sqrt{1-d^2}$ for $0<d<1$. For an effect $E$ on the joint space, define
\begin{equation}\label{eq:transfer-errors}
 x\coloneqq\Tr[E(\rho\otimes\proj v)],\qquad
 y\coloneqq\Tr[E(\rho\otimes\proj w)].
\end{equation}
Then $y\le\mathcal T_d(x)$, with $\mathcal T_d$ defined in \eqref{eq:transfer-function}. The same conclusion with $x$ replaced by $x_0$ holds when only $x\le x_0\le1$ is known. The function $\mathcal T_d$ is nondecreasing and concave. It is the smallest universal upper bound on $y$ given only $d$ and an upper bound $x_0$ on $x$.
\end{lemma}
\begin{proof}
Choose a purification $\ket\psi$ of $\rho$ and write $\ket a\coloneqq\ket\psi\ket v$ and $\ket b\coloneqq\ket\psi\ket w$. Extend $E$ by the identity on the purifying register. The two corresponding outcome probabilities are still $x,y$, and $|\langle a|b\rangle|=\sqrt{1-d^2}$. Insert $E+(I-E)=I$ into the overlap and apply Cauchy--Schwarz separately to the two terms:
\begin{align}
\sqrt{1-d^2}
 &=|\langle a|E|b\rangle+\langle a|(I-E)|b\rangle|\label{eq:binary-fidelity-transfer}\\
&\le\|\sqrt E\ket a\|\,\|\sqrt E\ket b\|
    +\|\sqrt{I-E}\ket a\|\,\|\sqrt{I-E}\ket b\|\label{eq:line-A-9-2}\\
&=\sqrt{xy}+\sqrt{(1-x)(1-y)}.\label{eq:line-A-9-3}
\end{align}
Define $\alpha\coloneqq\arcsin\sqrt x$, $\beta\coloneqq\arcsin\sqrt y$, and $\vartheta\coloneqq\arcsin d$, all in $[0,\pi/2]$. The right-hand side is $\cos(\alpha-\beta)$ and the left-hand side is $\cos\vartheta$. Monotonicity of cosine on $[0,\pi/2]$ implies $|\alpha-\beta|\le\vartheta$. Hence
\begin{equation}\label{eq:transfer-angle-bound}
 y\le\sin^2\!\left(\min\{\pi/2,\alpha+\vartheta\}\right)
   =\mathcal T_d(x).
\end{equation}
The sine addition formula gives the first case of \eqref{eq:transfer-function}; the condition $\alpha+\vartheta\le\pi/2$ is equivalent to $x\le1-d^2$.

For completeness, on $0<x<1-d^2$ one may expand and differentiate:
\begin{align}
\mathcal T_d(x)
 &=d^2+(1-2d^2)x+2d\sqrt{1-d^2}\sqrt{x(1-x)},\label{eq:transfer-concavity}\\
\mathcal T_d''(x)
 &=-\frac{d\sqrt{1-d^2}}{2[x(1-x)]^{3/2}}<0.\label{eq:line-A-11-2}
\end{align}
The first derivative tends to zero at $x=1-d^2$, is positive to its left, and is zero on the constant part. Continuity at the junction proves monotonicity and concavity on the full interval. Monotonicity allows $x$ to be replaced by an upper bound.

To show sharpness, take a one-dimensional output space and choose phases so that
$\ket w=\sqrt{1-d^2}\ket v+d\ket{v_\perp}$ with $\ket{v_\perp}\perp\ket v$. For $0\le x_0\le1-d^2$, the rank-one effect onto
$\sqrt{x_0}\ket v+\sqrt{1-x_0}\ket{v_\perp}$ has $x=x_0$ and $y=\mathcal T_d(x_0)$. For $x_0\ge1-d^2$, the effect $\proj w$ has $x=1-d^2\le x_0$ and $y=1$. Thus no smaller bound based only on $d,x_0$ is possible.
\end{proof}

For the present label family, the overlap is $(1-\delta)^{k/2}$, so $d=d_k(\delta)$. Apply the lemma separately to each completed decoding error effect $I-\widetilde\Lambda_m$. Concavity then gives, for an arbitrary prior,
\begin{equation}\label{eq:transfer-prior-average}
 P_{\mathrm e}=\sum_m p_m y_m
 \le\sum_m p_m\mathcal T_{d_k(\delta)}(x_m)
 \le\mathcal T_{d_k(\delta)}\!\left(\sum_m p_m x_m\right).
\end{equation}
The final argument is exactly the tagged ensemble's average error. This step uses one completed error effect per transmitted message, not a separate perturbation comparison for every competitor.

An affine form is useful when a bound with the original additive structure is desired. For every $\lambda\ge0$, binary testing gives
\begin{equation}\label{eq:affine-transfer}
 y\le\lambda x+\gamma_\lambda(d),\qquad
 \gamma_\lambda(d)\coloneqq
 \frac{1-\lambda+\sqrt{(1-\lambda)^2+4\lambda d^2}}{2}.
\end{equation}
Indeed, the maximum of $y-\lambda x$ over effects is the trace of the positive part of
$\rho\otimes(\proj w-\lambda\proj v)$. Tensoring with the state $\rho$ leaves this trace unchanged. On the span of the labels, $\proj w-\lambda\proj v$ has trace $1-\lambda$ and determinant $-\lambda d^2$, so its positive eigenvalue is the displayed expression. The case $\lambda=1$ recovers $y\le x+d$. Sharpness in \lemref{lem:sharp-transfer} then also shows $\mathcal T_d(x)\le\min\{1,x+d\}$.

\subsection{Combined bounds and their use in the applications}\label{app:combine}
We now collect the two improvements without changing the state-based PGM. Let $(p_m,\rho_m)_{m\in\mathcal M}$ and the projectors satisfy the setting of \corref{cor:weighted-packing}. Retain the exact error-pattern factors by defining
\begin{equation}\label{eq:weighted-b-delta}
 b_{p,\delta}\coloneqq
 \sum_m p_m\sum_{\ell\ne m}(1-\delta)^{|S(m,\ell)|}
       \Tr[P_{S(m,\ell),\ell}\rho_m]
 \le(1-\delta)b_p.
\end{equation}
The inequality follows because every competing pair has $|S(m,\ell)|\ge1$. The calculation in \eqref{eq:k-wrong}--\eqref{eq:line-6-16-6} proves the false-acceptance bound with these factors before they are discarded. Replacing only the correct-message estimate by \eqref{eq:sharp-correct-rejection}--\eqref{eq:line-A-5-2}, the same weighted PGM analysis gives
\begin{equation}\label{eq:sharp-tagged-average}
 \widetilde P_{\mathrm e}
 \le X_{p,\delta}\coloneqq\kappa_k(\delta)a_p+b_{p,\delta}.
\end{equation}
Combining this with \eqref{eq:transfer-prior-average} proves
\begin{equation}\label{eq:sharp-weighted-full}
 P_{\mathrm e}\le
 \mathcal T_{d_k(\delta)}\!\left(\min\{1,X_{p,\delta}\}\right).
\end{equation}
For effects in place of projectors, the common dilation in \sref{subsec:general-packing} preserves the weighted testing sums, their error-pattern factors, and the compressed state-based PGM. Thus no additional loss is incurred. Equation~\eqref{eq:sharp-weighted-full} proves \corref{cor:sharp-packing}.

\begin{theorem}[Uniform-prior refinement with Cheng's quadratic correction]\label{thm:sharp-uniform}
In the setting of \thref{thm:k-packing}, or its general-effect extension, let $M\coloneqq\prod_iM_i\ge2$, define $b_\delta$ by \eqref{eq:weighted-b-delta} with $p_m=1/M$, and set
\begin{equation}\label{eq:uniform-sharp-input}
 u_\delta\coloneqq\min\{1,\kappa_k(\delta)a+b_\delta\}.
\end{equation}
The same tilted-state PGM, compressed to the original output, satisfies
\begin{equation}\label{eq:strongest-uniform}
 P_{\mathrm e}\le
 \mathcal T_{d_k(\delta)}\!\left(u_\delta-\frac{u_\delta^2}{M}\right).
\end{equation}
The functions are defined in \eqref{eq:k-constants}, \eqref{eq:sharp-gram-definition}, and \eqref{eq:transfer-function}. If $M=1$, assigning every measurement outcome to the unique message gives error zero.
\end{theorem}
\begin{proof}
Use the unnormalized uniform-prior convention for the tagged ensemble. For each $m$, let
$B_m\coloneqq\sum_{\ell\ne m}\widetilde\rho^\ell_{BL}$ and
$e_m\coloneqq\mathsf e(\widetilde\rho^m_{BL},B_m)$.
The two operators have total trace $M$. Since the test $T=0$ has cost one, $0\le e_m\le1$. The stronger part of \lemref{lem:cheng} therefore gives
\begin{equation}\label{eq:sharp-cheng-individual}
 1-\Tr[\widetilde\Lambda_m\widetilde\rho^m_{BL}]
 \le e_m\left(1-\frac{e_m}{M}\right).
\end{equation}
Using the analysis tests $\Gamma_m$ as in \eqref{eq:k-cheng-error}, and interchanging the two message indices in the false-acceptance sum, bounds the average binary cost by
\begin{equation}\label{eq:sharp-cost-average}
 \frac1M\sum_m e_m
 \le \kappa_k(\delta)a+b_\delta,
 \qquad \frac1M\sum_m e_m\le1.
\end{equation}
The function $f_M(x)\coloneqq x-x^2/M$ is concave on $[0,1]$ and nondecreasing there for $M\ge2$. Average \eqref{eq:sharp-cheng-individual} and apply these two properties:
\begin{equation}\label{eq:sharp-cheng-average}
 \widetilde P_{\mathrm e}
 \le\frac1M\sum_m f_M(e_m)
 \le f_M\!\left(\frac1M\sum_m e_m\right)
 \le f_M(u_\delta).
\end{equation}
Finally apply \eqref{eq:transfer-prior-average}. The quadratic correction is asserted here for uniform priors only; the arbitrary-prior estimate \eqref{eq:sharp-weighted-full} does not require it.
\end{proof}

For readers who prefer the form of the basic theorem, set $\lambda\coloneqq(1-\delta)^{-1}$ in \eqref{eq:affine-transfer}, and use \eqref{eq:weighted-b-delta}. This gives the additional bound
\begin{equation}\label{eq:sharp-affine-bound}
 P_{\mathrm e}\le\min\{1,\widehat d_k(\delta)+\widehat C_k(\delta)a_p+b_p\},
\end{equation}
where
\begin{align}
\widehat C_k(\delta)&\coloneqq\frac{\kappa_k(\delta)}{1-\delta},\label{eq:sharp-affine-constants}\\
\widehat d_k(\delta)&\coloneqq
 \frac{\sqrt{\delta^2+4(1-\delta)[1-(1-\delta)^k]}-\delta}{2(1-\delta)}.\label{eq:line-A-23-2}
\end{align}
Equation~\eqref{eq:sharp-coefficient-comparison} proves $\widehat C_k\le C_k$. Also $\widehat d_k<d_k$. To verify the latter inequality, the positive root $y=\gamma_\lambda(d)$ obeys
$y^2+(\lambda-1)y-\lambda d^2=0$.
At $y=d\in(0,1)$ this polynomial equals $(\lambda-1)d(1-d)>0$ when $\lambda>1$, whereas its positive root is smaller. Thus both constants can be improved while the coefficient of $b_p$ remains one. The nonlinear estimate retains more information and is at least as strong as this affine consequence when the same testing sums are used.

\paragraph{Random codebooks and hashes.}
Let $\xi$ denote any randomness in the ensemble, codebook, or hashes, with a fixed $\delta$ and fixed message sizes. The function
$x\mapsto\mathcal T_d(\min\{1,x\})$ on $[0,\infty)$ is nondecreasing and concave. Hence
\begin{equation}\label{eq:sharp-random-average}
 \ee_\xi P_{\mathrm e}(\xi)
 \le\mathcal T_{d_k(\delta)}\!\left(
   \min\{1,\kappa_k(\delta)\ee_\xi a_p(\xi)+\ee_\xi b_{p,\delta}(\xi)\}\right).
\end{equation}
For a uniform prior and fixed $M\ge2$, the composition
$x\mapsto\mathcal T_d(f_M(\min\{1,x\}))$ is also nondecreasing and concave, because each constituent has these properties on the relevant interval. Therefore the quadratic correction can likewise be retained after codebook averaging. In contrast, source priors in distributed compression are generally not uniform, so \eqref{eq:sharp-random-average} is the applicable formula there.

The substitutions needed by the applications are concise. With $N_S=\prod_{i\in S}(M_i-1)$, both ordinary random coding and entanglement-assisted position-based coding use
\begin{equation}\label{eq:sharp-channel-substitution}
 \ee a=\sum_{S\ne\varnothing}\alpha_S,
 \qquad
 \ee b_\delta=\sum_{S\ne\varnothing}
        (1-\delta)^{|S|}N_S\beta_S.
\end{equation}
For position-based coding there is no codebook expectation. For compression, the proof of \thref{thm:compression-one-shot} gives
\begin{equation}\label{eq:sharp-source-substitution}
 a_p=\sum_{S\ne\varnothing}\alpha_S,
 \qquad
 \ee_f b_{p,\delta}\le
 \sum_{S\ne\varnothing}(1-\delta)^{|S|}
            \frac{d_{X_S}}{L(S)}\beta_S.
\end{equation}
Each factor is constant on an exact error class, so it is preserved by the original averaging calculation. The same substitutions allow $\alpha_S\le\eps_S$ and $\beta_S=2^{-D_H^{\eps_S}}$. No first-order rate proof needs to be repeated. One may choose $\delta$ to minimize the available bound; different values specify different members of the already defined decoder family, not a single measurement chosen independently of $\delta$.

\subsection{Explicit coefficients and the scope of optimality}\label{app:constants}
The matrix in \eqref{eq:sharp-gram-definition} depends only on $k$ and $\delta$, so its smallest eigenvalue can be evaluated without using the channel output dimension. We next prove the explicit estimate \eqref{eq:sharp-simple-coefficient}, give a numerically stable alternative expression, and delimit the optimality statements.

\paragraph{An explicit inverse and a dimension-independent estimate.}
Include the empty subset temporarily. In binary subset order the full rescaled Gram matrix and its inverse are
\begin{equation}\label{eq:full-rescaled-gram}
 \mathsf L_{\mathrm{full}}
 =\begin{pmatrix}1&1\\1&(1-\delta)^{-1}\end{pmatrix}^{\otimes k},
 \qquad
 \mathsf L_{\mathrm{full}}^{-1}
 =\delta^{-k}
   \begin{pmatrix}1&-(1-\delta)\\-(1-\delta)&1-\delta\end{pmatrix}^{\otimes k}.
\end{equation}
Multiplying the two $2\times2$ factors verifies the inverse. Its $(S,T)$ entry is
$\delta^{-k}(-1)^{|S|+|T|}(1-\delta)^{|S\cup T|}$.
Write its block decomposition with the empty subset first as
$\mathsf L_{\mathrm{full}}^{-1}=\left(\begin{smallmatrix}a&r^\dagger\\r&D\end{smallmatrix}\right)$.
The lower-right block of $\mathsf L_{\mathrm{full}}$ is $\mathsf L_k$, and the block inverse identity gives $\mathsf L_k^{-1}=D-ra^{-1}r^\dagger$. Here $a=\delta^{-k}$ and $r_S=\delta^{-k}(-1)^{|S|}(1-\delta)^{|S|}$. Substitution proves
\begin{equation}\label{eq:rescaled-inverse-positive}
 \mathsf L_k^{-1}=J\mathsf B_kJ,
 \qquad J_{S,S}\coloneqq(-1)^{|S|},
\end{equation}
where
\begin{equation}\label{eq:positive-inverse-entries}
 (\mathsf B_k)_{S,T}\coloneqq
 \frac{(1-\delta)^{|S\cup T|}
       [1-(1-\delta)^{|S\cap T|}]}{\delta^k},
 \qquad S,T\ne\varnothing.
\end{equation}
In particular, $\mathsf B_k$ is entrywise nonnegative, and
$\kappa_k(\delta)=\lambda_{\max}(\mathsf B_k)$. The block inverse identity follows, for example, by multiplying
$\left(\begin{smallmatrix}1&0\\-ra^{-1}&I\end{smallmatrix}\right)
 \left(\begin{smallmatrix}a&r^\dagger\\r&D\end{smallmatrix}\right)
 \left(\begin{smallmatrix}1&-a^{-1}r^\dagger\\0&I\end{smallmatrix}\right)$:
the result is $\operatorname{diag}(a,D-ra^{-1}r^\dagger)$, whose inverse has lower-right block $\mathsf L_k$.

Since $S\cup T\ne\varnothing$ and $1-(1-\delta)^j\le j\delta$ for nonnegative integers $j$,
\begin{equation}\label{eq:incidence-entry-bound}
 0\le(\mathsf B_k)_{S,T}
 \le\frac{1-\delta}{\delta^{k-1}}|S\cap T|.
\end{equation}
Let $F$ be the real subset-incidence matrix, $F_{S,i}\coloneqq\mathbf1_{\{i\in S\}}$, so $(FF^{\mathsf T})_{S,T}=|S\cap T|$. For $k\ge2$, counting subsets gives
\begin{equation}\label{eq:incidence-eigenvalue}
 F^{\mathsf T}F=2^{k-2}(I_k+\ket{\boldsymbol 1}\bra{\boldsymbol 1}),
 \qquad \lambda_{\max}(FF^{\mathsf T})=(k+1)2^{k-2},
\end{equation}
where $\ket{\boldsymbol 1}\coloneqq\sum_{i=1}^k\ket i$ is the unnormalized all-ones vector. For $k=1$ the same largest-eigenvalue value is one, verified directly.

For $\ket x\coloneqq\sum_S x_S\ket S$, let $\ket{x_{\mathrm{abs}}}\coloneqq\sum_S|x_S|\ket S$. This replacement preserves its norm. Entrywise nonnegativity and \eqref{eq:incidence-entry-bound} yield
\begin{align}
\langle x|\mathsf B_k|x\rangle
 &\le \langle x_{\mathrm{abs}}|\mathsf B_k|x_{\mathrm{abs}}\rangle\label{eq:coefficient-explicit-proof}\\
&\le\frac{1-\delta}{\delta^{k-1}}
       \langle x_{\mathrm{abs}}|FF^{\mathsf T}|x_{\mathrm{abs}}\rangle\label{eq:line-A-32-2}\\
&\le\frac{(k+1)2^{k-2}(1-\delta)}{\delta^{k-1}}\norm{\ket x}^2.\label{eq:line-A-32-3}
\end{align}
Taking the maximum over unit vectors proves the first estimate in \eqref{eq:sharp-simple-coefficient}. Also,
$\delta^{k-1}(\mathsf B_k)_{S,T}\to|S\cap T|$ as $\delta\downarrow0$. Since $k$ is fixed, entrywise convergence is norm convergence and eigenvalues are continuous. Consequently,
\begin{equation}\label{eq:sharp-coefficient-asymptotic}
 \lim_{\delta\downarrow0}\delta^{k-1}\kappa_k(\delta)
  =(k+1)2^{k-2}.
\end{equation}
The leading small-$\delta$ constant in the explicit bound is therefore sharp.

For two senders, in the order $1,2,12$, the inverse formula reduces to
\begin{equation}\label{eq:two-sharp-matrix}
 \mathsf B_2=\frac{1-\delta}{\delta}
 \begin{pmatrix}
  1&0&1-\delta\\
  0&1&1-\delta\\
  1-\delta&1-\delta&(1-\delta)(2-\delta)
 \end{pmatrix},
 \qquad \kappa_2(\delta)\le\frac{3(1-\delta)}{\delta}.
\end{equation}
Thus the quantitative refinement already saves a power of $\delta$ in the first nontrivial case. Figure~\ref{fig:bound-comparison} illustrates the improved guarantees; it compares bounds for the same decoder, not measured errors of different receivers.

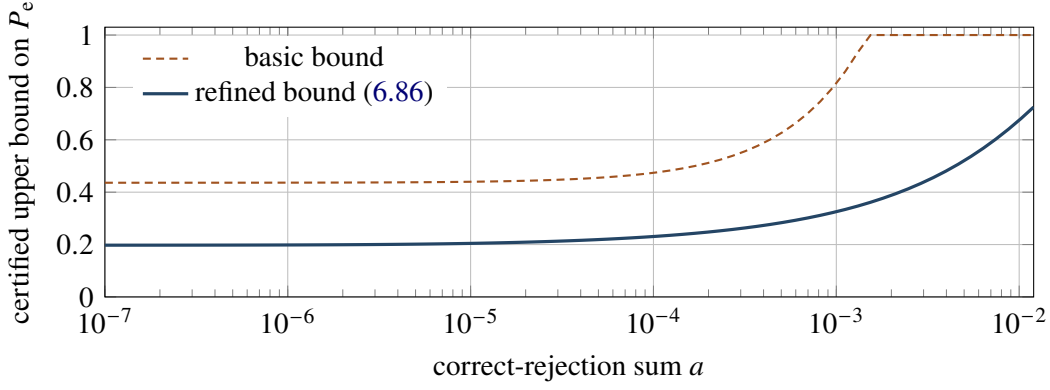
\begin{figure}[!t]
\centering
\pgfplotstableread{
a basic refined
1e-07 0.436027868 0.19759928
1.17121759e-07 0.43603437 0.197616345
1.37175063e-07 0.436041985 0.197636285
1.60661846e-07 0.436050904 0.197659574
1.8816998e-07 0.436061349 0.197686763
2.20387989e-07 0.436073584 0.197718488
2.58122288e-07 0.436087913 0.197755484
3.02317363e-07 0.436104695 0.197798596
3.54079412e-07 0.436124351 0.197848795
4.14704034e-07 0.436147373 0.197907192
4.85708657e-07 0.436174336 0.197975058
5.6887052e-07 0.436205915 0.198053832
6.66271157e-07 0.436242902 0.198145148
7.80348495e-07 0.436286221 0.198250843
9.1395788e-07 0.436336958 0.198372979
1.07044354e-06 0.436396381 0.198513852
1.2537223e-06 0.436465979 0.198676007
1.4683816e-06 0.436547493 0.198862246
1.71979436e-06 0.436642963 0.199075636
2.01425339e-06 0.43675478 0.199319509
2.35912899e-06 0.436885742 0.199597466
2.76305336e-06 0.437039127 0.199913374
3.23613669e-06 0.437218774 0.200271362
3.7902202e-06 0.43742918 0.200675822
4.43917254e-06 0.437675611 0.201131407
5.19923695e-06 0.437964235 0.201643037
6.08943774e-06 0.438302277 0.202215906
7.13205657e-06 0.438698198 0.202855506
8.35319007e-06 0.439161907 0.203567644
9.7834031e-06 0.439705011 0.204358485
1.14584937e-05 0.440341104 0.20523459
1.34203894e-05 0.441086108 0.206202969
1.5718196e-05 0.44195867 0.207271143
1.84094276e-05 0.442980629 0.208447214
2.15614453e-05 0.444177566 0.209739937
2.52531439e-05 0.445579439 0.211158808
2.95769263e-05 0.447221338 0.212714154
3.46410162e-05 0.449144358 0.214417231
4.05721673e-05 0.451396634 0.216280336
4.75188358e-05 0.454034538 0.218316919
5.56548961e-05 0.457124099 0.220541716
6.5183993e-05 0.460742646 0.222970887
7.63446389e-05 0.464980752 0.225622169
8.94161836e-05 0.469944497 0.228515048
0.000104725807 0.475758122 0.231670942
0.000122656706 0.482567142 0.235113412
0.000143657691 0.490541986 0.238868386
0.000168254414 0.499882264 0.242964413
0.000197062529 0.510821761 0.247432945
0.000230803099 0.523634292 0.252308643
0.000270320648 0.538640554 0.257629724
0.000316604297 0.556216152 0.263438335
0.00037081252 0.576801002 0.26978097
0.000434302144 0.60091034 0.276708927
0.000508662309 0.62914762 0.284278803
0.000595754241 0.66221962 0.292553036
0.000697757843 0.700954127 0.301600488
0.000817226256 0.746320663 0.311497058
0.000957149762 0.799454748 0.322326335
0.00112103063 0.861686322 0.334180269
0.00131297079 0.934573036 0.347159834
0.00153777448 1 0.361375676
0.00180106851 1 0.376948675
0.00210944311 1 0.394010387
0.00247061687 1 0.41270325
0.00289362992 1 0.433180438
0.00338907025 1 0.455605159
0.00396933867 1 0.480149131
0.00464895925 1 0.506989815
0.00544494283 1 0.536305803
0.00637721279 1 0.568269451
0.00746910376 1 0.603035354
0.00874794567 1 0.640722453
0.0102457478 1 0.681386177
0.012 1 0.724974519
}\qsdBoundsData
\begin{tikzpicture}
\begin{axis}[width=.84\linewidth,height=5.15cm,xmode=log,
 xmin=1e-7,xmax=.012,ymin=0,ymax=1.03,
 xlabel={correct-rejection sum $a$},ylabel={certified upper bound on $P_{\mathrm e}$},
 legend style={at={(.03,.97)},anchor=north west,draw=none,font=\small},
 tick label style={font=\small},label style={font=\small},grid=major]
 \addplot[inkorange,thick,densely dashed] table[x=a,y=basic] {\qsdBoundsData};
 \addlegendentry{basic bound}
 \addplot[inkblue,very thick] table[x=a,y=refined] {\qsdBoundsData};
 \addlegendentry{refined bound \eqref{eq:sharp-main}}
\end{axis}
\end{tikzpicture}
\caption{Certified bounds for the same decoder at $k=2$, $\delta=0.1$, and $b=10^{-4}$. The basic curve is $\min\{1,d_2+C_2a+b\}$; the refined curve is \eqref{eq:sharp-main} with the exact $\kappa_2$. The testing sums are hypothetical, not channel data. Neither curve optimizes $\delta$. Error-pattern weights and Cheng's correction permit further improvements.}
\label{fig:bound-comparison}
\end{figure}

\paragraph{Optimality of the geometric coefficient for these labels.}
We finally show that $\kappa_k(\delta)$ cannot be reduced in \eqref{eq:sharp-correct-rejection}--\eqref{eq:line-A-5-2} if that inequality is to hold for every output state and every collection of rejection projectors with the specified labels. Fix $\delta$ and $k$. Let $G$ be the normalized error-label Gram matrix and $D_\delta$ be as in \eqref{eq:gram-congruence}. Then
$\mathsf L_k^{-1}=D_\delta G^{-1}D_\delta$ is a real symmetric matrix. Choose a real unit vector $\ket\xi\coloneqq\sum_S\xi_S\ket S$ for its largest eigenvalue $\kappa_k(\delta)$. On a qubit output, define, for $\epsilon>0$,
\begin{equation}\label{eq:geometric-sharpness-family}
 \rho\coloneqq\proj0,\qquad
 \ket{q_S(\epsilon)}\coloneqq
  \frac{\ket1+\epsilon\xi_S\ket0}{\sqrt{1+\epsilon^2\xi_S^2}},
 \qquad Q_S(\epsilon)\coloneqq\proj{q_S(\epsilon)}.
\end{equation}
The sum of individual rejection probabilities satisfies
\begin{equation}\label{eq:geometric-sharpness-denominator}
 \sum_S\Tr[Q_S(\epsilon)\rho]
 =\sum_S\frac{\epsilon^2\xi_S^2}{1+\epsilon^2\xi_S^2}
 =\epsilon^2+O(\epsilon^4).
\end{equation}
Let $R_\epsilon$ be the matrix whose columns are
$\ket{q_S(\epsilon)}\otimes\ket{w_S}$, and let
$G_\epsilon\coloneqq R_\epsilon^\dagger R_\epsilon$.
As $\epsilon\to0$, $G_\epsilon=G+O(\epsilon^2)$ and is invertible for sufficiently small $\epsilon$. Its column-space projector is
$\Pi_\epsilon=R_\epsilon G_\epsilon^{-1}R_\epsilon^\dagger$;
self-adjointness, idempotence, and action as the identity on every column verify this formula directly. The vector of overlaps with $\ket0\ket t$ is
\begin{equation}\label{eq:geometric-sharpness-overlap}
 \ket{h_\epsilon}\coloneqq R_\epsilon^\dagger(\ket0\ket t)
  =\epsilon D_\delta\ket\xi+O(\epsilon^3).
\end{equation}
Here all $O(\cdot)$ estimates are for fixed $k,\delta$ in a fixed finite-dimensional space. Continuity of the inverse at $G$ gives
\begin{align}
\langle0,t|\Pi_\epsilon|0,t\rangle
 &=\langle h_\epsilon|G_\epsilon^{-1}|h_\epsilon\rangle\label{eq:geometric-sharpness-numerator}\\
&=\epsilon^2\langle\xi|D_\delta G^{-1}D_\delta|\xi\rangle+O(\epsilon^4)\label{eq:line-A-38-2}\\
&=\epsilon^2\kappa_k(\delta)+O(\epsilon^4).\label{eq:line-A-38-3}
\end{align}
Dividing by \eqref{eq:geometric-sharpness-denominator} shows that the ratio approaches $\kappa_k(\delta)$. Even pure qubit states and rank-one rejection projectors require this coefficient. This proves optimality for that particular geometric inequality and label family.

The transfer bound is separately optimal under the information specified in \lemref{lem:sharp-transfer}. These component statements do not show simultaneous saturation of all coding inequalities. We assert neither globally optimal decoding nor optimal reliability or second-order rates. Both bounds certify the same decoder.

\end{document}